\documentclass{article}

\usepackage{arxiv}

\usepackage[utf8]{inputenc} 
\usepackage[T1]{fontenc}    
\usepackage{hyperref}       
\usepackage{url}            
\usepackage{booktabs}       
\usepackage{array}
\usepackage{amsfonts}       
\usepackage{nicefrac}       
\usepackage{graphicx}
\usepackage{subcaption}
\usepackage{multirow}
\usepackage{adjustbox}
\usepackage[numbers]{natbib}
\usepackage{doi}
\usepackage{amsmath}
\usepackage{amssymb}
\usepackage{algorithm}
\usepackage{algpseudocode}
\usepackage{xcolor}
\usepackage{verbatim}
\usepackage{fancyhdr}
\usepackage{amsthm}
\usepackage{cleveref}
\usepackage{bm}
\usepackage[mathscr]{euscript}
\usepackage{textcomp}

\usepackage{mathrsfs}

\DeclareMathAlphabet{\pazocal}{OMS}{zplm}{m}{n}

\usepackage[normalem]{ulem}

\usepackage{lscape}		
\usepackage{fancyvrb}

\newtheoremstyle{remarkstyle}  
  {5pt}                        
  {5pt}                        
  {}                           
  {}                           
  {\bfseries}                  
  {.}                          
  { }                          
  {}                           
\newtheorem{theorem}{Theorem}[section]

\newtheorem{lemma}[theorem]{Lemma}

\theoremstyle{remarkstyle}
\newtheorem{remark}[theorem]{Remark}

\crefname{figure}{Fig.}{Figs.}
\Crefname{figure}{Figure}{Figures}
\crefname{algorithm}{Algorithm}{Algorithms}
\crefname{equation}{Eq.}{Eqs.}
\crefname{section}{Appendix}{Appendices}
\Crefname{section}{Appendix}{Appendices}

\definecolor{darkorange}{RGB}{0,0,0}

\title{
Constitutive Priors for Inverse Design
}

\author{
Jinkyo Han$^{1}$, 
Bahador Bahmani$^{1,2}$\thanks{Corresponding author: \texttt{bahador.bahmani@northwestern.edu}} \\
\\
$^{1}$Department of Mechanical Engineering, Northwestern University, Evanston, IL 60208, USA \\
$^{2}$Theoretical and Applied Mechanics, Northwestern University, Evanston, IL 60208, USA \\
}

\date{}

\renewcommand{\shorttitle}{Constitutive Prior}

\hypersetup{
    colorlinks=true,      
    linkcolor=green,      
    urlcolor=blue,        
    citecolor=blue        
}

\begin{document}
\maketitle

\begin{abstract}
With recent advances in material synthesis and additive manufacturing, material systems can be designed to achieve prescribed mechanical responses. 
An important class of such problems is the inverse design of elastic networks that attain a target configuration under loading, with applications in robotics, aerospace, and shape-morphing structures. 
This work presents a framework that formulates inverse design directly within the learned family of constitutive behaviors.
Given a collection of stress--strain responses, we construct a constitutive prior, defined as a low-dimensional latent representation of admissible material laws learned directly from these responses. Spatially varying latent variables are then optimized subject to the governing equilibrium equations so that the deformed network matches a target configuration. The constitutive prior is represented using an energy-based, partially input-convex neural network that enforces the constitutive constraints by construction. To improve robustness in the resulting nonconvex optimization problem, the framework combines homotopy-continuation-based optimization with correspondence-free point cloud matching, allowing the target and optimized geometries to have different discretizations. The proposed approach is demonstrated on several inverse design problems for nonlinear elastic networks, and quantitative comparisons with alternative optimization strategies show improved robustness and optimization performance.
\end{abstract}


\section{Introduction}

Recent advances in additive manufacturing, architected materials, and programmable mechanical systems have enabled structures whose mechanical functionality can be tailored through geometry, topology, and material composition
\cite{steeves2007concepts,bertoldi2017flexible,zheng2014ultralight}.
Much of the existing work in inverse mechanical design has therefore focused on identifying geometric or topological configurations that achieve prescribed structural responses
\cite{clausen2015topology,rafsanjani2018kirigami,da2020topology}.
Despite their success, such approaches may lead to intricate structural layouts that are difficult to manufacture, particularly when multiscale features or highly complex topologies are required
\cite{garaigordobil2019overhang,brackett2011topology,guo2017self,li2018topology}.
At the same time, developments in stimuli-responsive, multimaterial, and functionally graded systems have demonstrated that complex structural behavior can also arise from spatially varying material responses while retaining comparatively simple geometries
\cite{kuang2019advances,boley2019shape,kim2018printing,wang2023inverse}.
These developments motivate inverse-design formulations in which the unknowns are not only geometric, but may also describe spatially distributed constitutive behavior coupled through nonlinear equilibrium constraints.

Material-aware inverse design methods commonly prescribe a constitutive family and optimize its associated material parameters.
Related inverse homogenization and topology-optimization approaches seek material architectures with prescribed effective properties~\cite{sigmund1994materials,sigmund1995tailoring,osanov2016topology}.
More recently, data-driven constitutive models have enabled mechanical responses to be learned directly from experimental or simulation data
\cite{as2022mechanics,klein2022polyconvex,linden2023neural,bahmani2024physics,bahmani2026conformal}.
These developments include methods that identify or represent individual classes of constitutive behavior,
as well as models that map material or microstructural descriptors to effective constitutive responses
\cite{vlassis2020geometric,yang2018deep}.
Together with advances in differentiable programming and differentiable mechanical solvers
\cite{farrell2013automated,schoenholz2020jax,xue2023jax,pundir2026versatile},
such constitutive surrogates have been incorporated into multiscale inverse problems and design optimization
\cite{xia2015multiscale,white2019multiscale,zheng2021data,vijayakumaran2025consistent,stollberg2025multiscale}.

These approaches typically assume that either a constitutive parameterization or a descriptor associated with each material response is available.
The present work considers a  setting in which a collection of constitutive datasets is available, but the material, chemical, microstructural, or manufacturing descriptors responsible for the differences among those datasets are unknown or unavailable.

To address this setting, we introduce a \emph{constitutive prior}: a latent representation of a family of constitutive responses learned jointly from the available stress--strain data.
A shared energy-based model captures structure common to the observed responses, while latent coordinates distinguish individual members of the learned family.
After training, the shared model is fixed, and spatially distributed latent coordinates are used as the inverse design variables.
Thus, the optimization remains finite-dimensional, and the latent variables select constitutive responses from a learned family, whereas these variables may not necessarily be physically meaningful.

The proposed formulation is not intended to replace classical parameterized constitutive models when an appropriate analytical form and physically meaningful material parameters are available.
Rather, it targets settings in which constitutive responses are observed without paired material, microstructural, chemical, or manufacturing descriptors, so that neither a prescribed parameterization nor a descriptor-to-response mapping is available.
The learned family restricts the inverse design problem to responses generated by a model fitted to the available data and satisfying the constitutive conditions encoded in its representation.

The term \emph{prior} reflects the role of this learned family in the inverse problem.
Once trained, it is fixed before inverse design and restricts the constitutive responses that may be assigned throughout the structure.
This perspective allows candidate constitutive responses to be inferred from a desired global mechanical state without first requiring their underlying material descriptors to be observed.
Realizing an inferred response through a particular composition, microstructure, or manufacturing process may then be addressed as a separate downstream problem.
Such a realizability assessment lies beyond the scope of the present work.
The overall perspective is illustrated in \cref{fig:motivation}.

\begin{figure}[htbp]
    \centering
    \begin{subfigure}{0.75\textwidth}
        \centering
        \includegraphics[width=\textwidth]{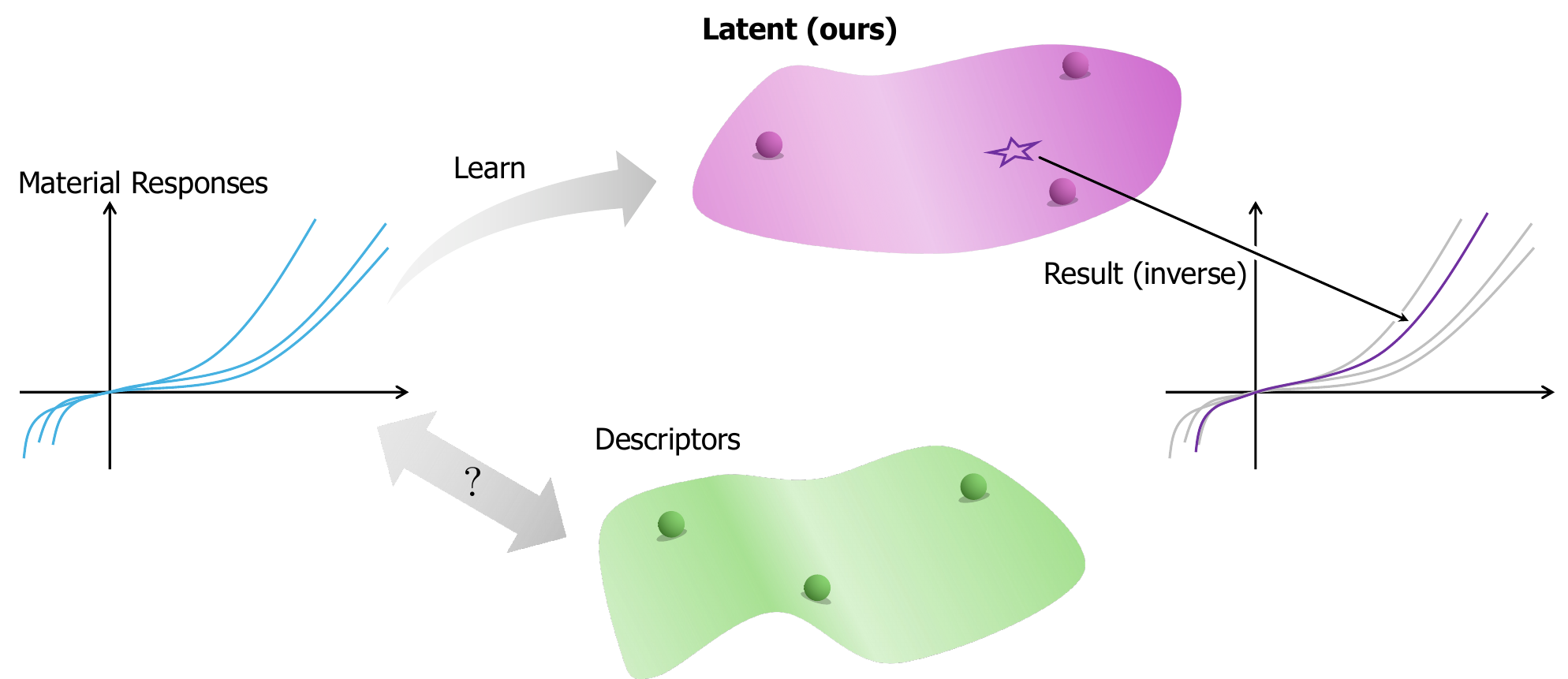}
    \end{subfigure}
    \caption{
    Illustration of the perspective underlying the constitutive prior.
    A latent family of constitutive responses is learned from observed material-response data without requiring the corresponding material or microstructural descriptors.
    Inverse design then selects responses from this learned family, while their physical realization constitutes a separate downstream problem.
    }
    \label{fig:motivation}
\end{figure}

The constitutive prior is inspired by DeepSDF~\cite{park2019deepsdf}, in which shared network parameters and object-specific latent coordinates represent a continuous family of shapes, and by the deep image prior~\cite{ulyanov2018deep}, which uses the structure of a neural network itself to restrict admissible solutions in an inverse problem. Here, the represented objects are constitutive responses, and the learned representation is coupled to the governing mechanical equations and constitutive admissibility conditions.

We investigate the proposed framework in nonlinear elastic networks designed to attain prescribed deformed configurations, with applications in shape-morphing structures and soft robotics
\cite{chen2026physical,wang2025autonomous,kotikian2024liquid}.
The constitutive prior is represented using an energy-based partially input-convex neural network (PICNN) \cite{amos2017input},
allowing continuous interpolation within the learned constitutive family while enforcing the constitutive conditions adopted for the axial members.

A further challenge arises when the target and simulated configurations do not share a common discretization or pointwise correspondence. We therefore represent both configurations as point clouds and measure their discrepancy using a correspondence-free geometric objective~\cite{fan2017point,dubuisson1994modified}. This formulation permits the target and structural configurations to contain different numbers of points or sampling resolutions without requiring nodal correspondence or interpolation onto a common discretization, in contrast to the formulations in~\cite{wang2025autonomous,kotikian2024liquid}, which assume corresponding discretizations or construct a common representation for shape comparison.

The resulting inverse-design problem is strongly nonlinear and nonconvex because the spatially distributed constitutive responses, equilibrium configuration, and geometric objective are coupled through the governing equilibrium equations.
To improve robustness, we employ a homotopy-continuation strategy~\cite{allgower2003introduction,wayburn1987homotopy} that progressively deforms an initial target point cloud toward the desired target.
When only the final target is available, affine registration~\cite{du2010affine,feldmar1996rigid} is used to initialize the continuation path.
We also investigate continuous neural parameterizations of the spatial latent field and use a graph-based metric~\cite{zhou2003learning} to quantify the smoothness of the inferred material distribution.

The main contributions of this work are:
\begin{enumerate}
    \item a formulation of inverse design over a latent family of constitutive responses learned jointly from multiple constitutive datasets without requiring known material or microstructural descriptors;
    
    \item an adjoint-based computational framework that couples spatially varying latent constitutive representations with nonlinear mechanical equilibrium;
    
    \item a correspondence-free point cloud objective combined with homotopy continuation and affine registration for inverse design toward geometrically mismatched targets; and
    
    \item a graph-based metric for quantifying the spatial smoothness of the resulting constitutive design.
\end{enumerate}

The remainder of this paper is organized as follows.
In Section~\ref{sec:formulations}, we present the constitutive-prior inverse-design formulation, its specialization to elastic networks, the geometric objective, sensitivity analysis, continuation strategy, and spatial-regularity metric.
In Section~\ref{sec:result}, we demonstrate the framework through representative inverse-identification and shape-morphing problems. 
Finally, Section~\ref{sec:conclusion} summarizes the main observations, limitations, and directions for future work.

\section{Constitutive Inference Framework}
\label{sec:formulations}

This section introduces the proposed constitutive-prior inverse design framework, in which constitutive responses themselves are treated as the optimization object rather than parameters within a prescribed constitutive model. Section~\ref{sec:formulations:inverse} formulates the general PDE-constrained constitutive inference problem considered in this work, while Section~\ref{sec:formulations:constitutive} introduces the constitutive prior together with its data-driven construction. Section~\ref{sec:formulations:network} then specializes the formulation to elastic networks.

The remaining sections introduce the computational components enabling robust solution of the resulting nonlinear inverse problem. Sections~\ref{sec:formulations:objective}--\ref{sec:formulations:homotopy} define the geometric objective, derive the adjoint formulation, and introduce the homotopy-continuation strategy. Finally, Section~\ref{sec:formulations:smoothness} introduces the graph-based metric used to quantify spatial smoothness of inferred material distributions.

\subsection{Design Problem Formulation}
\label{sec:formulations:inverse}

We consider the inverse design of mechanical systems whose response is governed by nonlinear physics and material behavior. The objective is to determine design variables such that the resulting system configuration matches a prescribed target under given loading conditions.

A challenge is that the appropriate constitutive family and the parameterization of the observed responses may not be known \textit{a priori}. 
Rather than prescribing an analytical constitutive form and optimizing its parameters, we consider a setting in which the admissible family of constitutive responses and its low-dimensional parameterization are learned from observed constitutive data.
To represent this family, we introduce a \emph{constitutive prior}, i.e., a latent-parameterized family of admissible constitutive relations.
This representation restricts the design space to material responses supported by the available data and enables interpolation within the learned family.

Within this framework, the design variables correspond to latent parameters that select material responses from 
the learned constitutive family,
thereby transforming the inverse design problem into an optimization over latent coordinates that parameterize admissible constitutive functions.
For heterogeneous systems, these latent variables are assigned at the level of structural components, enabling spatially varying material behavior.
The inverse design problem is thus posed as a PDE-constrained optimization problem over the latent variables, where the governing equations enforce physical consistency and the objective measures discrepancy between the system configuration and a target configuration.
The general formulation of this latent-space inverse design problem is introduced next.

\begin{remark}
    By a configuration, we refer here to the geometric shape occupied by the structure in physical space, represented as a subset of $\mathbb{R}^d$, with $d\in{2,3}$. The system configuration is the deformed shape obtained under the applied loading, whereas the target configuration is the prescribed shape to be matched.
\end{remark}

Let $\boldsymbol{\theta}\in\mathbb{R}^{n_\theta}$ denote the vector of design variables, 
whose entries are latent coordinates that select material responses from the learned constitutive prior. These variables are defined more explicitly in the following sections.
Given a prescribed loading trajectory $\boldsymbol{\eta}(t)$ for $t \in [0,t_f]$, the resulting system configuration is obtained as the solution of a nonlinear physics-constrained problem, which we denote by the state $\boldsymbol{u}(t;\boldsymbol{\theta})$.

The target configuration trajectory is denoted by $\mathcal{T}(t)$. Since the system and target configurations may not share a common discretization or pointwise correspondence, we represent them as empirical probability measures supported on point sets in physical space.
This representation accommodates different numbers of points, is invariant to their ordering, and provides a common notation for defining geometric discrepancies as expectations over the two configurations~\citep{peyre2019computational}.

For each time $t$, we define the empirical probability measures
\begin{equation}
\mu_{\boldsymbol{\theta}}(t)
=
\frac{1}{N_{\mathrm{sys}}}
\sum_{i=1}^{N_{\mathrm{sys}}}
\delta_{\boldsymbol{x}_i(\boldsymbol{\theta},t)},
\qquad
\nu(t)
=
\frac{1}{N_{\mathrm{targ}}}
\sum_{k=1}^{N_{\mathrm{targ}}}
\delta_{\boldsymbol{y}_k(t)},
\end{equation}
where $\delta_{(\cdot)}$ is a Dirac measure centered at $(\cdot)$, 
$\{\boldsymbol{x}_i(\boldsymbol{\theta},t)\}_{i=1}^{N_{\mathrm{sys}}}$ are sampled points from the system configuration and $\{\boldsymbol{y}_k(t)\}_{k=1}^{N_{\mathrm{targ}}}$ are sampled points from the target configuration.
Each measure assigns uniform weight to the points in the corresponding configuration, so that expectations with respect to these measures reduce to averages over the associated point sets, as used below in the definition of the Chamfer discrepancy.
The points defining each measure need not correspond between configurations, e.g., a common mesh.

We define the trajectory matching objective using a discrepancy functional $\rho$ between probability measures:
\begin{equation}
\mathcal{J}(\boldsymbol{\theta})
=
\int_0^{t_f}
\rho
\left(
\mu_{\boldsymbol{\theta}}(t),
\nu(t)
\right)
\, dt
+
\mathcal{R}(\boldsymbol{\theta}),
\end{equation}
and $\mathcal{R}$ is a regularization function. In general, different choices of the discrepancy $\rho$ are possible depending on the application and the available data. In this work, we specialize $\rho$ to a point-cloud-based discrepancy, namely the Chamfer distance, as it naturally accommodates mismatched discretizations and does not require pointwise correspondence between configurations.

The inverse design problem is then formulated as
\begin{equation}
\begin{aligned}
\boldsymbol{\theta}^{\star}
=
\underset{\boldsymbol{\theta}}{\mathrm{argmin}}
\quad &
\mathcal{J}(\boldsymbol{u})
\\
\text{subject to}
\quad &
\mathcal{G}
\left(
\boldsymbol{u}(t),
\boldsymbol{\theta},
\boldsymbol{\eta}(t)
\right)
=
\boldsymbol{0},
\qquad t \in [0,t_f],
\\
&
\mathcal{B}
\left(
\boldsymbol{u}(t),
\boldsymbol{\eta}(t)
\right)
=
\boldsymbol{0}.
\end{aligned}
\label{eq:gen-optim}
\end{equation}

The dependence on the design variables $\boldsymbol{\theta}$ is implicit through the equilibrium solution $\boldsymbol{u}(\boldsymbol{\theta})$.

Here, $\boldsymbol{u}(t)$ denotes the system state (e.g., displacement field), while $\mathcal{G}=\boldsymbol{0}$ represents the governing equations arising from physical principles, and $\mathcal{B}=\boldsymbol{0}$ encodes boundary and loading conditions.

\subsection{Constitutive Prior Representation}
\label{sec:formulations:constitutive}

We first formulate the ingredients of the constitutive prior for a general
three-dimensional isotropic hyperelastic continuum and then specialize the
construction to elastic networks composed of axial rods.

Suppose that $M$ datasets are obtained from nominally identical experimental protocols applied to different material specimens.
The dataset associated with specimen $m$ is written as
\begin{equation}
\label{eq:data}
    \mathcal{D}^{(m)}
    =
    \{
        (
            \boldsymbol{F}^{(m,i)},
            \boldsymbol{P}^{(m,i)}
        )
    \}_{i=1}^{N^{(m)}},
    \qquad
    m=1,\ldots,M,
\end{equation}
where $\boldsymbol{F}^{(m,i)}$ and
$\boldsymbol{P}^{(m,i)}$ denote the measured deformation gradient and the
corresponding first Piola--Kirchhoff stress at observation $i$ for specimen
$m$, respectively. 
We do not require the datasets to contain the same number
of observations, nor do we require the sampled deformation states to coincide
across specimens. 
Thus, both $N^{(m)}$ and the values of
$\boldsymbol{F}^{(m,i)}$ may vary among the datasets.

We restrict attention to isotropic hyperelastic materials.
Although the specimen-level constitutive responses may differ, we assume that they belong to a common material family and share an underlying constitutive structure.
These differences are attributed to unobserved specimen-specific variations, such as variations in microstructure or manufacturing history. 
We further hypothesize that this variability admits a low-dimensional latent representation, such that the associated family of constitutive responses can be parameterized by a shared nonlinear mapping.
Motivated by implicit neural priors for images and joint latent representations for geometric shapes \citep{ulyanov2018deep,park2019deepsdf}, we adopt an analogous latent parameterization in which the represented objects are constitutive responses.

We represent the strain-energy density by a neural mapping
$\mathcal{M}_{\boldsymbol{\phi}}$ with shared learnable parameters
$\boldsymbol{\phi}$. Its inputs are a deformation measure derived from
$\boldsymbol{F}$ and a material-specific latent code
$\boldsymbol{z}\in\mathcal{Z}\subset\mathbb{R}^{n_z}$.
To enforce material-frame indifference and isotropy, the deformation input is chosen as the vector of principal invariants of right stretch tensor $\boldsymbol{U}$, denoted by $\boldsymbol{i}(\boldsymbol{F})=(i_1,i_2,i_3)$ with $i_1=\mathrm{tr}\,\boldsymbol{U}$, $i_2=\mathrm{tr}(\mathrm{cof}\,\boldsymbol{U})$, and $i_3=\det\boldsymbol{U}=J$~\cite{steigmann2003isotropic,ball1976convexity}.
The strain-energy
density is therefore written as $\mathcal{M}_{\boldsymbol{\phi}}
(
    \boldsymbol{i}(\boldsymbol{F});
    \boldsymbol{z}
)$.
The associated first Piola--Kirchhoff stress is obtained by differentiation as follows,
\begin{equation}
    \boldsymbol{P}_{\boldsymbol{\phi}}
    \left(
        \boldsymbol{F};
        \boldsymbol{z}
    \right)
    =
    \partial_{\boldsymbol{F}}
    \mathcal{M}_{\boldsymbol{\phi}}
    \left(
        \boldsymbol{i}(\boldsymbol{F});
        \boldsymbol{z}
    \right).
    \label{eq:constitutive_stress_mapping}
\end{equation}

To enforce an energy-free and stress-free reference configuration, we construct $\mathcal{M}_{\boldsymbol{\phi}}$ from an auxiliary scalar-valued mapping $g_{\boldsymbol{\phi}}$ as
\begin{align}
    \mathcal{M}_{\boldsymbol{\phi}}
    \left(
        \boldsymbol{i};
        \boldsymbol{z}
    \right)
    ={}&
    g_{\boldsymbol{\phi}}
    \left(
        \boldsymbol{i};
        \boldsymbol{z}
    \right)
    -
    g_{\boldsymbol{\phi}}
    \left(
        \boldsymbol{i}_0;
        \boldsymbol{z}
    \right)
    -
    c(\boldsymbol{z})
    (i_3-1)
    \label{eq:stress_free}
    ,\\
    c(\boldsymbol{z})
    ={}&
    \left.
    \left(
    \frac{
    \partial g_{\boldsymbol{\phi}}
    (\boldsymbol{i};\boldsymbol{z})
    }{
    \partial i_1
    }
    +
    2
    \frac{
    \partial g_{\boldsymbol{\phi}}
    (\boldsymbol{i};\boldsymbol{z})
    }{
    \partial i_2
    }
    +
    \frac{
    \partial g_{\boldsymbol{\phi}}
    (\boldsymbol{i};\boldsymbol{z})
    }{
    \partial i_3
    }
    \right)
    \right|_{\boldsymbol{i}=\boldsymbol{i}_0},
    \label{eq:stress_free_fact}
\end{align}
where
$\boldsymbol{i}_0
=
\boldsymbol{i}(\boldsymbol{F}_0)$
denotes the invariant representation of the reference state.
Here, the last term in \cref{eq:stress_free} combines the stress-free condition from \cite{steigmann2003isotropic} with the linear-in-$J$ correction term in \cite{linden2023neural}.
This formulation is thermodynamically consistent for hyperelastic materials by construction, since the stress is derived from a strain-energy potential and the resulting response exhibits zero mechanical dissipation.

Additional restrictions may be required to obtain the appropriate
stability and well-posedness properties for a particular constitutive
setting. In general finite-strain continuum hyperelasticity, these
requirements commonly involve coercivity and polyconvexity with respect to
the deformation gradients~\cite{ball1976convexity}.
In the three-dimensional case, polyconvexity is preserved if $g_{\boldsymbol{\phi}}$ is convex in its arguments and non-decreasing with respect to $i_1$ and $i_2$ \cite{steigmann2003isotropic}.
In a general setting, $g_{\boldsymbol{\phi}}$ may be parameterized by any function that is convex in $\boldsymbol{i}$ by construction. The required monotonicity with respect to $i_1$ and $i_2$ may then be enforced weakly by augmenting the training objective introduced below with a penalty of the form
$\max(0,-\partial g_{\boldsymbol{\phi}}/\partial i_1)
+\max(0,-\partial g_{\boldsymbol{\phi}}/\partial i_2)$. 
In Appendix~\ref{apdx:stress_free}, we show that the normalization introduced in \cref{eq:stress_free} yields a zero-energy, stress-free reference configuration and preserves convexity with respect to $\boldsymbol{i}$ whenever $g_{\boldsymbol{\phi}}$ is convex in $\boldsymbol{i}$.

Each observed material is assigned a training latent code
$\boldsymbol{z}^{(m)}\in\mathcal{Z}$. These latent codes are not prescribed
material descriptors; they are inferred jointly with the shared model
parameters $\boldsymbol{\phi}$. The shared parameters and training latent
codes are obtained by solving
\begin{equation}
\label{eq:surrogate_loss}
(
    \boldsymbol{\phi}^{\star},
    \{
        \boldsymbol{z}^{(m)\star}
    \}_{m=1}^{M}
)
=
\underset
{
    \boldsymbol{\phi},
    \{\boldsymbol{z}^{(m)}\}_{m=1}^M
}
{\text{argmin}}
\;
\mathbb{E}_m
\Big[
\mathbb{E}_{(\boldsymbol{F}, \boldsymbol{P}) \sim \mathcal{D}^{(m)}}
\big[
\|
\boldsymbol{P}_{\boldsymbol{\phi}}
(\boldsymbol{F}; \boldsymbol{z}^{(m)})
-
\boldsymbol{P}\|_F^2
\big]
+
\frac{1}{s^2}
\|\boldsymbol{z}^{(m)}\|_2^2
\Big],
\end{equation}
where $\lVert \cdot \rVert_F$ denotes the Frobenius norm, and $\lVert \cdot \rVert_2$ denotes the Euclidean norm.
The first term measures the discrepancy between the predicted and observed
stress responses.
The second term is a quadratic latent-code regularizer,
equivalent, up to scaling and an additive constant, to the negative
log-density of a zero-mean isotropic Gaussian distribution.
It discourages
the individual training codes from diverging during joint optimization
\citep{park2019deepsdf}.
Smaller values of $s$ impose stronger
regularization, whereas larger values permit greater variation among the training latent codes.

\begin{remark}
The proposed formulation is itself parametrized, but differs from a conventional parametrized constitutive model in how material-to-material variability is represented. In a conventional approach, a constitutive functional form, such as a Neo-Hookean model or a neural strain-energy model, is selected in advance, and variability may be represented through specimen-dependent model parameters, possibly endowed with probability distributions. In the present formulation, the shared parameters $\boldsymbol{\phi}$ are treated as deterministic, learned jointly from all datasets, and common to all specimens. Variability is instead represented by specimen-specific latent codes $\boldsymbol{z}^{(m)}$, which enter the shared nonlinear mapping as inputs and select individual responses within the learned constitutive family. Thus, $\boldsymbol{\phi}$ captures the common constitutive family of responses, whereas the low-dimensional latent codes distinguish the specimen-specific responses. This latent-variable construction is inspired by DeepSDF \citep{park2019deepsdf} and represents one possible approach for learning a family of constitutive responses.
\end{remark}

For fixed $\boldsymbol{\phi}$, variation of the latent coordinate defines
the constitutive family
\begin{equation}
    \mathcal{C}_{\boldsymbol{\phi}}
    =
    \left\{
        \boldsymbol{F}
        \mapsto
        \mathcal{M}_{\boldsymbol{\phi}}
        \left(
            \boldsymbol{i}(\boldsymbol{F});
            \boldsymbol{z}
        \right)
        \;\middle|\;
        \boldsymbol{z}\in\mathcal{Z}
    \right\}.
    \label{eq:constitutive_prior}
\end{equation}
We refer to $\mathcal{C}_{\boldsymbol{\phi}}$ as the
\emph{constitutive prior}.
Here, \emph{prior} is used in a
non-probabilistic sense: after training, the learned family is fixed before
the inverse-design problem is solved and restricts the constitutive
responses available to the design procedure. Thus,
$\mathcal{M}_{\boldsymbol{\phi}}$ denotes the shared strain-energy mapping,
whereas $\mathcal{C}_{\boldsymbol{\phi}}$ denotes the family generated by
varying $\boldsymbol{z}$.

In this work, the general construction is applied to lattice systems composed of axial rods undergoing finite deformation. 
The kinematics of each rod are described by the positive axial stretch $\lambda>0$, introduced in Section~\ref{sec:formulations:network}.
Accordingly, we use $\lambda$  directly as the sole deformation input to the auxiliary mapping.
The normalization term becomes $c(\boldsymbol{z})
=
\left.
\partial g_{\boldsymbol{\phi}}(\lambda;\boldsymbol{z})
/
\partial\lambda
\right|_{\lambda=1}$ multiplied by $(\lambda-1)$.
The monotonicity requirements associated with $i_1$ and $i_2$ do not arise in this axial-rod setting, and the remaining constitutive condition is convexity with respect to $\lambda$.
To impose this condition exactly by construction, 
we represent $g_{\boldsymbol{\phi}}(\lambda;\boldsymbol{z})$ using a PICNN \citep{amos2017input} that is convex with respect to $\lambda$ for every fixed $\boldsymbol z$, while remaining
unrestricted with respect to $\boldsymbol z$. Since the normalization in \cref{eq:stress_free} subtracts an affine function of $\lambda$, the resulting strain-energy model $\mathcal{M}_{\boldsymbol{\phi}}$ remains convex with respect to $\lambda$. Details of the PICNN architecture are provided in Appendix~\ref{apdx:icnn}.

\begin{remark}
For other constitutive settings, the invariant representation, network
architecture, and construction used to enforce a stress-free and energy-free reference configuration must be selected consistently with the relevant kinematics,
material symmetries, and stability requirements
\citep{ball1976convexity,schroder2003invariant,
klein2022polyconvex,linden2023neural,bahmani2026conformal}.
\end{remark}

Once the training problem in \cref{eq:surrogate_loss} is solved, the shared
parameters are fixed at $\boldsymbol{\phi}^{\star}$. 
For a lattice with
element set $\mathcal{E}$, each element is assigned a latent design variable
$\boldsymbol{z}_e\in\mathcal{Z}$, and the complete design vector is
\begin{equation}
    \boldsymbol{\theta}
    =
    \operatorname{vec}
    \left(
        \{
            \boldsymbol{z}_e
        \}_{e\in\mathcal{E}}
    \right)
    \in
    \mathcal{Z}^{|\mathcal{E}|}
    \subset
    \mathbb{R}^{|\mathcal{E}|n_z}.
    \label{eq:element_latent_design_variables}
\end{equation}
Each $\boldsymbol{z}_e$ selects a local strain-energy function from
$\mathcal{C}_{\boldsymbol{\phi}^{\star}}$. The inverse-design problem
therefore optimizes the spatial distribution of constitutive responses
within the learned family while keeping the shared model parameters fixed.

\subsection{Specialization to Elastic Networks}
\label{sec:formulations:network}

We now specialize the general formulation to mechanical systems represented as elastic networks composed of axial elements \cite{Bonet_Gil_Wood_2016,de2012nonlinear,wriggers2008nonlinear}.
We consider a structure defined by a graph 
$G = (\mathcal{V}, \mathcal{E})$,
where $\mathcal{V}$ denotes the set of nodes.
Each element $e = (i,j) \in \mathcal{E}$ connects nodes
$i,j \in \mathcal{V}$.
Let $\boldsymbol{X}_i, \boldsymbol{x}_i \in \mathbb{R}^d$ denote the reference
and current positions of node $i$, respectively. The displacement is defined as $\boldsymbol{x}_i = \boldsymbol{X}_i + \boldsymbol{u}_i$.
For each element $e=(i,j)$, we define the total stretch
\begin{equation}\label{eq:stretch}
\lambda_e
=
\frac{\|\boldsymbol{x}_j - \boldsymbol{x}_i\|_2}{L_e},
\qquad
L_e = \|\boldsymbol{X}_j - \boldsymbol{X}_i\|_2.
\end{equation}
The strain energy density is defined as $\Psi_e(\lambda_e; \boldsymbol{z}_e)$,
where $\boldsymbol{z}_e$ is the latent material parameter of the learned constitutive prior.
The corresponding first Piola--Kirchhoff stress is $P_e = \partial_{\lambda_e} \Psi_e$.

The total potential energy of the system is given by
\begin{equation}
\Pi(\boldsymbol{u})
=
\sum_{e \in \mathcal{E}}
A_e L_e \, \Psi_e(\lambda_e; \boldsymbol{z}_e)
-
\sum_{i \in \mathcal{V}}
\boldsymbol{f}_i^\top \boldsymbol{u}_i,
\end{equation}
where $A_e$ is the cross-sectional area and $\boldsymbol{f}_i$ are external forces.
The equilibrium configuration is obtained as a stationary point of the energy,
\begin{equation}\label{eq:equilibrium}
\boldsymbol{R}(\boldsymbol{u};\boldsymbol{\theta})
=
\frac{\partial \Pi}{\partial \boldsymbol{u}}
=
\boldsymbol{0},
\end{equation}
which provides a realization of the governing operator $\mathcal{G}$ introduced in \cref{eq:gen-optim}, i.e., $\boldsymbol{R}(\boldsymbol{u};\boldsymbol{\theta})$ corresponds to a specific instance of $\mathcal{G}$ for the lattice system.

\subsection{Geometric Objective Formulation}
\label{sec:formulations:objective}

The objective of the inverse design problem is to match a prescribed target
configuration under given loading conditions. In the general formulation,
this is expressed through a discrepancy operator $\rho$ between the
system configuration and the target configuration.
When the configurations share a common discretization and pointwise
correspondence, a natural choice is an $\ell^2$-based discrepancy between
corresponding degrees of freedom. However, such formulations are inherently
mesh-dependent and impose artificial constraints when the objective concerns
geometric shape rather than specific nodal locations.

\begin{figure}[htbp]
    \centering
    \begin{subfigure}{0.55\textwidth}
        \centering
        \includegraphics[width=\textwidth]{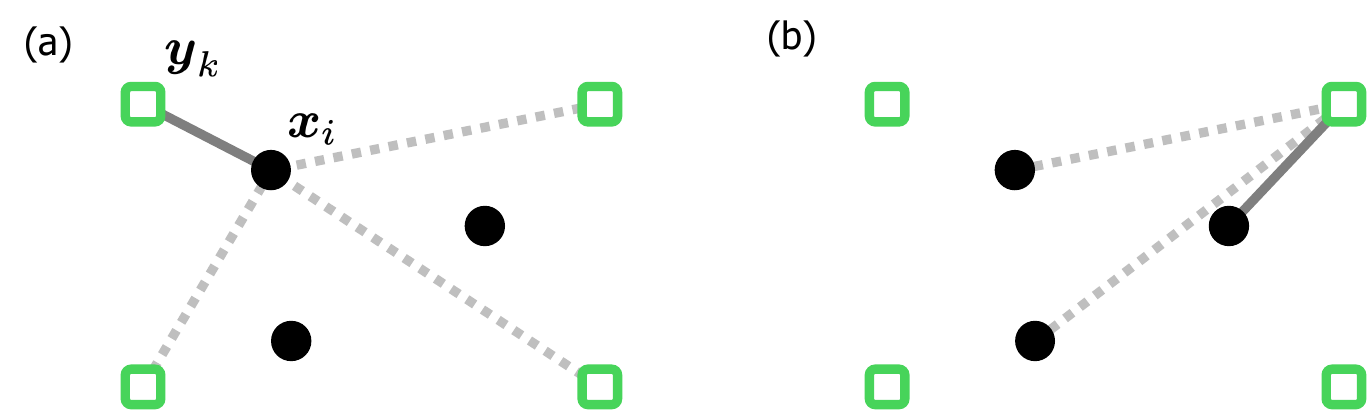}
    \end{subfigure}
    \hfill
    \caption{
        Illustration of the Chamfer distance evaluation between two point sets, depicted by black circles and green open squares.
        (a) For a selected point in one set, distances to all points in the other set are evaluated; the solid segment denotes the nearest-neighbor distance, whereas the dotted segments denote the remaining candidate distances.
        This operation is repeated for all points in that set.
        (b) The same procedure is performed in the reverse direction.
        The resulting directional collections of nearest-neighbor distances are aggregated to obtain the Chamfer distance.
    }
    \label{fig:chamfer_distance}
\end{figure}

To obtain a practical and computable realization of the discrepancy operator $\rho$, we specialize the measure-based formulation to empirical point clouds, resulting in a point-cloud-based approximation of the underlying distributional discrepancy.
Let
\begin{equation}
\mathcal{T}
=
\{\, \boldsymbol{y}_k \in \mathbb{R}^d \,\}_{k=1}^{N_{\mathrm{targ}}}
\end{equation}
denote the target point set.
Let $\tilde{\mathcal{V}} \subset \mathcal{V}$ denote the subset of system nodes whose positions are used for matching. 
Let $\mu_{\mathcal{S}}$ and $\mu_{\mathcal{T}}$ denote the empirical uniform measures over the sets 
$\mathcal{S}=\{\boldsymbol{x}_i\}_{i \in \tilde{\mathcal{V}}}$ and $\mathcal{T}$, respectively. 
We define the discrepancy using the symmetric Chamfer distance~\cite{fan2017point,dubuisson1994modified} between the current configuration and the target configuration as
\begin{equation}\label{eq:chamfer_distance}
\mathcal{J}(\boldsymbol{u};\mathcal{T}, \mathcal{S})
=
\frac12
\mathbb{E}_{\boldsymbol{x} \sim \mu_{\mathcal{S}}}
[
\min_{\boldsymbol{y}\in\mathcal{T}}
\| \boldsymbol{x} - \boldsymbol{y} \|_2^2
]
+
\frac12
\mathbb{E}_{\boldsymbol{y} \sim \mu_{\mathcal{T}}}
[
\min_{\boldsymbol{x}\in \mathcal{S}}
\| \boldsymbol{y} - \boldsymbol{x} \|_2^2
]
\;,
\end{equation}
where its graphical description is illustrated in \cref{fig:chamfer_distance}.
In practice, an expectation with respect to an empirical probability measure reduces to a finite average over its support points.
The dependence on the design variables $\boldsymbol{\theta}$ is implicit through the equilibrium configuration $\boldsymbol{u}(\boldsymbol{\theta})$.

For problems involving loading trajectories, we consider a sequence of
target geometric configurations
\begin{equation}
\mathcal{T}_{\mathrm{traj}}
=
\{\mathcal{T}^{(\tau)}\}_{\tau=1}^{N_{\mathrm{step}}},
\end{equation}
where $\mathcal{T}^{(\tau)}$ denotes the target configuration at load step $\tau=1,\dots,N_{\mathrm{step}}$.
The corresponding time-aggregated objective is defined as
\begin{equation}\label{eq:temporal_loss}
\mathcal{J}_{\mathrm{time}}
(\boldsymbol{u};\mathcal{T}_{\mathrm{traj}})
=
\frac{1}{N_{\mathrm{step}}}
\sum_{\tau=1}^{N_{\mathrm{step}}}
\mathcal{J}\bigl(
\boldsymbol{u}^{(\tau)},
\mathcal{T}^{(\tau)}
\bigr),
\end{equation}
where $\boldsymbol{u}^{(\tau)}$ denotes the equilibrium configuration under
the $\tau$-th loading condition.

It is important to note that the objective depends on the design variables
$\boldsymbol{\theta}$ only implicitly through the equilibrium configuration
$\boldsymbol{u}(\boldsymbol{\theta})$.

\subsection{Adjoint Formulation}\label{sec:formulations:adjoint}

We seek to optimize the latent design variables $\boldsymbol{\theta}
=
\{\boldsymbol{z}_e\}_{e \in \mathcal{E}}$ such that the equilibrium
configuration of the system matches the prescribed objective.
The resulting design problem is formulated as
\begin{equation}\label{eq:constrained_optim}
    \underset{\boldsymbol{\theta}}{\mathrm{argmin}}\;
    \mathcal{J}(\boldsymbol{u};\mathcal{T})
    \quad
    \text{subject to}
    \quad
    \boldsymbol{R}(\boldsymbol{u};\boldsymbol{\theta})
    =
    \boldsymbol{0},
\end{equation}
where $\boldsymbol{R}(\boldsymbol{u};\boldsymbol{\theta})$ denotes the
residual of the governing equations, and $\boldsymbol{u}$ is implicitly
defined through equilibrium. This implicit dependence of $\boldsymbol{u}$ on $\boldsymbol{\theta}$ makes direct differentiation of the objective intractable, motivating the use of an adjoint-based approach.
Following the Lagrangian viewpoint of the adjoint approach~\cite{giles2000introduction,hinze2008optimization},
to enable gradient-based optimization we introduce a Lagrange multiplier $\boldsymbol{\Lambda}$ and define the Lagrangian
\begin{equation}
    \mathscr{L}(
    \boldsymbol{u},
    \boldsymbol{\theta},
    \boldsymbol{\Lambda}
    )
    =
    \mathcal{J}(\boldsymbol{u};\mathcal{T})
    +
    \boldsymbol{\Lambda}^\top
    \boldsymbol{R}(\boldsymbol{u};\boldsymbol{\theta}).
\end{equation}

The constrained optimization problem is equivalently written as the
saddle-point problem
\begin{equation}
    \underset{\boldsymbol{u},\,\boldsymbol{\theta}}{\mathrm{argmin}}\;
    \underset{\boldsymbol{\Lambda}}{\mathrm{sup}}\;
    \mathscr{L}(\boldsymbol{u},\boldsymbol{\theta},\boldsymbol{\Lambda}).
\end{equation}
Stationarity of the Lagrangian with respect to $\boldsymbol{\Lambda}$
recovers the equilibrium constraint
\begin{equation}
    \frac{\partial\mathscr{L}}{\partial\boldsymbol{\Lambda}}
    =
    \boldsymbol{R}(\boldsymbol{u};\boldsymbol{\theta})
    =
    \boldsymbol{0},
\end{equation}
while stationarity with respect to $\boldsymbol{u}$ yields the adjoint equation
\begin{equation}
    \frac{\partial\mathscr{L}}{\partial\boldsymbol{u}}
    =
    \frac{\partial\mathcal{J}}{\partial\boldsymbol{u}}
    +
    \left(
    \frac{\partial\boldsymbol{R}}{\partial\boldsymbol{u}}
    \right)
    \boldsymbol{\Lambda}
    =
    \boldsymbol{0}.
\end{equation}

Solving for the adjoint variable gives
\begin{equation}
    \boldsymbol{\Lambda}
    =
    -\left(
    \frac{\partial\boldsymbol{R}}{\partial\boldsymbol{u}}
    \right)^{-1}
    \frac{\partial\mathcal{J}}{\partial\boldsymbol{u}}.
\end{equation}

At the stationary point with respect to $\boldsymbol{u}$ and $\boldsymbol{\Lambda}$, the gradient of the objective with respect to the design variables reduces to
\begin{equation}\label{eq:adjoint_gradient}
    \frac{\partial\mathscr{L}}{\partial\boldsymbol{\theta}}
    =
    \left(
    \frac{\partial\boldsymbol{R}}{\partial\boldsymbol{\theta}}
    \right)
    \boldsymbol{\Lambda}
    =
    -\left(
    \frac{\partial\boldsymbol{R}}{\partial\boldsymbol{\theta}}
    \right)
    \left(
    \frac{\partial\boldsymbol{R}}{\partial\boldsymbol{u}}
    \right)^{-1}
    \frac{\partial\mathcal{J}}{\partial\boldsymbol{u}}.
\end{equation}

This expression enables efficient gradient-based updates of the latent
design variables $\boldsymbol{\theta}$ without explicitly differentiating
through the solution of the nonlinear equilibrium problem.

\begin{figure}[htbp]
    \centering
    \begin{subfigure}{0.7\textwidth}
        \centering
        \includegraphics[width=\textwidth]{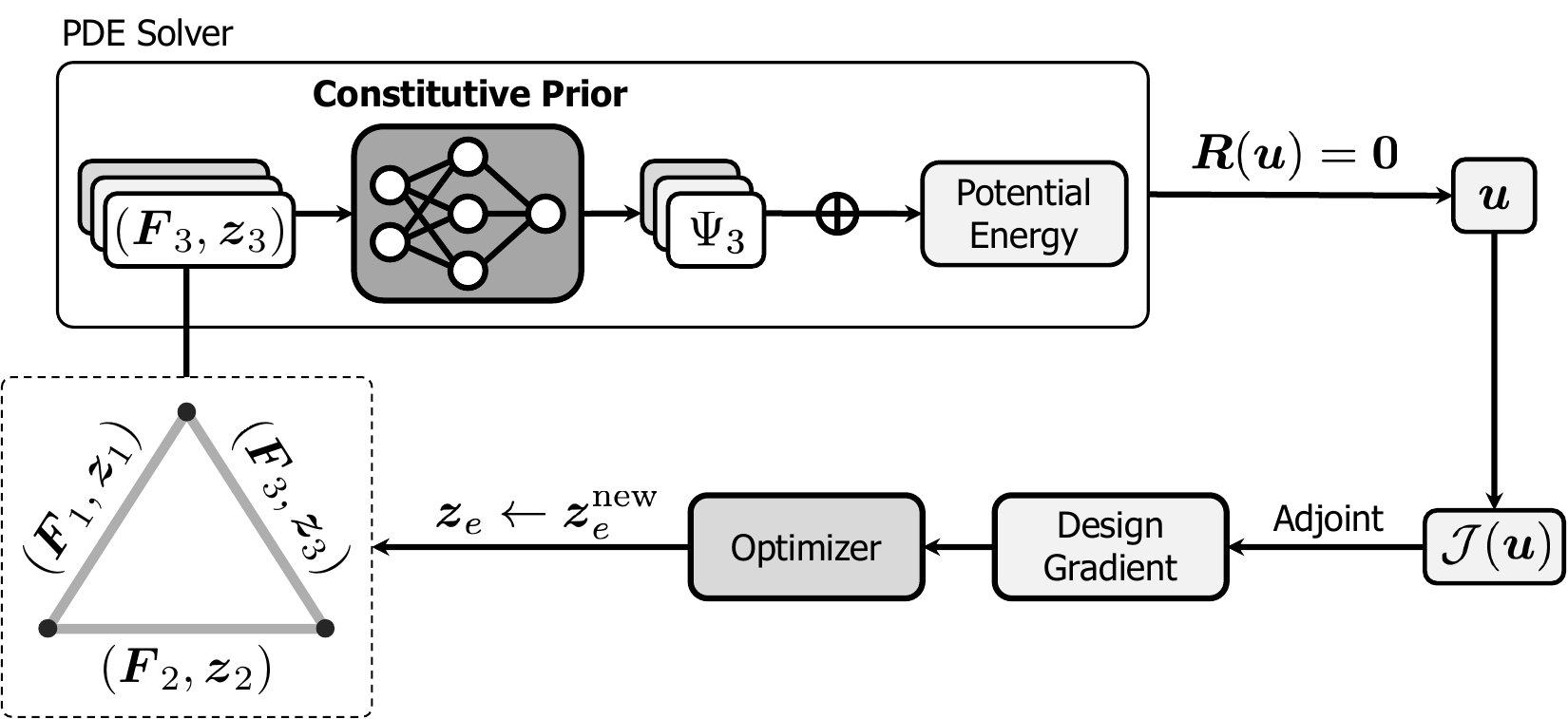}
    \end{subfigure}
    \hfill
    \caption{
        Schematic illustration of the adjoint-based optimization for heterogeneous material distributions in elastic networks using constitutive priors.
    }
    \label{fig:main_scheme}
\end{figure}

For the time-aggregated setting, the objective in the constrained optimization problem in \cref{eq:constrained_optim} is replaced by the objective in \cref{eq:temporal_loss}.
Because equilibrium is imposed at each load step along the prescribed trajectory, the corresponding Lagrangian is
\begin{equation}
    \mathscr{L}
    =
    \mathcal{J}_{\mathrm{time}}+
    \sum^{N_{\mathrm{step}}}_{\tau=1}
    (\boldsymbol{\Lambda}^{(\tau)})^\top
    \boldsymbol{R}^{(\tau)}
    (\boldsymbol{u}^{(\tau)};\boldsymbol{\theta}),
\end{equation}
where $\boldsymbol{\Lambda}^{(\tau)}$ and $\boldsymbol{R}^{(\tau)}$ denote the Lagrange multiplier and the PDE residual at the $\tau$-th load step, respectively.
Stationarity with respect to each $\boldsymbol{\Lambda}^{(\tau)}$ recovers the equilibrium constraint at each load step.
Stationarity with respect to each $\boldsymbol{u}^{(\tau)}$ yields
\begin{equation}
    \frac{1}{N_{\mathrm{step}}}
    \frac{
        \partial\mathcal{J}(
            \boldsymbol{u}^{(\tau)};
            \mathcal{T}^{(\tau)}
        )
    }
    {\partial\boldsymbol{u}^{(\tau)}}
    +
    \left(
        \frac
        {\partial\boldsymbol{R}^{(\tau)}}
        {\partial\boldsymbol{u}^{(\tau)}}
    \right)
    \boldsymbol{\Lambda}^{(\tau)}
    =
    \boldsymbol{0},
    \qquad
    \tau=1,\dots,N_{\mathrm{step}}.
\end{equation}
Solving these systems gives $\boldsymbol{\Lambda}^{(\tau)}$.
Then, substituting $\boldsymbol{\Lambda}^{(\tau)}$ into
\begin{equation}\label{eq:adjoint_gradient_temporal}
    \frac
    {\partial\mathscr{L}}
    {\partial\boldsymbol{\theta}}
    =
    \sum^{N_{\mathrm{step}}}_{\tau=1}
    \left(
        \frac
        {\partial\boldsymbol{R}^{(\tau)}}
        {\partial\boldsymbol{\theta}}
    \right)
    \boldsymbol{\Lambda}^{(\tau)}
\end{equation}
gives the gradient, which is used to update $\boldsymbol{\theta}$.

The material update loop using the constitutive prior is shown in \cref{fig:main_scheme}.
For each truss member, the deformation gradient $\boldsymbol{F}_e$ and latent material parameter $\boldsymbol{z}_e$ define the corresponding strain energy through the constitutive prior.
The resulting strain energies are assembled into the total potential energy, and the PDE solver updates the nodal displacement $\boldsymbol{u}_i$ until the equilibrium condition in \cref{eq:equilibrium} is satisfied.
The objective functional is evaluated at the equilibrium configuration.
The optimizer updates the member wise latent material parameters $\boldsymbol{z}_e$ using the adjoint gradient in \cref{eq:adjoint_gradient}.

\subsection{Homotopy-Continuation Strategy}
\label{sec:formulations:homotopy}

The inverse design problem considered in this work is nonlinear and nonconvex with respect to the design variables
$\boldsymbol{\theta}$, which can lead to poor convergence and
sensitivity to initialization.
To improve robustness, we employ a homotopy-based continuation strategy.
In contrast to classical continuation methods defined over parameter spaces~\cite{riks1979incremental,riks1984some}, the proposed homotopy operates directly on geometric configurations.

The key idea is to construct a sequence of intermediate target configurations that progressively deform the geometry from an initial configuration to the desired target, while using the solution at each stage to initialize the next.
This follows the homotopy-continuation strategy of connecting an easy problem to a difficult problem~\cite{allgower2003introduction,wayburn1987homotopy},
adopted as an initial target and a final target, respectively.

Let $\boldsymbol{u}(\boldsymbol{\theta})$ denote the equilibrium
configuration corresponding to the design variables.
Given an initial target configuration 
$\mathcal{T}^{(0)}=\{\boldsymbol{y}^{(0)}_k\}_{k=1}^{N_{\mathrm{targ}}}$ 
and a final target 
$\mathcal{T}^{(N_{\mathrm{cont}})}=\{\boldsymbol{y}^{(N_{\mathrm{cont}})}_k\}_{k=1}^{N_{\mathrm{targ}}}$, 
we define a family of intermediate targets
\begin{equation}\label{eq:homotopy_construction}
\mathcal{T}(\gamma) = 
\{
(1-\gamma)\,\boldsymbol{y}^{(0)}_k 
+ 
\gamma\,\boldsymbol{y}^{(N_{\mathrm{cont}})}_k
\}_{k=1}^{N_{\mathrm{targ}}},
\qquad \gamma \in [0,1].
\end{equation}
For each $r=0,\dots,N_{\mathrm{cont}}$ with $0=\gamma_0<\gamma_1<\dots<\gamma_{N_{\mathrm{cont}}}=1$ and $\mathcal{T}^{(r)}:=\mathcal{T}(\gamma_r)$, 
we solve the auxiliary problem
\begin{equation}
\boldsymbol{\theta}^{(r)}
\in
\underset{\boldsymbol{\theta}}{\mathrm{argmin}}
\;
\mathcal{J}(\boldsymbol{u}; \mathcal{T}^{(r)})
\quad
\text{subject to}
\quad
\boldsymbol{R}(\boldsymbol{u};\boldsymbol{\theta}) = \boldsymbol{0}.
\end{equation}
Starting from $\gamma_r=0$, where the problem is easier, the homotopy
parameter is gradually increased toward $\gamma_{N_{\mathrm{cont}}}=1$.
At each step, the solution $\boldsymbol{\theta}^{(r)}$ is used to initialize the optimization at $r+1$.
This continuation strategy smooths the optimization landscape and improves convergence.

In some practical settings, the target configuration is specified only partially as a point cloud $\mathcal{T}^{(N_{\mathrm{cont}})}$, without a corresponding initial configuration or explicit trajectory.
To construct a consistent homotopy path in the absence of a known initial target configuration,  we introduce an affine registration procedure~\cite{du2010affine,feldmar1996rigid}
that aligns the observed target point cloud with the reference configuration of the structure.
We define an affine transformation parameterized by $\boldsymbol{A} \in \mathbb{R}^{d \times d}$ and $\boldsymbol{b} \in \mathbb{R}^d$,
\begin{equation}
\mathcal{T}_{\mathrm{affine}}(\boldsymbol{A}, \boldsymbol{b})
=
\{ \boldsymbol{A} \boldsymbol{y}_k + \boldsymbol{b} \}_{k=1}^{N_{\mathrm{targ}}},
\end{equation}
where $\{\boldsymbol{y}_k\}_{k=1}^{N_{\mathrm{targ}}}$ denotes the given target configuration.

The optimal transformation $(\boldsymbol{A}^\star, \boldsymbol{b}^\star)$
is obtained by minimizing the Chamfer distance \cref{eq:chamfer_distance}
between the transformed target point cloud and the reference configuration of the structure.
We then define the initial target configuration as
\begin{equation}
\mathcal{T}^{(0)} := \mathcal{T}_{\mathrm{affine}}(\boldsymbol{A}^\star, \boldsymbol{b}^\star).
\end{equation}
This procedure enables the construction of a consistent homotopy path from partial geometric observations and provides a stable initialization for the continuation strategy.
We note that this affine registration step is applied only when the initial target configuration is not available.

Algorithm~\ref{alg:overview_inverse} summarizes the overall computational procedure combining homotopy continuation, equilibrium solves, and adjoint-based optimization updates. The convergence criterion at each continuation step is not prescribed by the framework and may instead be selected according to the requirements of the particular application.

\begin{algorithm}[htbp]
\caption{Homotopy-based inverse design with nested equilibrium solve}
\label{alg:overview_inverse}
\begin{algorithmic}[1]
\Require Initial state $\boldsymbol{u}^{(0)}$, initial design $\boldsymbol{\theta}^{(0)}$, targets $\{\mathcal{T}^{(r)}\}_{r=1}^{N_{\mathrm{cont}}}$
\For{$r=1,\dots,N_{\mathrm{cont}}$}
    \State $\boldsymbol{u}_{\mathrm{init}} \gets \boldsymbol{u}^{(r-1)}$, $\boldsymbol{\theta}^{(r,0)} \gets \boldsymbol{\theta}^{(r-1)}$
    \For{$\kappa=1,\dots,N_{\max}$}
        \State Solve equilibrium: 
        $\boldsymbol{R}(\boldsymbol{u};\boldsymbol{\theta}^{(r,\kappa-1)}) = \boldsymbol{0}$ with initial guess $\boldsymbol{u}_{\mathrm{init}}$
        \State Compute loss: $\mathcal{J} \leftarrow \mathcal{J}(\boldsymbol{u}; \mathcal{T}^{(r)})$
        \If{converged}
            \State \textbf{break}
        \EndIf
        \State Compute gradient: $\nabla_{\boldsymbol{\theta}} \mathscr{L}$
        \State Update design: $\boldsymbol{\theta}^{(r,\kappa)} \leftarrow \mathrm{Optimizer}(\boldsymbol{\theta}^{(r,\kappa-1)}, \nabla_{\boldsymbol{\theta}} \mathscr{L})$
        \State $
        \boldsymbol{u}_{\mathrm{init}} \gets \boldsymbol{u},\quad
        \boldsymbol{\theta} \gets \boldsymbol{\theta}^{(r,\kappa)}
        $
    \EndFor
    \State $\boldsymbol{u}^{(r)} \gets \boldsymbol{u}$,\quad $\boldsymbol{\theta}^{(r)} \gets \boldsymbol{\theta}$
\EndFor
\end{algorithmic}
\end{algorithm}

\subsection{Graph-based Smoothness Metric}
\label{sec:formulations:smoothness}

We introduce a graph-based metric to quantify spatial variation in
member-wise material fields. This is used to quantify the spatial smoothness of the inferred material distributions in post-analysis and to assess the extent of spatial regularity induced by the learned design.

Let $\boldsymbol{a} \in \mathbb{R}^{N_e}$ denote a scalar field defined
over the structural elements, where $N_e = |\mathcal{E}|$. For example,
$\boldsymbol{a}$ may represent thermal expansion coefficients or other
material parameters. To define adjacency, we construct an element-adjacency graph in which
each element is treated as a node, and edges connect elements that share
a common node in the structure. Let $\boldsymbol{W} \in \mathbb{R}^{N_e \times N_e}$ denote the adjacency matrix, where $W_{ee'}=1$ if elements $e$ and $e'$ are adjacent, and $W_{ee'}=0$ otherwise.
Let $\boldsymbol{D}$ be the degree matrix with $D_{ee} = \sum_{e'} W_{ee'}$,
and define the graph Laplacian $\boldsymbol{\Delta} = \boldsymbol{D} - \boldsymbol{W}$.
We quantify spatial variation using the normalized graph Dirichlet energy \cite{zhou2003learning}
\begin{equation}
\widetilde{\mathscr{E}}(\boldsymbol{a}) =
\frac{\boldsymbol{a}^\top \boldsymbol{\Delta} \boldsymbol{a}}
{\boldsymbol{a}^\top \boldsymbol{a}},
\end{equation}
which measures the relative variation of the field across adjacent elements. Smaller values indicate smoother spatial variation, while larger values indicate sharper transitions in material properties.

\section{Numerical Results}\label{sec:result}

Throughout this section, we illustrate the application of the present framework of constitutive prior to several inverse design problems for shape-morphing material systems.
In all examples, the design variable $\boldsymbol{z}$ is associated with the material response governing the system, where we assume that the material response can vary continuously with respect to parameters such as material constituents or manufacturing processes.
Accordingly, each member-wise parameter $\boldsymbol{z}_e$ evolves over the manifold induced by the constitutive prior, naturally yielding heterogeneous material designs.
Both data-driven constitutive priors and classical material models are considered.

Unless otherwise specified, the numerical examples follow common settings for geometry generation and optimization.
The lattice structures in Sections~\ref{sec:result:airfoil}, \ref{sec:result:crack}, \ref{sec:result:circle}, and \ref{sec:result:heart} are generated using \textsc{Gmsh}~\cite{geuzaine2009gmsh}.
The design variables are assigned heterogeneously on a member-wise basis, subject to admissible bounds and initialized within those bounds.
Inverse design procedures and model training processes in this work are carried out using Adam~\cite{kingma2014adam}.
This choice provides a unified gradient-based optimization scheme across the approaches considered in the present framework, while distinguishing it from conventional topology optimization settings, where MMA~\cite{svanberg1987method,bendsoe2013topology} is commonly adopted.
For the affine-registration example in Section~\ref{sec:result:crack}, the stopping tolerance is adjusted across continuation steps according to the procedure described in that section.

The remainder of this section is organized as follows.
The data-driven construction of the constitutive prior and the effect of surrogate smoothness on the optimization are shown in Section~\ref{sec:result:C-prior}.
The inverse identification of space--time trajectories, together with a comparison of performance across different surrogate landscapes, is illustrated in Section~\ref{sec:result:bar}.
Applicability of the data-driven surrogate to an engineering-motivated problem setting on airfoil geometry is shown in Section~\ref{sec:result:airfoil}.
The construction of an initial target shape via affine registration under topological mismatch is presented in Section~\ref{sec:result:crack}.
The effect of design-variable representation on thermoelastic material distributions, together with a comparison with MMA, is presented in Section~\ref{sec:result:circle}.
Inverse design for a geometrically complex target shape (nonconvex), together with an ablation of the homotopy-based continuation and the effect of the smoothness of the latent landscape, is demonstrated in Section~\ref{sec:result:heart}.
The non-unique nature of the inverse problems considered here is discussed in Section~\ref{sec:nonunique}.
Details on hyperparameter settings are provided in Appendix~\ref{apdx:optimization}.
The hardware, software, and computational costs associated with the selected examples are summarized in Appendix~\ref{apdx:compute_cost}.

\begin{remark}
Throughout this work, the term ``time'' refers to a pseudo-time or loading/continuation step used to parameterize a quasi-static deformation path.
We do not consider transient dynamics, or inertial effects.
\end{remark}

\subsection{Constitutive Prior Construction}\label{sec:result:C-prior}

The inverse design problems considered in this work may require a prescribed class of admissible material responses, such as manufacturable material responses, over which the design variables $\boldsymbol{z}$ can be optimized.
In general settings where no closed-form constitutive model is available, this admissible class may instead be specified through unordered data collected from experiments, simulations, or their combination.

In this section, we illustrate the data-driven surrogate modeling of hyperelastic material responses for constructing the constitutive prior.
The unknown latent design variable $\boldsymbol{z}$ selects a stress--strain response from the learned constitutive family. Thermodynamic consistency in the isothermal hyperelastic setting, in the sense of vanishing mechanical dissipation, follows from deriving the stress from a strain-energy potential, while \cref{eq:stress_free} enforces a zero-energy and stress-free reference configuration.
We further examine the smoothness of the resulting gradient landscape, since the behavior of gradient-based inverse design depends not only on the accuracy of the fitted material responses, but also on how the constitutive response varies with respect to the latent design variable.

\begin{figure}[htbp]
    \centering
    \begin{subfigure}{0.4\textwidth}
        \centering
        \includegraphics[width=\textwidth]{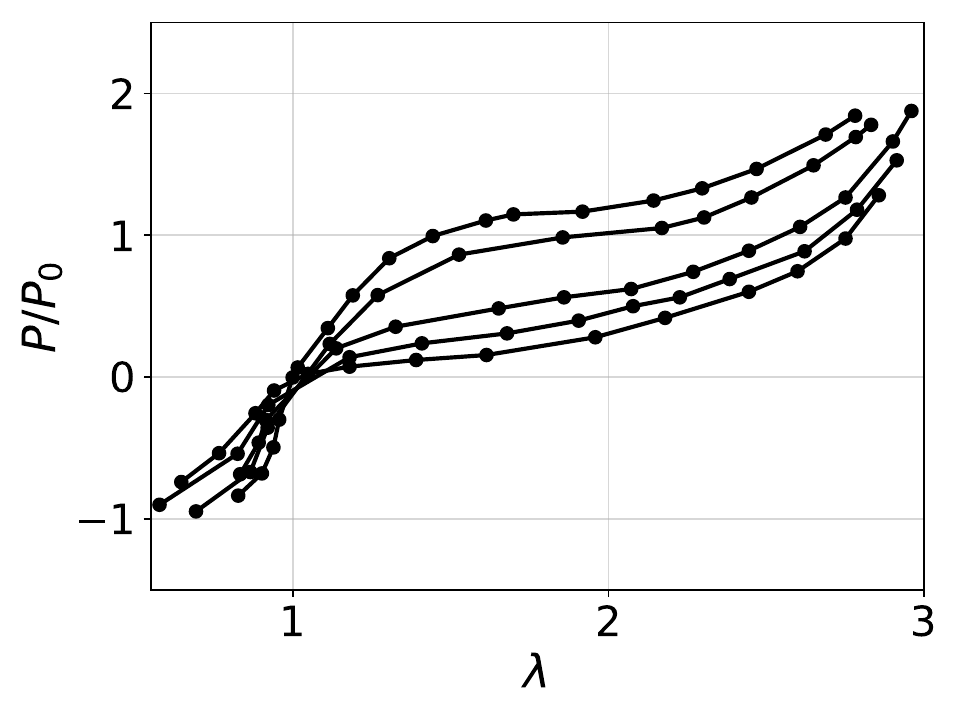}
        \caption{}\label{fig:pr1_deepSDF_a}
    \end{subfigure}
    \begin{subfigure}{0.4\textwidth}
        \centering
        \includegraphics[width=\textwidth]{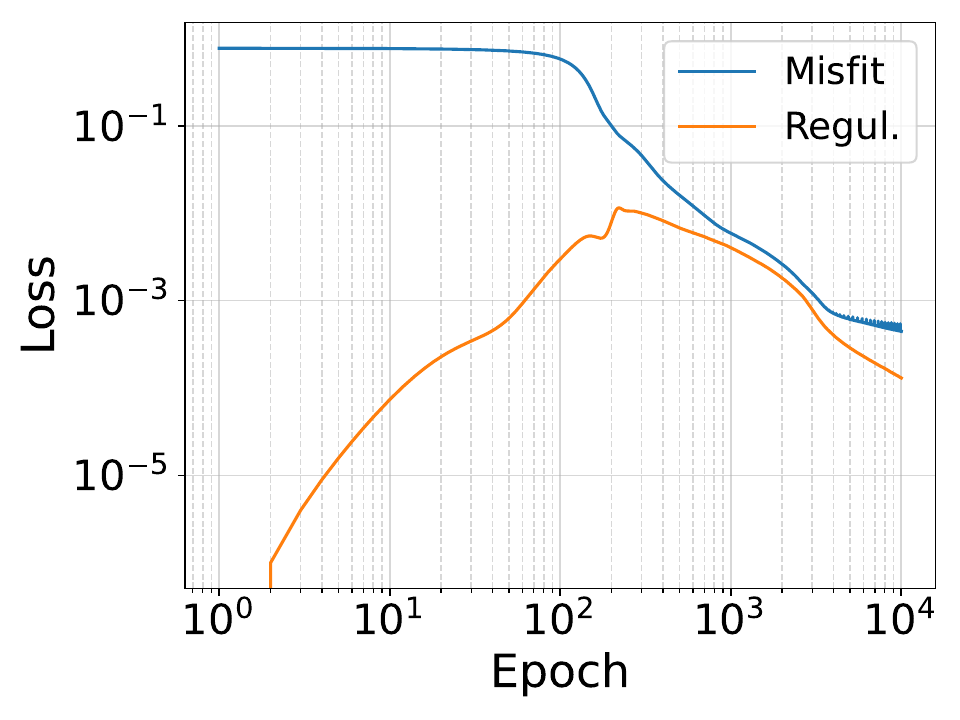}
        \caption{}\label{fig:pr1_deepSDF_b}
    \end{subfigure}
    \begin{subfigure}{0.4\textwidth}
        \centering
        \includegraphics[width=\textwidth]{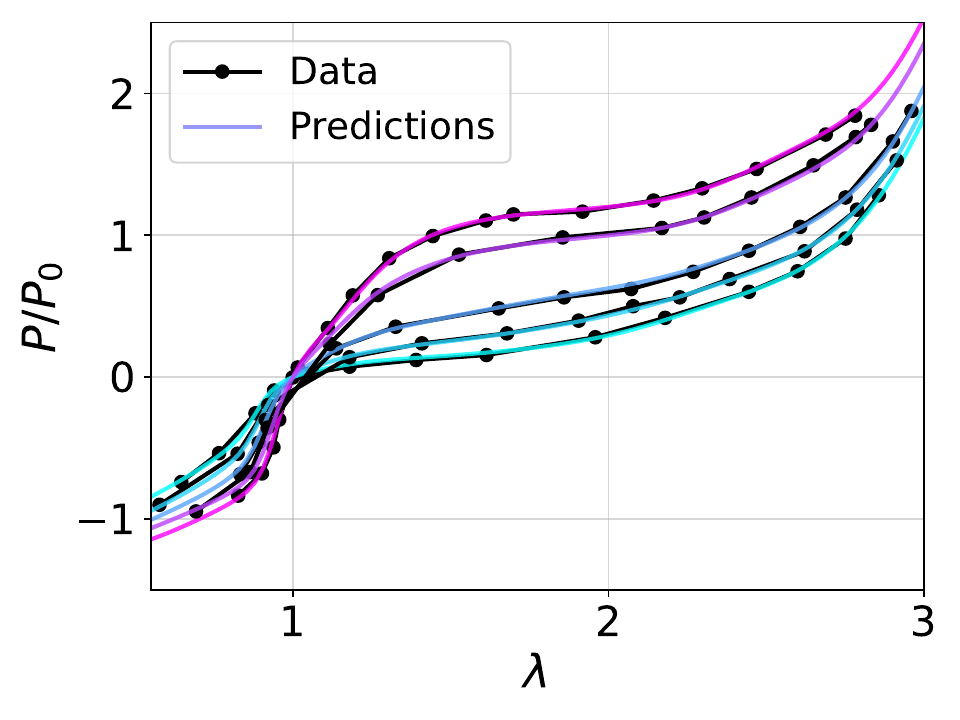}
        \caption{}\label{fig:pr1_deepSDF_c}
    \end{subfigure}
    \begin{subfigure}{0.4\textwidth}
        \centering
        \includegraphics[width=\textwidth]{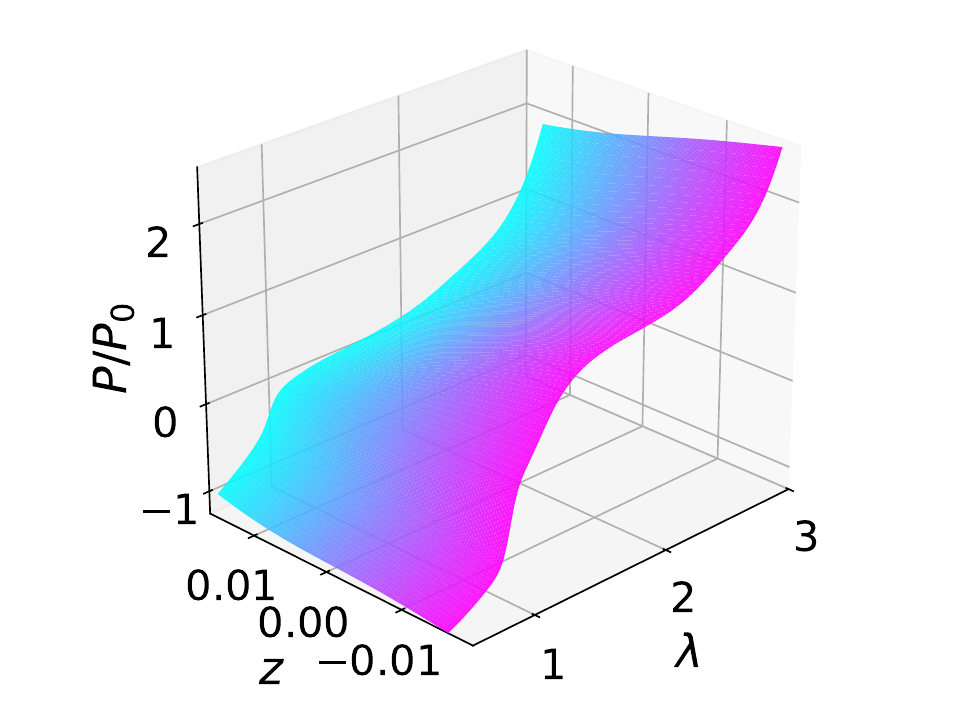}
        \caption{}\label{fig:pr1_deepSDF_d}
    \end{subfigure}
    \hfill
    \caption{
        Training of the material surrogate with trainable material parameter $z$.
        (\subref{fig:pr1_deepSDF_a}) Given data.
        (\subref{fig:pr1_deepSDF_b}) Loss histories of the data-misfit and regularization terms during training.
        Predictions of the trained model
        (\subref{fig:pr1_deepSDF_c}) at the optimized latent values $z$ for each realization, and
        (\subref{fig:pr1_deepSDF_d}) under continuous variation of $z$.
        The first Piola--Kirchhoff stress $P$ is normalized by an arbitrary reference stress $P_0$.
    }
    \label{fig:pr1_deepSDF}
\end{figure}

To represent a general setting, we consider an arbitrary class of hyperelastic materials characterized by five \textit{hand-drawn} datasets, as illustrated in \cref{fig:pr1_deepSDF_a}. The datasets are intentionally constructed to exhibit strong nonlinear behavior. 
The first Piola--Kirchhoff stress $P$ is normalized by the reference value $P_0 = 1\,\mathrm{MPa}$, resulting in a nondimensional constitutive response that can later be rescaled according to the application of interest. The datasets are not assumed to be thermodynamically consistent \textit{a priori}, reflecting potential noise, inconsistencies, or imperfections in experimental observations. 
Furthermore, both the sampling locations and the number of observations vary across datasets, consistent with the unordered point cloud representation employed in the surrogate construction framework.
The effects of noise in the training data on model training and subsequent inverse-design optimization are investigated in Appendix~\ref{apdx:noise}.

The material surrogate $\mathcal{M}_{\boldsymbol{\phi}}$ is trained using the loss in \cref{eq:surrogate_loss}, with the corresponding loss histories shown in \cref{fig:pr1_deepSDF_b}.
As a result, each realization is assigned a specific value of the latent variable $z\in\mathbb{R}$ as shown in \cref{fig:pr1_deepSDF_c}, which clearly reflects the overall monotonically increasing material response and the stress-free behavior at the unstretched state $\lambda=1$.
The trained surrogate further yields a continuous representation of the stress--stretch response parameterized by $z$, as shown in \cref{fig:pr1_deepSDF_d}. Overall, this surrogate can be viewed as a continuous material template grounded in the given data and defines a continuous family of candidate constitutive responses over the selected latent interval. Determining whether an inferred response can be realized through a particular composition, microstructure, or manufacturing process is a separate downstream problem.
The optimized latent values define the admissible interval $\mathcal{Z} := [z_{\min}, z_{\max}]$, 
which is used as the design domain for subsequent inverse design examples involving this surrogate.
This material surrogate is adopted as a constitutive prior in selected examples in Section~\ref{sec:result}, while classical hyperelastic models are also considered where appropriate.

Meanwhile, the smooth landscape shown in \cref{fig:pr1_deepSDF_d} results from assigning the latent variable $z$ to each material response adaptively during training, and this smoothness affects the subsequent optimization performance.
Revisiting \cref{eq:adjoint_gradient}, the sensitivity of the force resultant landscape with respect to the design parameters is reflected in $\partial\boldsymbol{R}/\partial\boldsymbol{\theta}$.
A non-smooth surrogate can therefore induce an irregular objective landscape and degrade the behavior of gradient-based optimization~\cite{nocedal2006numerical}.

To further investigate the effect of this smoothness, we train another surrogate with fixed latent variables $z$ as a counterpart that exhibits a non-smooth landscape.
The material-specific parameters are independently sampled from the standard normal distribution $\mathcal{N}(0,1)$ and assigned to each constitutive response curve. The surrogate model is then trained using the loss formulation defined in \cref{eq:surrogate_loss}, excluding both the regularization term and the non-trainable latent variables $\boldsymbol{z}$.

\begin{figure}[htbp]
    \centering
    \begin{subfigure}{0.4\textwidth}
        \centering
        \includegraphics[width=\textwidth]{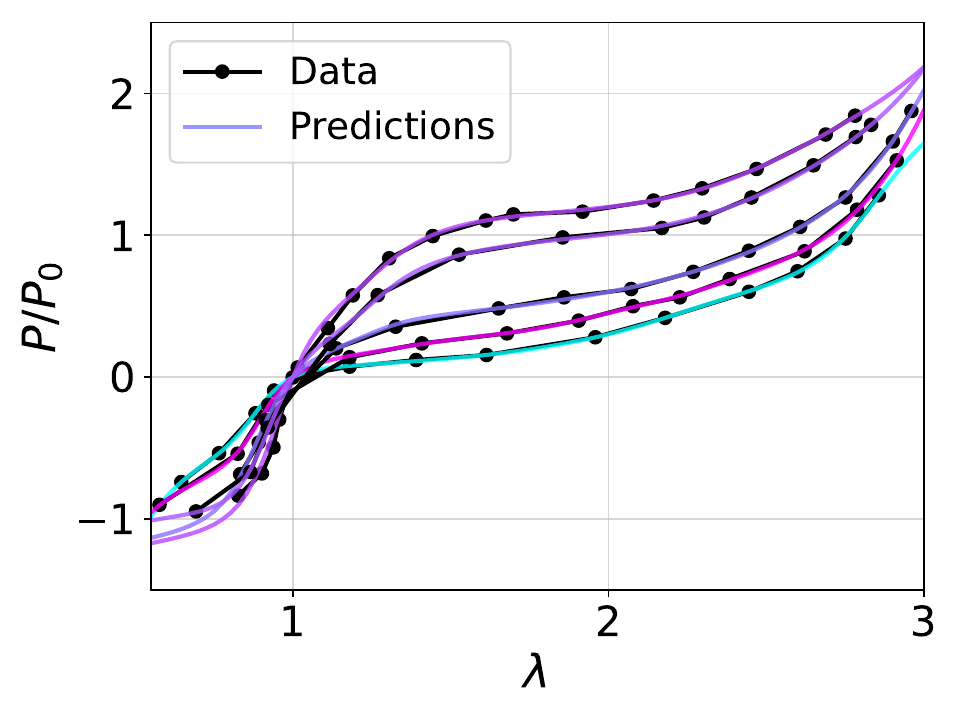}
        \caption{}
        \label{fig:pr1_bar_fixed_latent_a}
    \end{subfigure}
    \begin{subfigure}{0.4\textwidth}
        \centering
        \includegraphics[width=\textwidth]{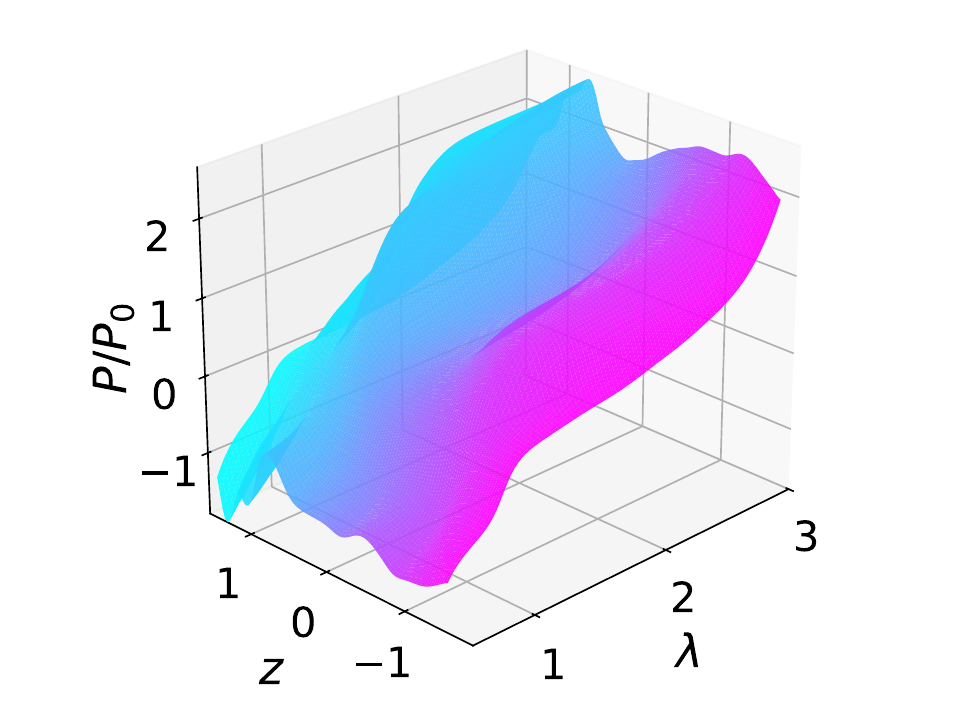}
        \caption{}
        \label{fig:pr1_bar_fixed_latent_b}
    \end{subfigure}
    \hfill
    \caption{
        Results of training the material surrogate with ``fixed'' (non trainable) latent variables $z$.
        Predictions of the trained model (\subref{fig:pr1_bar_fixed_latent_a}) at the given latent values $z$ for each realization, and (\subref{fig:pr1_bar_fixed_latent_b}) under continuous variation of $z$.
    }
    \label{fig:pr1_bar_fixed_latent}
\end{figure}

The results are shown in \cref{fig:pr1_bar_fixed_latent}.
As shown in \cref{fig:pr1_bar_fixed_latent_a}, the trained surrogate reproduces the data with quality comparable to that in \cref{fig:pr1_deepSDF_c}.
The notable difference appears in \cref{fig:pr1_bar_fixed_latent_b}: the surrogate output along $z$ is markedly more oscillatory than the corresponding landscape in \cref{fig:pr1_deepSDF_d}.
We note that this non-smooth surrogate remains continuous and differentiable with respect to $z$.
However, compared with the adaptively assigned latent representation, it yields less smooth interpolation between material responses and is therefore expected to provide less favorable gradients for subsequent optimization procedures.
This non-smooth surrogate is adopted as a constitutive prior in selected examples below to investigate the effect of non-smoothness on the optimization behavior of inverse design procedures.
For conciseness in the following examples, we refer to the surrogate with trainable latent variables in \cref{fig:pr1_deepSDF_d} as constitutive prior A, and the surrogate with fixed latent variables in \cref{fig:pr1_bar_fixed_latent_b} as constitutive prior B.

\subsection{One-dimensional Bar}\label{sec:result:bar}

In this section, we present a pedagogical example based on a one-dimensional bar exhibiting hyperelastic behavior.
This example illustrates the capacity of the present framework to recover pseudo-temporal trajectories.
Furthermore, we investigate the effect of the smoothness of the constitutive prior on the optimization behavior.

\begin{figure}[htbp]
    \centering
    \begin{subfigure}[b]{0.23\textwidth}
        \centering
        \includegraphics[width=\textwidth]{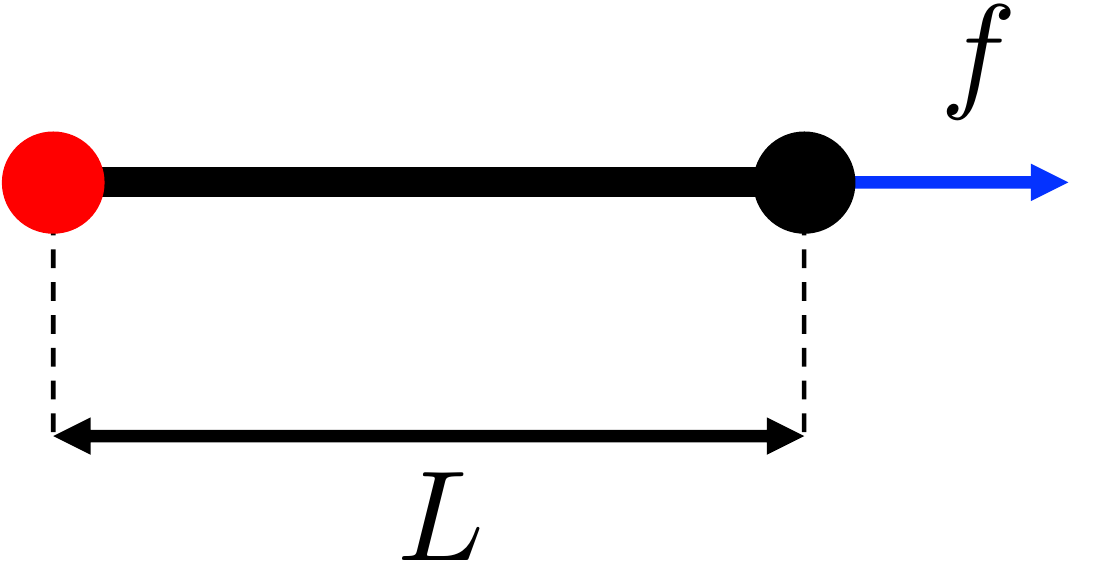}
    \end{subfigure}
    \caption{
        A one-dimensional bar with length $L=1\,\mathrm{m}$.
        The left node is fixed (red dot), and a tensile nodal force $f=2\,\mathrm{N}$ (blue arrow) is applied at the right node, where the cross sectional area of the bar is set to $A=1\,\mathrm{mm}^2$.
    }
    \label{fig:pr1_bar_setting}
\end{figure}

The bar shown in \cref{fig:pr1_bar_setting}~is elongated over multiple load steps, and the nodal displacement is recorded at each step under step-wise incremental nodal loading.
The target trajectory is generated using constitutive prior A with a target latent parameter $z_{\mathrm{targ}}$ sampled from the admissible range $\mathcal{Z}$, as shown in \cref{fig:pr1_trainsient_bar_a}.
The parameter $z$ is subsequently identified via adjoint-based design optimization by minimizing the time-aggregated objective \cref{eq:temporal_loss}, where the gradient in \cref{eq:adjoint_gradient_temporal} is used to update $z$, initialized from the same sampling strategy.
After each update, $z$ is clamped to the admissible range $\mathcal{Z}$.

\begin{figure}[htbp]
    \centering
    \begin{subfigure}{0.4\textwidth}
        \centering
        \includegraphics[width=\textwidth]{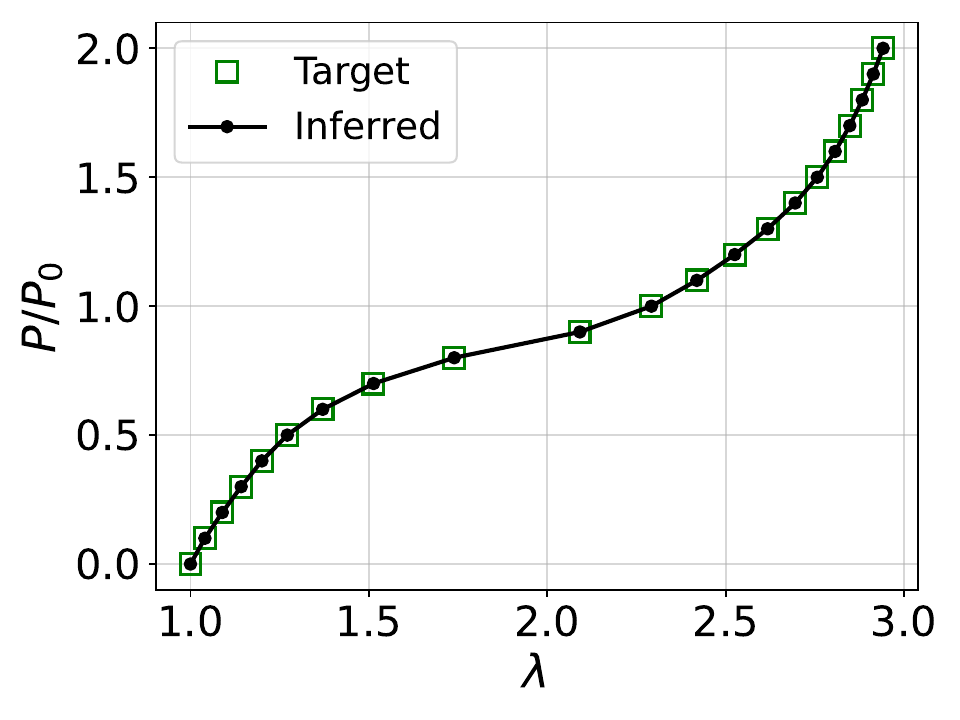}
        \caption{}
        \label{fig:pr1_trainsient_bar_a}
    \end{subfigure}
    \begin{subfigure}{0.4\textwidth}
        \centering
        \includegraphics[width=\textwidth]{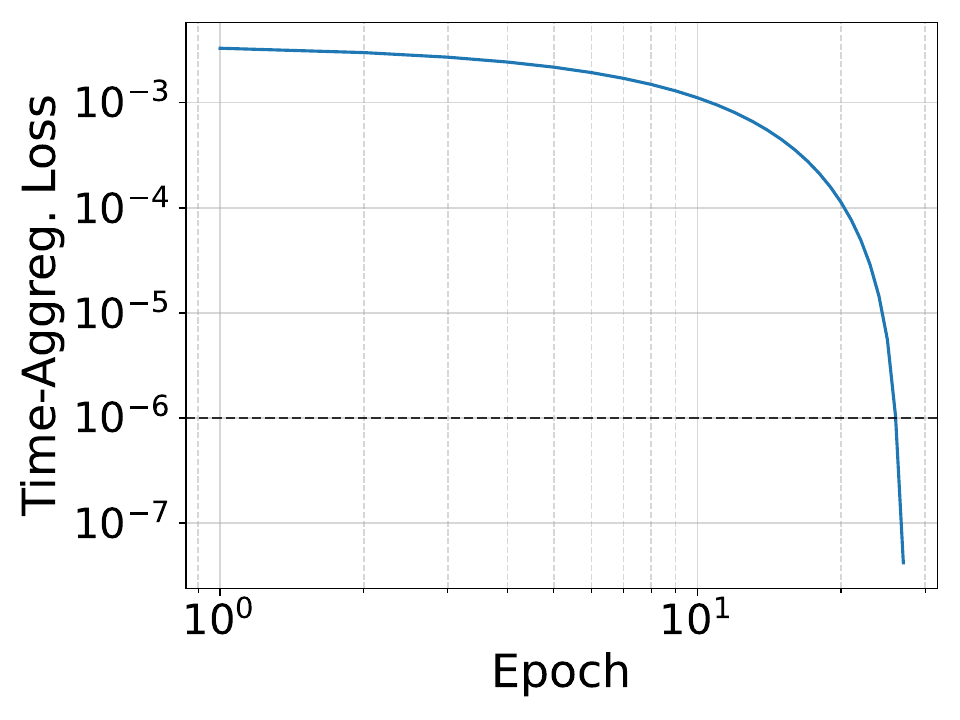}
        \caption{}
        \label{fig:pr1_trainsient_bar_b}
    \end{subfigure}
    \hfill
    \caption{
    Results of the inverse problem for the uniaxial extension of a bar, converged to a loss threshold for the displacement trajectory.
    (\subref{fig:pr1_trainsient_bar_a}) $P$--$\lambda$ trajectories of the target bar and the inferred bar.
    (\subref{fig:pr1_trainsient_bar_b}) The loss history during the inverse identification.
    The threshold is shown as a black dashed line.
    }
    \label{fig:pr1_trainsient_bar}
\end{figure}

The results are shown in \cref{fig:pr1_trainsient_bar}.
The inferred $P$--$\lambda$ curve is shown in \cref{fig:pr1_trainsient_bar_a}, demonstrating close agreement with the prescribed space--time trajectory.
As shown in \cref{fig:pr1_trainsient_bar_b}, the loss decreases smoothly and reaches the stopping criterion (final loss: $4.191\times10^{-8}$), yielding an inferred latent parameter $z=-4.43\times10^{-3}$, compared with the target value $z_{\mathrm{targ}}=-4.44\times10^{-3}$.
Overall, these results indicate that the proposed framework can match the full space--time trajectory of the structural response while accurately identifying the underlying material parameter.

\subsubsection{Effect of the Constitutive Prior Landscape}
\label{sec:result:bar:compare}

The smoothness of the constitutive prior with respect to $z$ directly affects the design gradient used in the inverse design procedure.
Thus, even when two surrogates reproduce the observed material responses with comparable accuracy, they may induce different optimization behavior.
To examine this effect, we compare constitutive prior A with constitutive prior B.

We use the same problem setting as in Section~\ref{sec:result:bar} for both surrogates.
For each surrogate, a target trajectory is generated by sampling a target latent value $z_{\mathrm{targ}}$ uniformly at random from the corresponding admissible range and solving the forward problem under the same loading sequence.
The inverse identification is then performed under three independent initial guesses, each drawn uniformly at random from the admissible range.
Because the two surrogates induce inverse optimization landscapes with different scalings along the latent variable $z$, a single learning rate is not directly comparable between them. Consequently, their performance cannot be fairly assessed using the same optimization hyperparameters. We hence sweep the learning rate over nine representative values $\{1\times10^{-4}, 3\times10^{-4}, 1\times10^{-3}, 3\times10^{-3}, \dots, 1\times10^{0}\}$ and report, for each surrogate, the three learning rates that yield the smallest average number of iterations to reach the stopping criterion across the trials.

\begin{figure}[htbp]
    \centering
    \begin{subfigure}{0.5\textwidth}
        \centering
        \includegraphics[width=\textwidth]{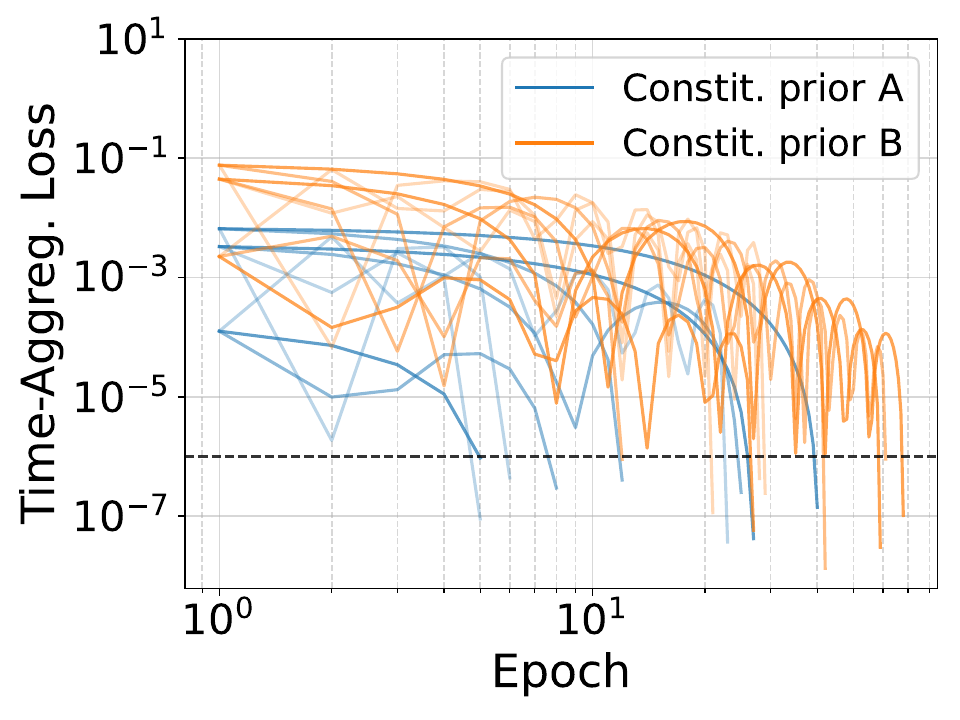}
    \end{subfigure}
    \hfill
    \caption{
    The comparison between the two surrogates in the inverse identification task.
    The prescribed threshold is shown as a black dashed line.
    For constitutive prior A, the learning rates $1\times10^{-4}$, $3\times10^{-4}$, and $3\times10^{-3}$ are shown.
    For constitutive prior B, the learning rates $3\times10^{-2}$, $1\times10^{-1}$, and $3\times10^{-1}$ are shown.}
    \label{fig:pr1_bar_compare}
\end{figure}

The loss histories from the multiple trials are shown in \cref{fig:pr1_bar_compare}.
For the cases shown, constitutive prior A reaches the stopping criterion in $16.8$ iterations on average, whereas constitutive prior B requires $38.6$ iterations on average.
This supports that, although both surrogates reproduce the observed material responses with comparable accuracy, the more oscillatory latent landscape leads to slower convergence in the inverse design procedure.

\subsection{Airfoil under Wind Gusts}\label{sec:result:airfoil}

In this section, we examine the use of the data-driven constitutive prior in an engineering-motivated problem setting.
Specifically, we consider an airfoil-shaped truss system under a prescribed wind gust loading profile and seek a heterogeneous material assignment within the admissible range that reproduces a target deformed configuration induced by a different loading profile.

The setup follows \cite{chen2026physical}, where a NACA 2412 airfoil is subjected to a wind-gust loading profile and is required to match a prescribed target shape under the induced deformation.
The airfoil coordinates are obtained from \cite{gorissen_naca_github}.
This problem setting is motivated by the control of airfoil deformation under aerodynamic loading.
Unlike \cite{chen2026physical}, where two linear materials are considered, we allow the material response to be selected from the templates introduced in \cref{fig:pr1_deepSDF_a} through the model in \cref{fig:pr1_deepSDF_d}.
Accordingly, we adopt the same material surrogate, with the reference stress value set to $P_0 = 10\,\mathrm{MPa}$.

\begin{figure}[htbp]
    \centering
    \begin{subfigure}{0.4\textwidth}
        \centering
        \includegraphics[width=\textwidth]{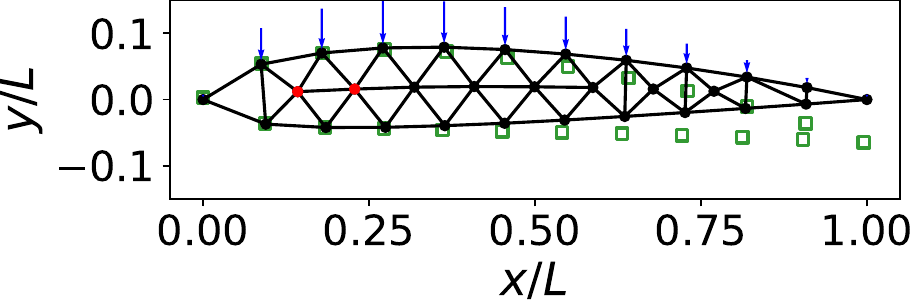}
        \caption{}
        \label{fig:pr2_airfoil_a}
    \end{subfigure}
    \hspace{0.04\textwidth}
    \begin{subfigure}{0.4\textwidth}
        \centering
        \includegraphics[width=\textwidth]{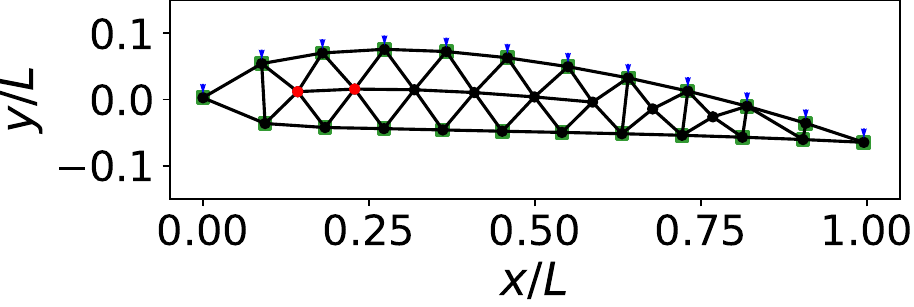}
        \caption{}
        \label{fig:pr2_airfoil_b}
    \end{subfigure}
    \hfill
    \caption{
        Problem setting of the airfoil-shaped truss system with $L=1\,\mathrm{m}$.
        The cross-sectional area of each member is prescribed as $A_e = 100\,\mathrm{mm}^2$.
        Fixed supports are shown as red markers, the target shape as green open squares, and the applied nodal forces as blue arrows.
        The original NACA 2412 airfoil under the prescribed load profile is shown in (\subref{fig:pr2_airfoil_a}), while the target shape generated under uniform nodal forces $f=1\,\mathrm{N}$ is shown in (\subref{fig:pr2_airfoil_b}).
    }
    \label{fig:pr2_airfoil}
\end{figure}

The problem setting is illustrated in \cref{fig:pr2_airfoil}.
As shown in \cref{fig:pr2_airfoil_a}, the truss system is subjected to a downward load profile prescribed as
\begin{align}
    f(\tilde{x}) =& -75\tilde{x}\left(1-\tilde{x}\right)^{2.5},\\
    \tilde{x} :=& \;x/L,
\end{align}
which is applied as downward nodal forces on the top-surface nodes according to their normalized $\tilde{x}$-coordinates.
The objective is not simply to stiffen the structure everywhere, since the target is not the undeformed configuration, but another deformed configuration associated with a smaller uniform loading level, as shown in \cref{fig:pr2_airfoil_b}.
The target configuration is generated by assigning the same value of $z_0$ to all members; here, $z_0$ is chosen as the midpoint of the admissible range $\mathcal{Z}$.
Intermediate target shapes are then constructed between the reference shape and the target shape for homotopy-based continuation using \cref{eq:homotopy_construction}, after which the inverse design procedure in \cref{alg:overview_inverse} is applied.

\begin{figure}[htbp]
    \centering
    \begin{subfigure}{0.55\textwidth}
        \centering
        \includegraphics[width=\textwidth]{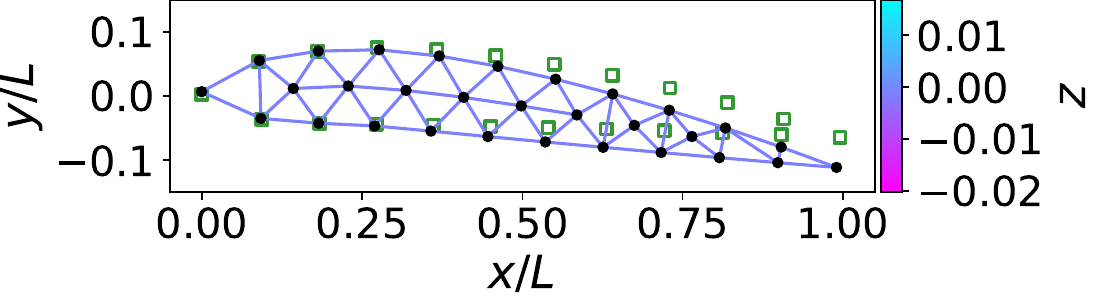}
        \caption{}
        \label{fig:pr2_airfoil_result_a}
    \end{subfigure}
    \hspace{0.04\textwidth}
    \begin{subfigure}{0.55\textwidth}
        \centering
        \includegraphics[width=\textwidth]{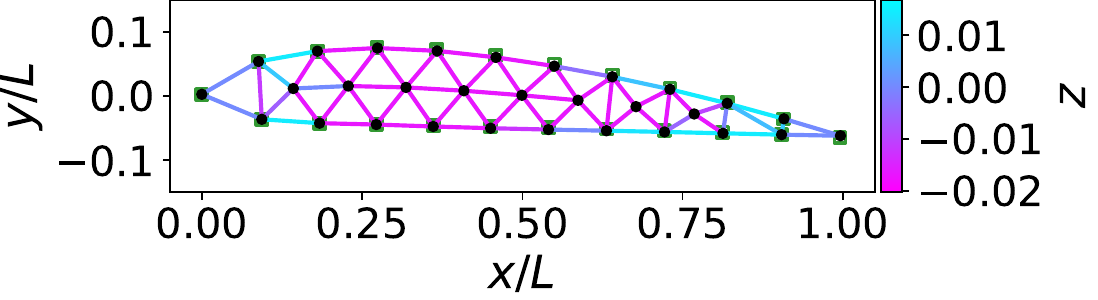}
        \caption{}
        \label{fig:pr2_airfoil_result_b}
    \end{subfigure}
    \hfill
    \caption{
        Inverse design results for the airfoil-shaped truss system under the arbitrary wind-gust loading profile.
        The target shape is shown as green open squares, and the values of latent parameter $z$ are indicated by color.
        (\subref{fig:pr2_airfoil_result_a}) Deformed configuration without optimization under the load shown in \cref{fig:pr2_airfoil_a}.
        (\subref{fig:pr2_airfoil_result_b}) Inferred design after optimization.
    }
    \label{fig:pr2_airfoil_result}
\end{figure}

The results are shown in \cref{fig:pr2_airfoil_result}.
The latent variables are initialized uniformly across the truss members.
As seen in \cref{fig:pr2_airfoil_result_a}, this initial design produces a substantially larger deflection than the target shape.
The obtained distribution of $z$ yields a deformed shape that closely matches the prescribed target (final Chamfer distance: $2.937\times10^{-6}$), as shown in \cref{fig:pr2_airfoil_result_b}.

\subsection{Edge Crack Lattice Identification}\label{sec:result:crack}

In the third example, we consider an inverse design problem in which a square truss system is driven toward a target configuration containing a central crack-like notch. This setting introduces an additional challenge beyond constitutive identification: the target geometry is topologically mismatched with the reference configuration, and therefore a suitable initial target configuration for the homotopy continuation is not available \textit{a priori}. 
To isolate the effect of geometric mismatch from constitutive modeling complexity, a classical material model is adopted, while the initial target configuration is constructed through the proposed affine registration procedure.

\begin{figure}[htbp]
    \centering
    \begin{subfigure}{0.25\textwidth}
        \centering
        \includegraphics[width=\textwidth]{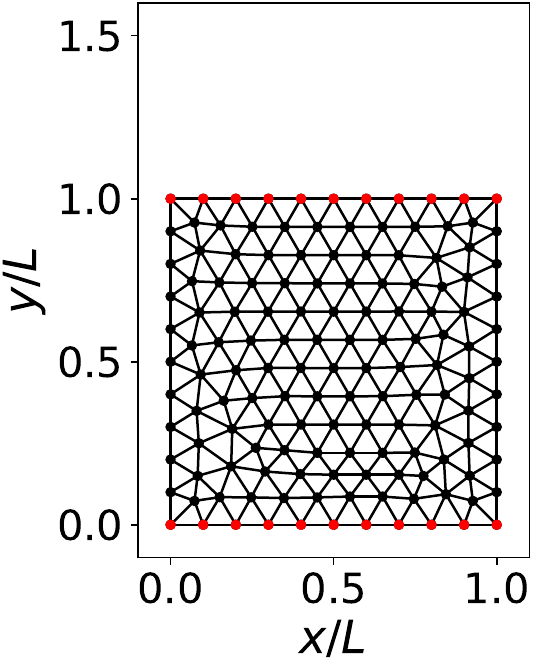}
        \caption{}
        \label{fig:pr3_crack_setting_a}
    \end{subfigure}
    \hspace{0.1\textwidth}
    \begin{subfigure}{0.25\textwidth}
        \centering
        \includegraphics[width=\textwidth]{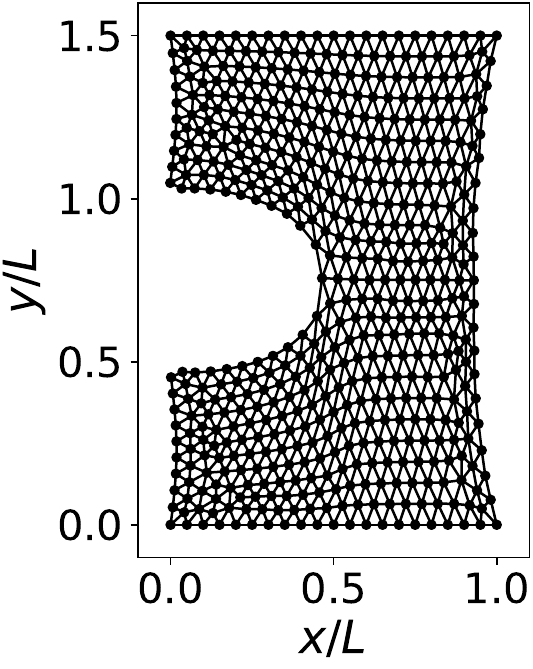}
        \caption{}
        \label{fig:pr3_crack_setting_b}
    \end{subfigure}
    \hfill
    \caption{
        Problem setting of the edge crack lattice example.
        (\subref{fig:pr3_crack_setting_a}) Reference square truss layout, shown by black points and solid lines, with nodes subject to prescribed displacements shown in red.
        (\subref{fig:pr3_crack_setting_b}) Target shape generated using a mesh with different topology.
        The cross-sectional areas of the truss members are prescribed so that every member has the same volume, i.e., $A_eL_e=1{,}000\mathrm{mm}^3$.
    }
    \label{fig:pr3_crack_setting}
\end{figure}

The problem setting and associated geometric configurations are shown in \cref{fig:pr3_crack_setting}.
The reference structure is a relatively coarse truss network that does not explicitly contain a crack, as shown in \cref{fig:pr3_crack_setting_a}.
The structure is then deformed by prescribing vertical displacements on the top and bottom boundaries so that its height reaches $150\%$ of the reference height.
We use the Neo-Hookean material model 
\begin{equation}\label{eq:neohookean}
    \Psi_{\mathrm{NH}}
    =
    \frac{E}{4}
    \left(
    \lambda^2 - 1 - 2\ln(\lambda)
    \right),
\end{equation}
where $E$ denotes the initial tangent modulus.

The target shape in \cref{fig:pr3_crack_setting_b} is generated under the same boundary condition and the material model \cref{eq:neohookean} with initial tangent modulus $E=1\,\mathrm{MPa}$.
The reference configuration of the target truss is generated from the same unit square but with finer mesh resolution, where a crack-like notch is introduced by removing nodes and connections along half of the midline.
As a result, the target shape differs from the system in \cref{fig:pr3_crack_setting_a} in both resolution and topology; the reference system has four outer sides, whereas the target shape has seven outer sides due to the notch.

In the inverse problem, we seek to infer the member-wise (heterogeneous) initial tangent modulus $\{E_e\}_{e\in\mathcal{E}}$.
Only the final nodal positions are assumed to be observed, and the modulus field is identified so that the resulting deformed configuration matches the target notched shape.
A central difficulty is that the final target shape alone does not directly provide suitable intermediate target shapes for the continuation procedure.
To address this issue, we first fit an affine transformation of the target shape in \cref{fig:pr3_crack_setting_b} to the reference shape in \cref{fig:pr3_crack_setting_a} by minimizing the Chamfer distance through gradient-based optimization.

\begin{figure}[htbp]
    \centering
    \begin{subfigure}{0.25\textwidth}
        \centering
        \includegraphics[width=\textwidth]{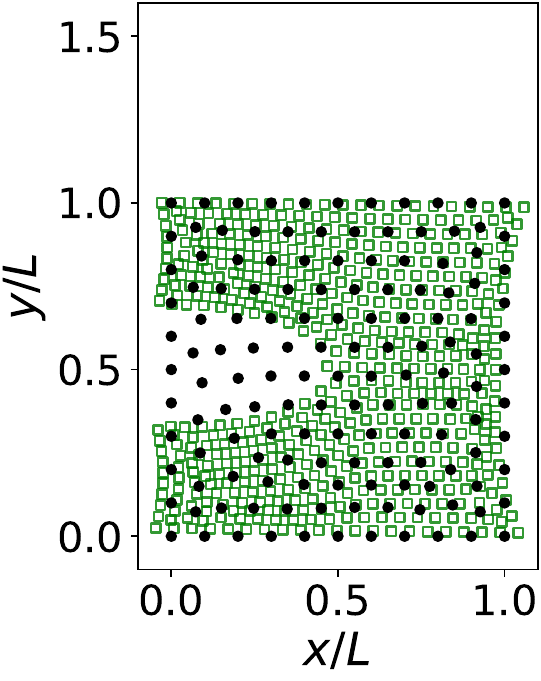}
        \caption{}
        \label{fig:pr3_crack_result_a}
    \end{subfigure}
    \hspace{0.1\textwidth}
    \begin{subfigure}{0.25\textwidth}
        \centering
        \includegraphics[width=\textwidth]{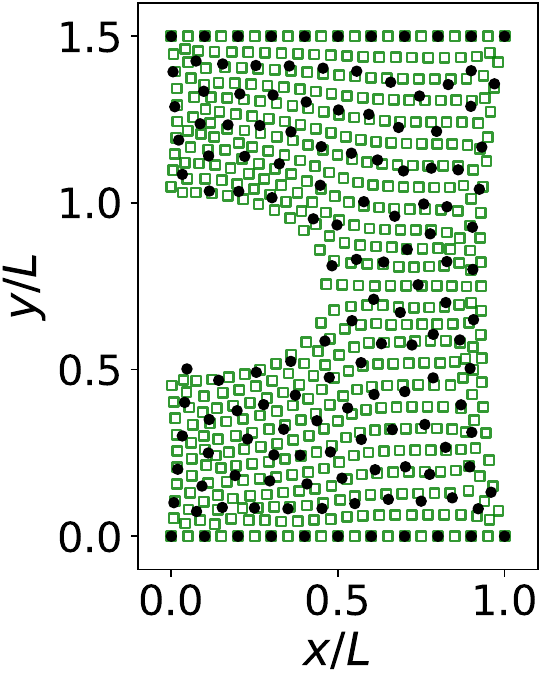}
        \caption{}
        \label{fig:pr3_crack_result_b}
    \end{subfigure}
    \hfill
    \caption{
    Inverse identification of heterogeneous stiffness of the truss system.
    (\subref{fig:pr3_crack_result_a}) Initial locations of nodes as black scatters, with an obtained initial target shape as green open squares.
    (\subref{fig:pr3_crack_result_b}) Inferred truss configuration as black scatters, aligned with the target point set in green open squares.
    }
    \label{fig:pr3_crack_result}
\end{figure}

The resulting affine-registered target point set, together with the initial hinge positions, is shown in \cref{fig:pr3_crack_result_a}.
We then construct homotopy-based guidance point clouds over multiple load steps between the affine-registered target and the final target shape.
The Chamfer-distance tolerance is prescribed separately across load steps, initialized by the initial Chamfer distance of $1.382\times10^{-3}$ and varied linearly toward the final target tolerance.
As a result, the final nodal point set closely matches the target notched shape after stretching (final Chamfer distance: $9.975\times10^{-4}$), as shown in \cref{fig:pr3_crack_result_b}.

\subsection{Thermal Shape Morphing into a Circle}\label{sec:result:circle}

In this example, we consider another form of admissibility from the perspective of manufacturing constraints on material distributions.
While the objective is to find thermoelastic material distributions that morph the reference square truss into a circular target shape, we use this setting to compare direct member-wise updates with a continuous representation of the design field and assess the induced structural bias in the optimized distribution.
We further compare the present optimization framework with MMA~\cite{svanberg1987method} as a representative choice in related structural design problems.

Here, we introduce the thermoelastic setting used to demonstrate the effect of design-variable representation on the inferred material distribution in inverse design of shape morphing.
The design variable in this example is the distribution of thermal expansion coefficients $\alpha$ over the truss members.
To account for thermal effects, we decompose the total stretch in \cref{eq:stretch} into thermal and elastic contributions, $\lambda_{e}^{\mathrm{th}}$ and $\lambda_{e}^{\mathrm{el}}$, as
\begin{equation}
\label{eq:thermoelastic_decomposition}
\lambda_e
=
\lambda_e^{\mathrm{el}}
\lambda_e^{\mathrm{th}}
, \quad 
\lambda_e^{\mathrm{th}} = 1 + \alpha_e \Delta \vartheta_e,
\end{equation}
where $\Delta \vartheta_e$ denotes the temperature change \cite{lu1975decomposition}.
Accordingly, the Neo-Hookean model in \cref{eq:neohookean} is evaluated using only the elastic stretch $\lambda_{e}^{\mathrm{el}}$, with the initial tangent modulus set to $E=1\,\mathrm{MPa}$.

We set $\Delta \vartheta = 150\,\mathrm{K}$ for each truss member and prescribe the admissible range of the thermal expansion coefficient as
\begin{equation}
    \alpha \in
    (
    4\times10^{-4}\,\mathrm{K^{-1}},
    \,4\times10^{-3}\,\mathrm{K^{-1}}
    ).
\end{equation}
We examine two parameterizations: a continuous representation of $\alpha$ based on sinusoidal representation networks (SIRENs)~\cite{sitzmann2020implicit}, and a member-wise parameterization of thermal expansion coefficients $\alpha_e$.
In the SIREN-based approach, spatial correlation is incorporated by assigning the value of $\alpha$ for each truss member from its position in the reference configuration.

\begin{figure}[htbp]
    \centering
    \begin{subfigure}{0.4\textwidth}
        \centering
        \includegraphics[width=\textwidth]{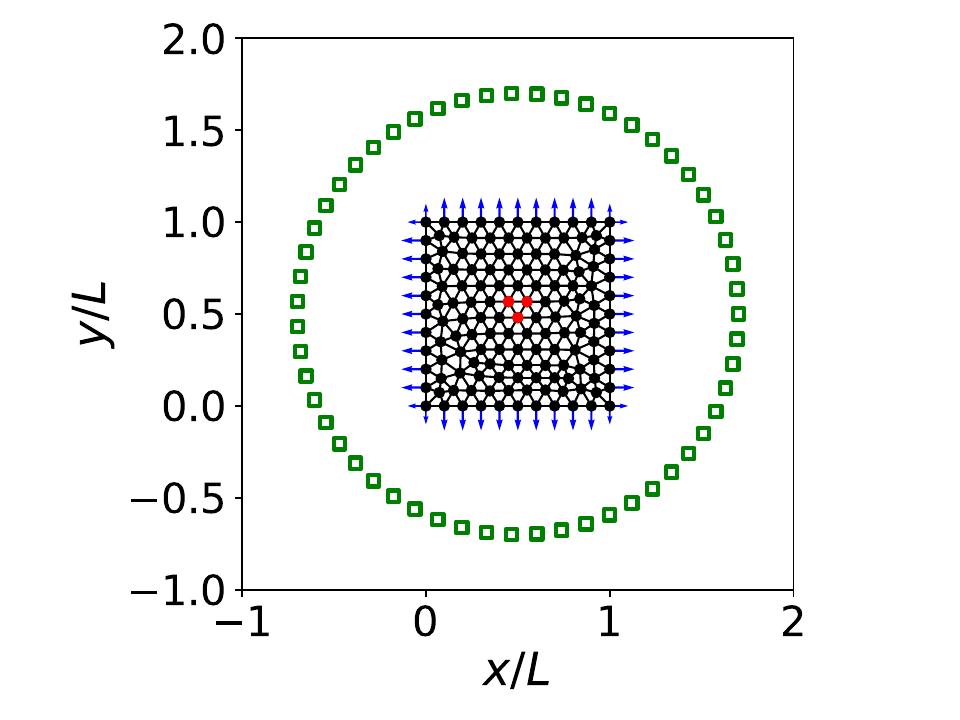}
    \end{subfigure}
    \hfill
    \caption{
        A reference configuration of the truss system with $L=1\,\text{m}$.
        The cross sectional areas of each member are assumed to be $A_e=1{,}000\,\mathrm{mm^3}/L_e$.
        The fixed supports applied as red scatters, and the applied forces on nodes as blue arrows.
        A uniform traction of $90\,\text{N/m}$ is applied to the boundary elements in the outward normal directions.
        The final target shape given as $55$ points on a circle with radius $1.2L$ shown by green hollow squares. The truss's boundary has $40$ nodes. 
    }
    \label{fig:pr4_mesh_circle}
\end{figure}

The reference truss, boundary conditions, and final target point cloud are shown in \cref{fig:pr4_mesh_circle}.
The initial target point cloud is constructed by sampling points along the square boundary with the same number of points as the final target point cloud.
The intermediate target shapes are then constructed according to \cref{eq:homotopy_construction}.
We note that as the discrepancy is measured by the Chamfer distance in \cref{eq:chamfer_distance}, explicit nodal correspondence between the truss system boundary and the target shape is not required.

\begin{figure}[htbp]
    \centering
    \begin{subfigure}{0.4\textwidth}
        \centering
        \includegraphics[width=\textwidth]{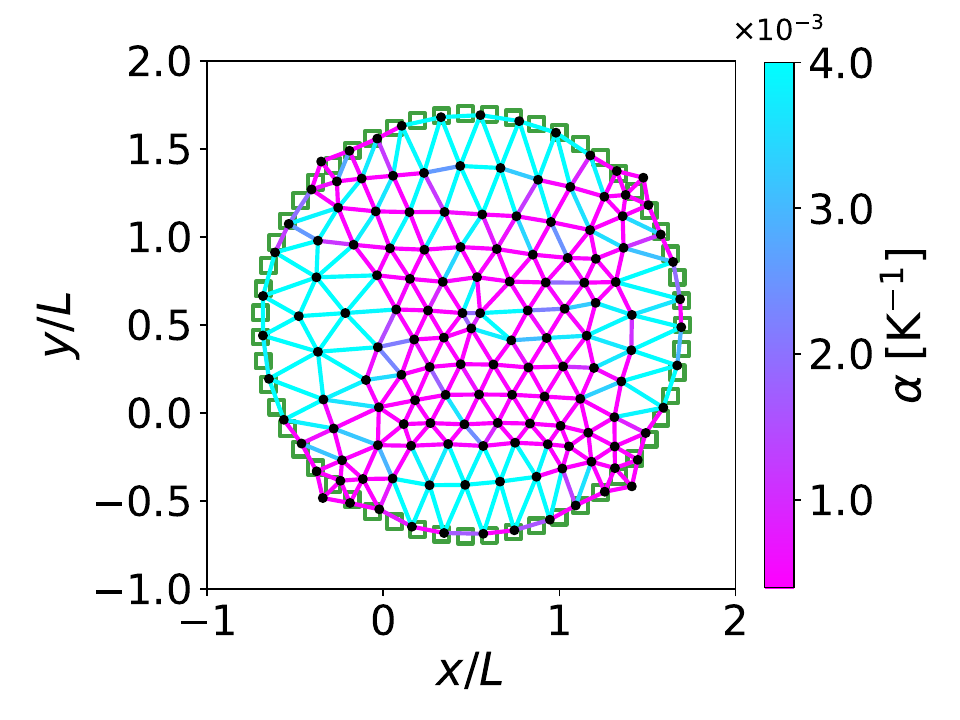}
        \caption{}
        \label{fig:pr4_circle_a}
    \end{subfigure}
    \begin{subfigure}{0.4\textwidth}
        \centering
        \includegraphics[width=\textwidth]{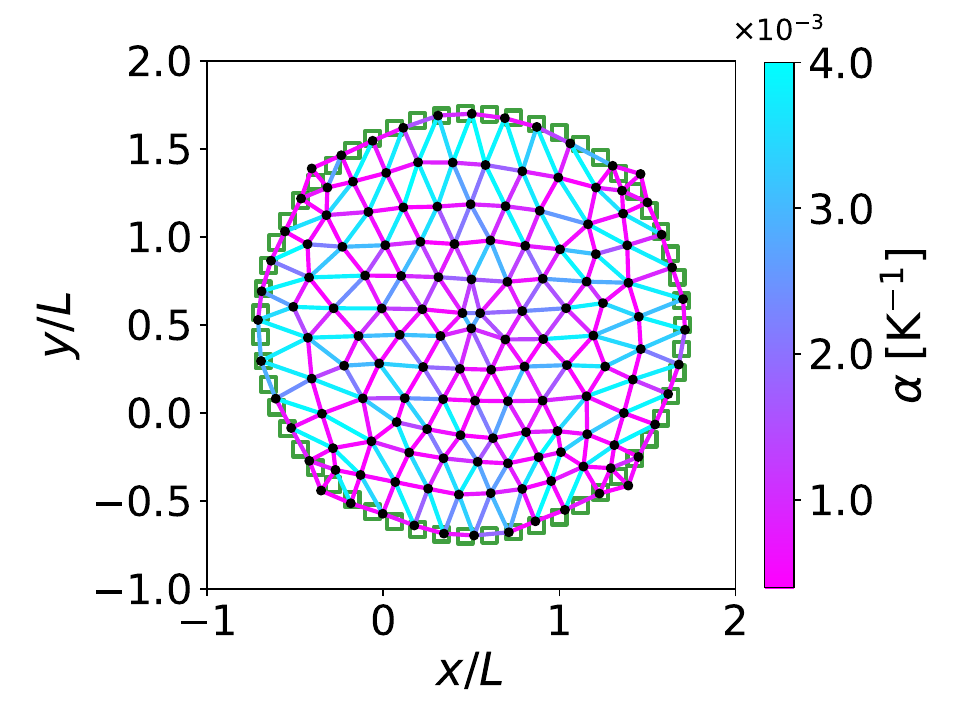}
        \caption{}
        \label{fig:pr4_circle_b}
    \end{subfigure}
    \hfill
    \caption{
    Inverse design of the distribution $\alpha$ in the Neo-Hookean truss system.
    (\subref{fig:pr4_circle_a}) The optimized truss system using SIREN, and
    (\subref{fig:pr4_circle_b}) the element-wise parametrization of $\alpha$.
    }
    \label{fig:pr4_circle}
\end{figure}

\cref{fig:pr4_circle} presents the inverse design results for the distribution of the thermal expansion coefficient $\alpha$ over the truss members, with final Chamfer distances of $2.797\times10^{-3}$ and $2.792\times10^{-3}$, respectively.
A notable difference between the SIREN-based parameterization in \cref{fig:pr4_circle_a} and the member-wise optimization in \cref{fig:pr4_circle_b} is that the former produces a visibly more spatially correlated distribution of $\alpha$.
To quantify this observation, we compute the normalized Dirichlet energy for both cases, obtaining $2.66$ and $3.72$, respectively.
The lower value supports that the continuous representation induced by SIREN yields smaller variation of $\alpha$ between adjacent members.
This behavior arises from the inductive bias of neural network parameterizations, where smooth dependence of the output on the input leads to spatially coherent variations when design variables are represented as functions of position. In contrast, member-wise optimization does not impose such structure, allowing for more abrupt variations across neighboring elements.
This tendency toward smoother designs with SIREN compared with direct member-wise optimization is consistently observed in additional trials, as reported in Appendix~\ref{apdx:continuous}.

Taken together, these results show that the prior assumption that material properties are correlated with the location of truss members can lead to a distinct family of design candidates.
More broadly, this example illustrates that highly ill-posed inverse problems can admit qualitatively different yet plausible solutions depending on the prior imposed on the unknown field.
In practical settings, this suggests that physically or structurally reasonable priors should be incorporated with care, as they may steer the optimization toward solutions that are not only admissible but also more meaningful for the intended design task.

We further assess the strengths of the proposed approach by comparing it with the MMA, which is widely used in structural topology optimization \cite{bendsoe2013topology}. We examine its behavior for the present inverse problem when the thermal expansion coefficient is parameterized using a SIREN. The details of the MMA implementation are provided in Appendix~\ref{apdx:mma}.

\begin{figure}[htbp]
    \centering
    \begin{subfigure}{0.4\textwidth}
        \centering
        \includegraphics[width=\textwidth]{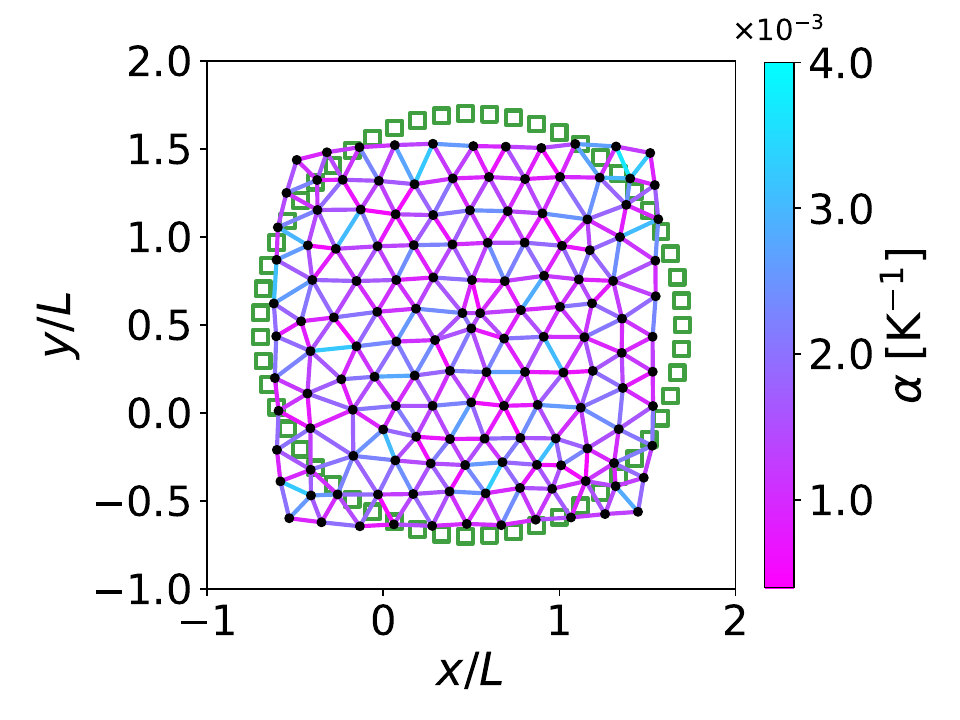}
        \caption{}\label{fig:pr4_circle_mma_a}
    \end{subfigure}
    \begin{subfigure}{0.4\textwidth}
        \centering
        \includegraphics[width=\textwidth]{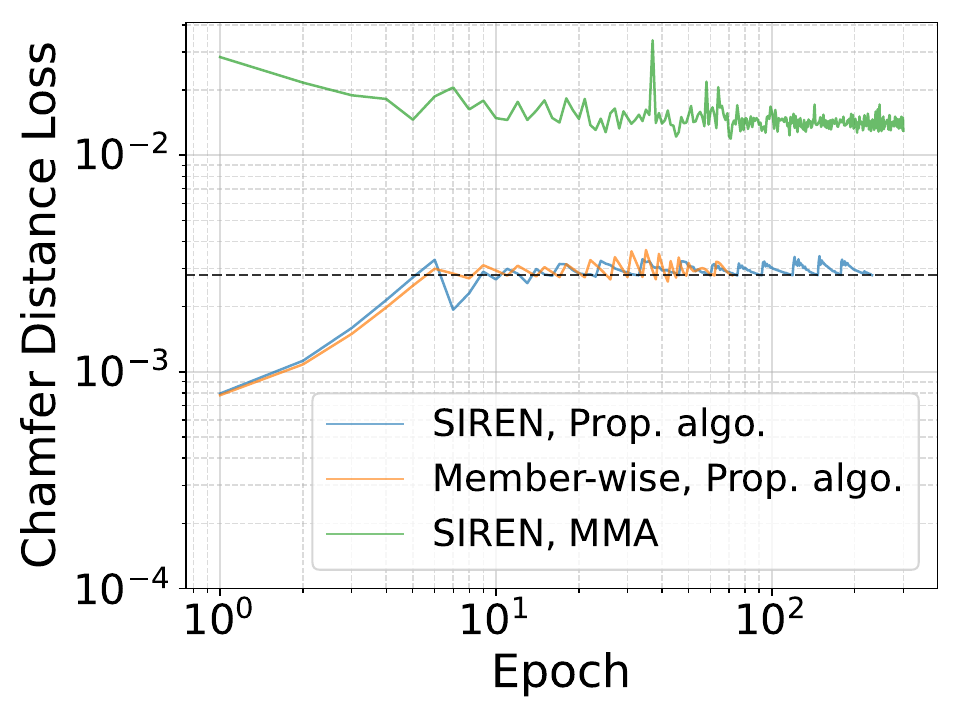}
        \caption{}\label{fig:pr4_circle_mma_b}
    \end{subfigure}
    \hfill
    \caption{
        Inverse design of the distribution of $\alpha$ in the thermoelastic truss system under a continuous representation, using the MMA as the optimizer.
        (\subref{fig:pr4_circle_mma_a}) Obtained distribution of the thermal expansion coefficient $\alpha$, parameterized by SIREN.
        (\subref{fig:pr4_circle_mma_b}) Corresponding loss history, shown together with the loss histories obtained by the proposed algorithm in \cref{fig:pr4_circle}.
    }
    \label{fig:pr4_circle_mma}
\end{figure}

The results are shown in \cref{fig:pr4_circle_mma}.
As shown in \cref{fig:pr4_circle_mma_a}, MMA does not fully recover the target shape in the present setting.
This behavior is more clearly illustrated in \cref{fig:pr4_circle_mma_b}. The proposed homotopy-continuation framework combined with the Adam optimizer successfully reaches the target configuration, while MMA was unable to recover the desired target configuration.
These results are consistent with the observations reported in \cite{sanu2025neural} in the context of topology optimization, where the MMA has been shown to be less effective for inverse problems involving the simultaneous training of neural network models, a behavior also observed in the present setting.

\subsection{Mechanical Heart Shape Morphing}\label{sec:result:heart}

In this section, we consider an inverse design problem involving a nonconvex heart-shaped target point cloud and a nonlinear data-driven constitutive prior.
Beyond reproducing the target shape through heterogeneous design, this example is used to empirically assess two optimization-related ingredients introduced in the present work: homotopy-based continuation and the smoothness of the constitutive prior.
Specifically, we ablate the continuation strategy to examine its effect on computational efficiency, and compare smooth and non-smooth constitutive priors to examine their effect on optimization performance.

\begin{figure}[htbp]
    \centering
    \begin{subfigure}{0.4\textwidth}
        \centering
        \includegraphics[width=\textwidth]{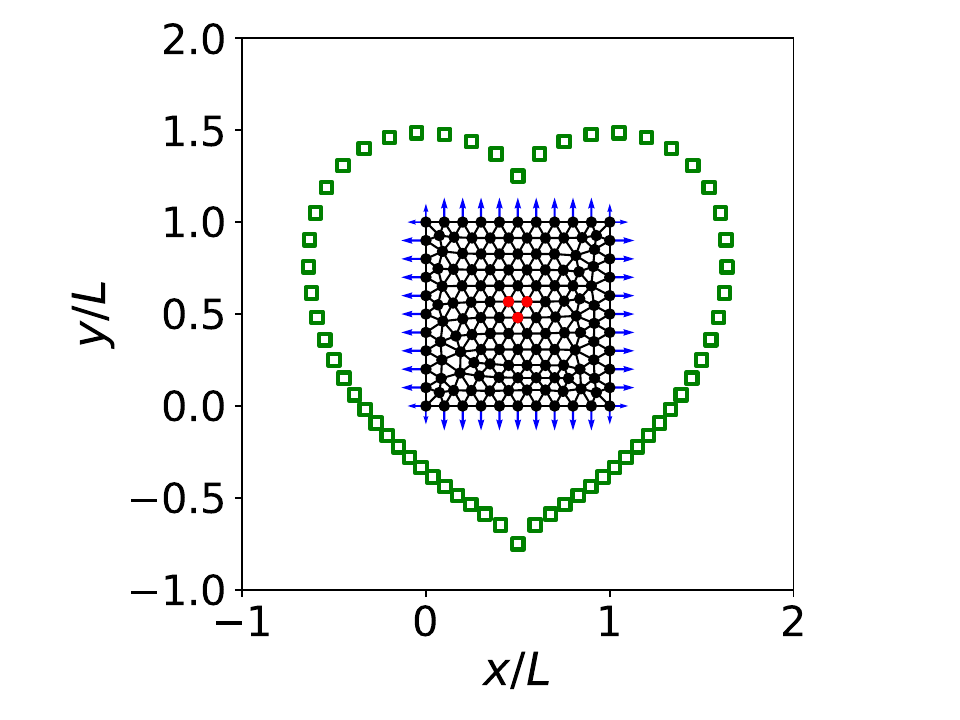}
    \end{subfigure}
    \hfill
    \caption{
        A reference configuration of the truss system with $L=1\,\text{m}$.
        The cross sectional areas of each member are assumed to be $A_e=1{,}000\,\mathrm{mm^3}/L_e$.
        The fixed supports applied as red scatters, and the applied forces on nodes as blue arrows.
        A uniform traction of $100\,\text{N/m}$ is applied to the boundary elements in the outward normal directions.
        The target shape consists of $60$ points sampled on a heart-shaped geometry, shown as green open squares. The boundary of the truss contains $40$ nodes.
    }
    \label{fig:pr3_mesh_heart}
\end{figure}

The problem setting is illustrated in \cref{fig:pr3_mesh_heart}.
We consider inverse design toward a heart-shaped target represented only by a point cloud sampling of its boundary, using constitutive prior A.
The target boundary is constructed from the implicit curve in \cite{taubin1994rasterizing},
\begin{equation}
    \left(x^2+y^2-1\right)^3 - x^2 y^3 = 0,
\end{equation}
by computing the outermost radial intersection in polar coordinates for each uniformly sampled angle.

\begin{figure}[htbp]
    \centering
    \begin{subfigure}{0.4\textwidth}
        \centering
        \includegraphics[width=\textwidth]{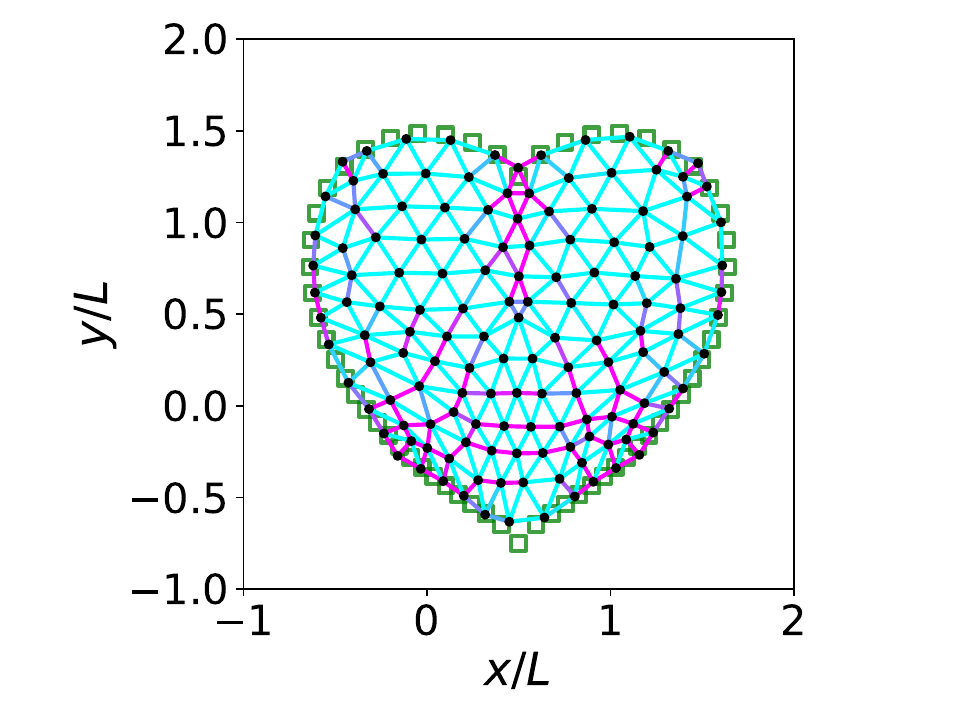}
        \caption{}
        \label{fig:pr3_heart_a}
    \end{subfigure}
    \begin{subfigure}{0.4\textwidth}
        \centering
        \includegraphics[width=\textwidth]{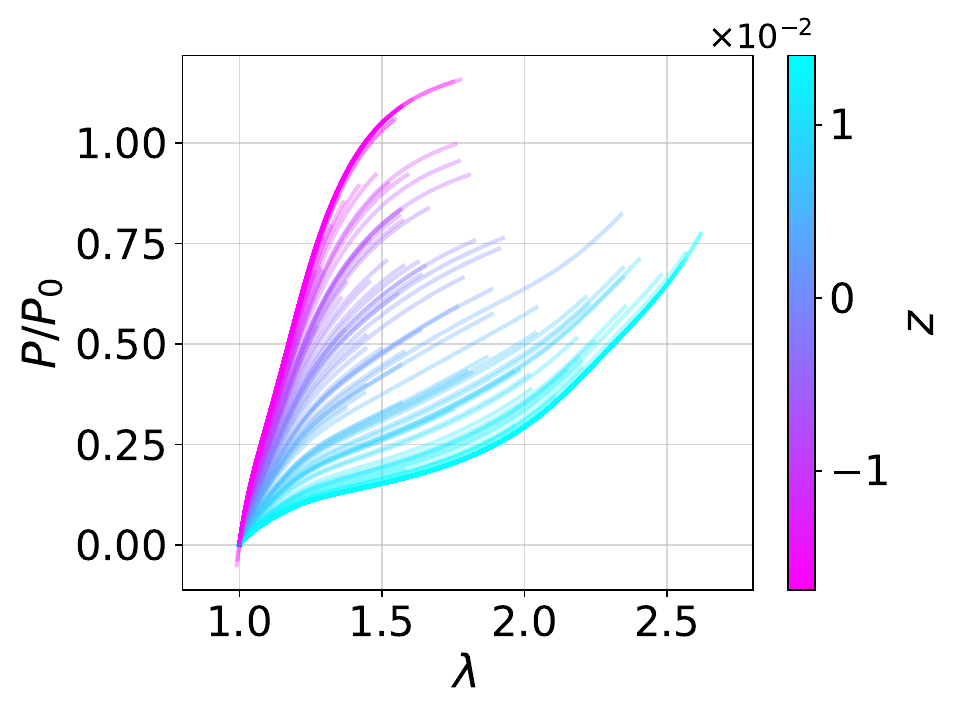}
        \caption{}
        \label{fig:pr3_heart_b}
    \end{subfigure}
    \hfill
    \caption{
    Inverse design of the material distribution in the truss system.
    (\subref{fig:pr3_heart_a}) The optimized truss system using the material surrogate, the final target shape as green scatters.
    (\subref{fig:pr3_heart_b}) The $P$--$\lambda$ curves of the inferred configuration. The colors of the bar elements in (\subref{fig:pr3_heart_a}) correspond to those shown in (\subref{fig:pr3_heart_b}).
    }
    \label{fig:pr3_heart}
\end{figure}

The results are illustrated in \cref{fig:pr3_heart}.
As shown in \cref{fig:pr3_heart_a}, despite the geometric complexity of the heart-shaped target, the proposed framework identifies a distribution of material models that effectively reproduces key features of the shape (final Chamfer distance: $2.099\times10^{-3}$).
This behavior can be examined more closely in \cref{fig:pr3_heart_b}, where corresponding elements are shown in the same color as in \cref{fig:pr3_heart_a}. 
In particular, relatively stiff properties are assigned along the upper indentation, which must recess inward, and in the region near the bottom, to match the structured V-shape in the bottom, whereas softer properties are distributed near the rounded lobes of the heart, thereby enabling a deformation pattern that differs substantially from the reference square configuration.

\begin{figure}[htbp]
    \centering
    \begin{subfigure}{0.4\textwidth}
        \centering
        \includegraphics[width=\textwidth]{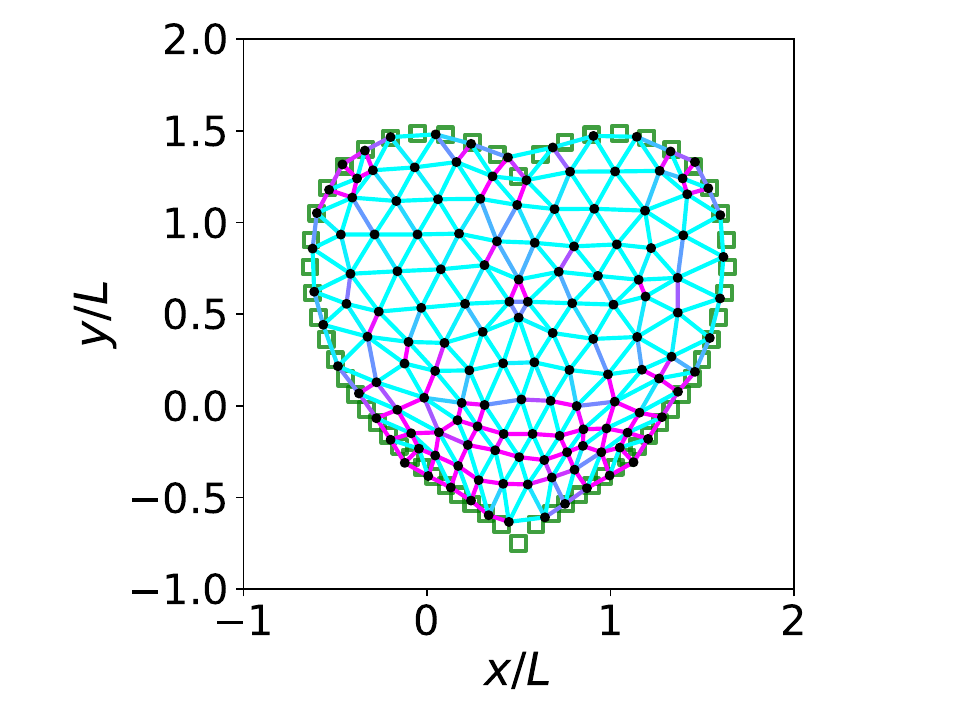}
        \caption{}
        \label{fig:pr3_heart_non_a}
    \end{subfigure}
    \begin{subfigure}{0.4\textwidth}
        \centering
        \includegraphics[width=\textwidth]{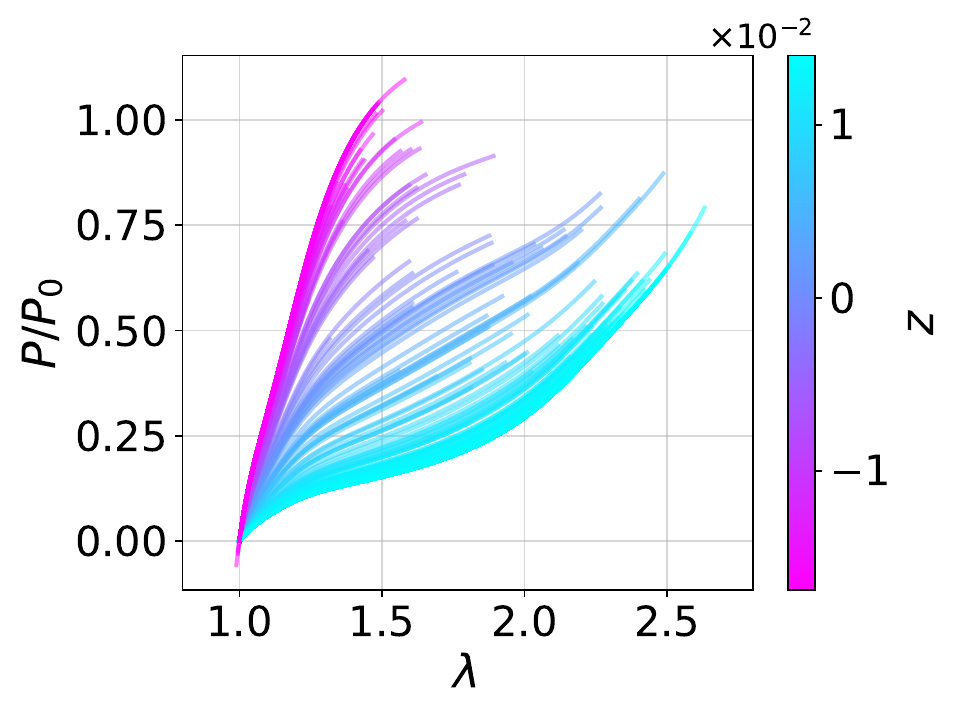}
        \caption{}
        \label{fig:pr3_heart_non_b}
    \end{subfigure}
    \hfill
    \caption{
    Inverse design of the material distribution in the truss system, resulted without the proposed homotopy-based method.
    (\subref{fig:pr3_heart_non_a}) The optimized truss system using the material surrogate, the final target shape as green scatters.
    (\subref{fig:pr3_heart_non_b}) The $P$--$\lambda$ curves of the inferred configuration.
    }
    \label{fig:pr3_heart_non}
\end{figure}

Next, we compare the proposed framework with and without the homotopy-based continuation.
The results are shown in \cref{fig:pr3_heart_non}.
A first observation is that the optimization remains reachable even without the continuation, attaining the prescribed Chamfer distance tolerance as shown in \cref{fig:pr3_heart_non_a}.
The identified material distribution is also qualitatively consistent with the homotopy-based case, with both relatively stiffer and softer members assigned to capture the heart-shaped target, as shown in \cref{fig:pr3_heart_non_b}.

However, a significant difference is observed in terms of the number of iterations.
The histories of both settings are shown in \cref{fig:pr_heart_history}.
The Chamfer distance histories in \cref{fig:pr_heart_history_a} contain comparable numbers of recorded loss values, 76 for the setting with homotopy and 71 without.
The numbers of parameter updates, however, differ: 56 and 70, respectively.
This is because the homotopy-based setting uses 20 continuation steps, whereas the non-homotopy setting uses a single target shape, and under the Chamfer distance-based stopping criterion no parameter update is performed at the converged state of each continuation step.
As a result, the number of gradient evaluations is smaller for the proposed approach in the case considered here.
A more pronounced difference is observed in the Newton--Raphson iteration count, shown in \cref{fig:pr_heart_history_b}.
The warm-started continuation structure of the proposed approach places the load-step iteration as the outer loop of the optimization, whereas the non-homotopy setting places it as the inner loop, which substantially amplifies the total number of Newton--Raphson iterations.
Overall, this comparison supports the interpretation that the present homotopy-based continuation splits the problem into a sequence of easier subproblems, yielding benefits in both the number of optimizer calls and the computational cost of the equilibrium solves.

\begin{figure}[htbp]
    \centering
    \begin{subfigure}{0.4\textwidth}
        \centering
        \includegraphics[width=\textwidth]{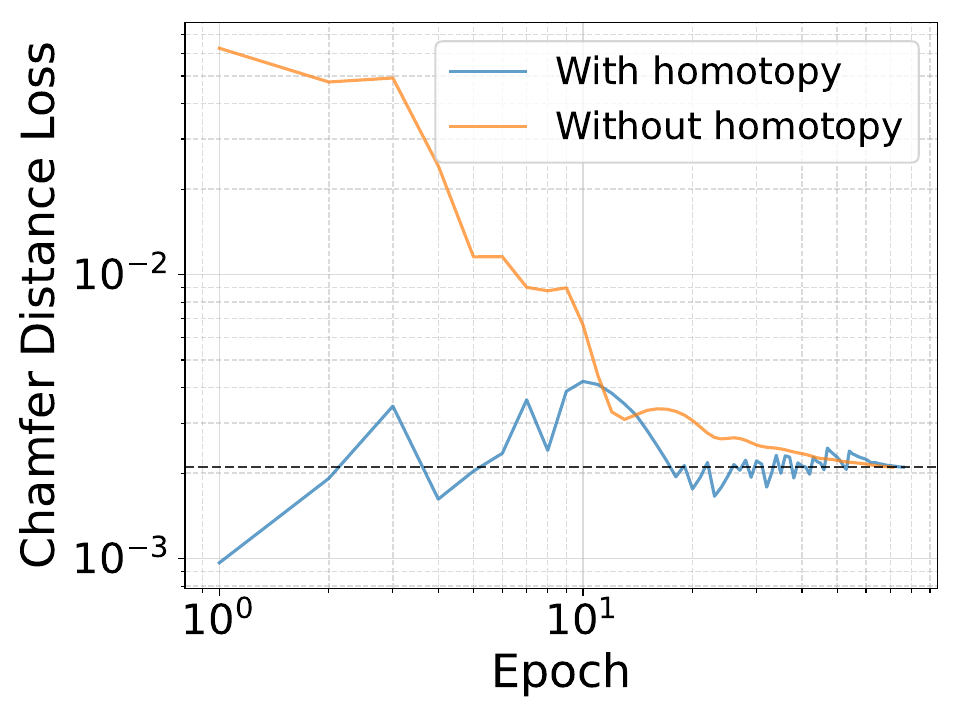}
        \caption{}
        \label{fig:pr_heart_history_a}
    \end{subfigure}
    \begin{subfigure}{0.4\textwidth}
        \centering
        \includegraphics[width=\textwidth]{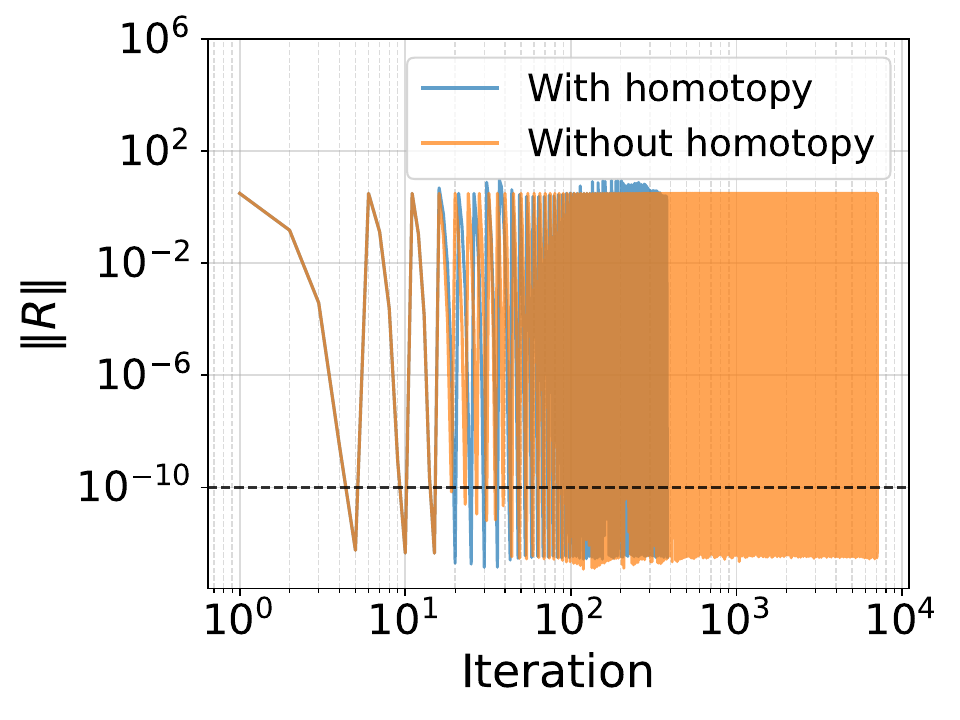}
        \caption{}
        \label{fig:pr_heart_history_b}
    \end{subfigure}
    \hfill
    \caption{
    Histories of (\subref{fig:pr_heart_history_a}) Chamfer distance and (\subref{fig:pr_heart_history_b}) residual norm for the settings with and without homotopy-based continuation.
    The tolerances for the Chamfer distance and the residual norm are shown as black dashed lines.
    The residual norm is reported in the $\ell^2$ norm.
    }
    \label{fig:pr_heart_history}
\end{figure}

\begin{figure}[htbp]
    \centering
    \begin{subfigure}{0.5\textwidth}
        \centering
        \includegraphics[width=\textwidth]{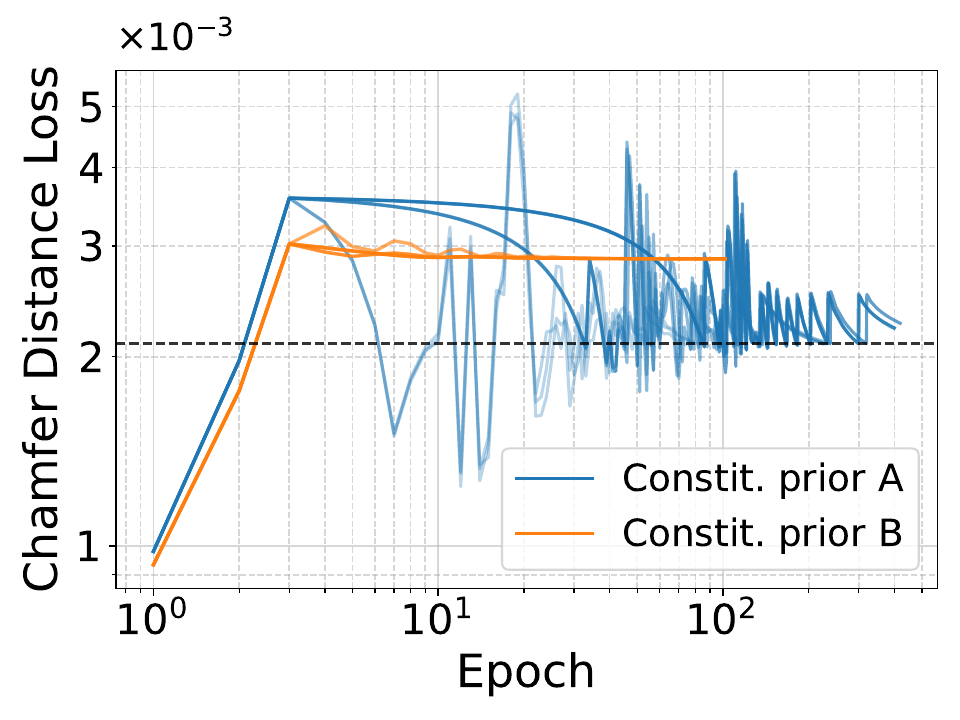}
    \end{subfigure}
    \hfill
    \caption{
    Loss histories of the inverse identification using constitutive prior A and constitutive prior B. The threshold is shown as a black dashed line; multiple trials belonging to the same learning rate are shown with the same transparency.
    }
    \label{fig:pr_heart_compare_smoothness}
\end{figure}

Similar to the study in Section~\ref{sec:result:bar:compare}, we examine the importance of a smoothness of the landscape of learned constitutive family during the inversion process.
We compare the performance of constitutive prior A, with constitutive prior B.
For each surrogate, we use the three learning rates that yielded favorable performance in \cref{fig:pr1_bar_compare}.
For each learning rate, three trials are performed.
In each trial, the optimization is terminated if the objective does not reach the prescribed threshold within $100$ iterations, indicating failure to match even the intermediate target shape. The latent variable $z$ for each truss member is initialized uniformly within $\pm 5\%$ of the latent range about its midpoint.
The results are shown in \cref{fig:pr_heart_compare_smoothness}. The proposed constitutive prior A consistently reaches the threshold (or progresses to the final load step), whereas constitutive prior B stalls at earlier stages (typically around the third load step) for the same learning rates. 
This behavior is consistent with the observations in \cref{fig:pr1_bar_compare}, supporting the conclusion that the smoother optimization landscape induced by the proposed constitutive prior is more suitable for inverse design.

\subsection{Discussion on Non-Uniqueness of the Inverse Problem}
\label{sec:nonunique}

The inverse design problems considered in this work are nonlinear, nonconvex, and defined under partial observability, and therefore are inherently non-unique. In general, multiple distinct material distributions may produce indistinguishable structural responses, resulting in a solution space that may contain several local minima as well as families of equivalent or near-equivalent global minimizers.

The introduction of a constitutive prior restricts this space by imposing a structured inductive bias over admissible material behaviors, thereby improving the tractability of the optimization. Nevertheless, even within this reduced space, multiple combinations of material assignments can yield comparable objective values, particularly in lattice systems where different configurations may lead to similar effective responses. \Cref{apdx:nonunique,apdx:continuous} provide representative examples illustrating this non-uniqueness, where multiple distinct material configurations achieve comparable objective values.

This non-uniqueness can be further mitigated by incorporating additional constraints or regularization terms, such as manufacturability requirements or spatial smoothness, which reduce the admissible design space. While these aspects are not fully explored in the present work, they represent an important direction for future research.

\section{Conclusion}
\label{sec:conclusion}

This work introduced a constitutive-prior-based inverse design framework for spatially distributing nonlinear material behaviors throughout a structure to achieve prescribed mechanical responses.
Rather than optimizing predefined material parameters within a prescribed constitutive law, the proposed formulation learns a low-dimensional latent parameterization of constitutive responses from data and uses spatially varying latent coordinates as the inverse design variables.
The framework was investigated in elastic networks, where the objective is to match the deformed configuration of the structure to a prescribed target geometry. The resulting inverse problem was formulated as a PDE-constrained optimization problem over latent constitutive coordinates constrained by a learned constitutive prior. To improve robustness in complex inverse design settings, the formulation further incorporated homotopy continuation, affine registration, and smooth latent-field parameterizations. 
The proposed optimization strategy was additionally compared with the commonly used MMA, demonstrating improved convergence behavior in terms of both the total number of optimization iterations and the attained objective values.

Across the numerical examples, the proposed framework demonstrated the ability to recover spatially varying material distributions capable of producing prescribed global deformations under increasingly challenging inverse design settings. The introduced Chamfer-distance-based objective avoided the need for pointwise nodal correspondence between the reference and target configurations, thereby enabling flexibility with respect to discretization, geometric representation, and resolution without requiring ad hoc interpolation or projection procedures. It was further shown that appropriately chosen neural network architectures can serve as effective priors within the inverse design formulation, including priors over admissible material behaviors and spatial smoothness. The former was realized through a PICNN architecture, while the latter was enforced through continuous neural representations of the latent design field with respect to spatial coordinates. The resulting smoothness of the inferred material distributions was additionally quantified through a graph-based metric.

More broadly, the present work reformulates inverse design from parameter identification within prescribed constitutive laws toward optimization within a learned admissible family of constitutive responses.
In this perspective, constitutive models are no longer treated as fixed analytical forms with unknown parameters, but instead as elements of a learned constitutive family constrained by physical principles. Such a formulation provides additional flexibility for inverse design in settings where classical constitutive assumptions may be insufficient to capture complex, heterogeneous, or spatially evolving material behavior.

Although the present work focused on elastic network systems, the proposed formulation is general and may be extended to broader classes of nonlinear and multiphysics problems. 
Future work will investigate how manufacturability and realizability
information can be incorporated into the learned constitutive representation, so that inverse design is restricted to responses that
are both mechanically admissible and physically attainable.
Other future directions include applications to continuum finite-element systems, inelastic and history-dependent constitutive behavior, damage and fracture, uncertainty-aware inverse design, and integration with experimental observations.
The proposed framework therefore provides a foundation for coupling constitutive learning, physics-constrained optimization, and inverse design in settings where admissible material behavior cannot be adequately represented through fixed parametric constitutive assumptions.

\section*{Limitations}

Despite the applicability demonstrated above, the present framework has several limitations.
First, the numerical examples are restricted to two-dimensional truss systems. This restriction is not inherent to the formulation, as neither the governing residual nor the optimization scheme references the spatial dimension explicitly, and an extension to three-dimensional systems follows by replacing the nodal coordinates with their three-dimensional counterparts.

Second, the demonstrated design variables are scalar on each truss member, namely the latent parameter $z$, the tangent modulus $E$, or the thermal expansion coefficient $\alpha$. This is not a structural restriction of the framework, since the adjoint-based update rule is agnostic to the dimension of the per-member design vector, and richer per-member parameterizations can be accommodated within the same scheme.

Third, the Chamfer distance thresholds are selected heuristically. This originates in part from the nature of the Chamfer distance, which is a discrepancy rather than a metric and may remain nonzero when two point clouds represent the same underlying geometry but differ in their sampling locations or densities.
The present values are also chosen to provide a consistent criterion across the comparisons reported in this work; in practice, one may calibrate the threshold according to the point-cloud resolution, refine the sampling, run the optimization for more iterations when optimization error is dominant, or adopt an alternative discrepancy measure.

\section*{Data Availability}
All data used in this work are either publicly available or can be reproduced from the information provided and are available from the corresponding author upon reasonable request.

\section*{Acknowledgment}
B.B. acknowledges support from the startup fund provided by Northwestern University.

\section*{Declaration on the Use of Generative AI}
The authors used ChatGPT and Claude, large language models developed by OpenAI and Anthropic, respectively, to assist with English language editing and grammar refinement.
The scientific content, technical interpretations, and conclusions are solely the responsibility of the authors.

\bibliographystyle{plainnat}
\bibliography{bibliography}

\appendix
\section*{Appendix}

\section{Partially Input-Convex Neural Networks}
\label{apdx:icnn}

In this section, we briefly review the partially input-convex neural networks (PICNNs) architecture introduced in \cite{amos2017input}. The PICNN produces a scalar-valued output that is convex with respect to a subset of the input dimensions (e.g., $\boldsymbol{x}$), while remaining unrestricted with respect to the remaining inputs (e.g., $\boldsymbol{z}$).
Specifically,
\begin{equation}
    g_{\boldsymbol{\phi}}
    :
    \mathcal{X}\times\mathcal{Z}
    \to
    \mathbb{R},
    \qquad
    (
    \boldsymbol{x},
    \boldsymbol{z}
    )
    \mapsto
    g_{\boldsymbol{\phi}}
    (
    \boldsymbol{x};
    \boldsymbol{z}
    )
\end{equation}
is convex in $\boldsymbol{x}\in\mathcal{X}$ for every fixed $\boldsymbol{z}\in\mathcal{Z}$.

The architecture maintains two parallel hidden states. Let $\boldsymbol{q}_n$ and $\boldsymbol{p}_n$ denote the $\boldsymbol{z}$-path and $\boldsymbol{x}$-path hidden states at layer $n$, initialized as $\boldsymbol{q}_0=\boldsymbol{z}$ and $\boldsymbol{p}_0=\boldsymbol{x}$.
The recursions read
\begin{align}
    \boldsymbol{q}_{n+1}
    &=
    \tilde{h}
    \bigl(
        \tilde{m}_n(\boldsymbol{q}_n)
    \bigr),
    \\
    \boldsymbol{p}_{n+1}
    &=
    h
    \bigl(
        \boldsymbol{w}^{(p)}_n
        (
            \boldsymbol{p}_n \odot
            (
            m^{(pq)}_{n}
            (
                \boldsymbol{q}_n
            )
            )_{+}
        )
        +
        \boldsymbol{w}^{(x)}_n
        \bigl(
            \boldsymbol{x} \odot
            (
                m^{(x q)}_n
                (
                    \boldsymbol{q}_n
                )
            )
        \bigr)
        +
        m^{(q)}_{n}(
            \boldsymbol{q}_n
        )
    \bigr),
\end{align}
where $\odot$ denotes the Hadamard (entry-wise) product, and the positive gate operation $(\,\cdot\,)_+$ is applied entry-wise.

All $m_n$ are arbitrary affine maps constituted of weights and biases with appropriate sizes, and all weight matrices and bias vectors are collected in the trainable parameters $\boldsymbol{\phi}$.
The scalar output $g_{\boldsymbol{\phi}}(\boldsymbol{x};\boldsymbol{z})$ is read off from a terminal layer of the same form with output dimension one.

Convexity of the output in $\boldsymbol{x}$ at fixed $\boldsymbol{z}$ is obtained by combining three structural conditions:
\begin{itemize}
    \item the activation $h$ on the $\boldsymbol{x}$-path is convex and non-decreasing
    \item the entries of the weight matrices $\boldsymbol{w}^{(p)}_n$ are nonnegative.
    \item the positive gate term is nonnegative.
\end{itemize}

Under these conditions, an inductive argument establishes that each component of $\boldsymbol{p}_n$ is convex in $\boldsymbol{x}$. The base case $\boldsymbol{p}_0=\boldsymbol{x}$ is linear, and hence convex. For the inductive step, the argument of $h$ decomposes into three contributions: a nonnegatively scaled Hadamard product between the convex state $\boldsymbol{p}_n$ and a nonnegative gate, a linear term in $\boldsymbol{x}$ with $\boldsymbol{z}$-dependent coefficients, and a term independent of $\boldsymbol{x}$.
Each contribution is convex in $\boldsymbol{x}$, and composition with the convex non-decreasing $h$ preserves convexity.

By contrast, no convexity restriction is required on the maps that generate $\boldsymbol{q}_n$, as these enter the $\boldsymbol{x}$-path only as coefficients evaluated at fixed $\boldsymbol{z}$.
Accordingly, the remainder can be chosen arbitrarily.
This permits the architecture to represent nonconvex modulation of the strain energy by the latent material parameter while enforcing convexity in $\boldsymbol{x}$.

\section{Hyperparameter Settings}\label{apdx:optimization}

This section summarizes the hyperparameters used across the problems presented in Section~\ref{sec:result}, including neural network architectures, optimization parameters, parameter bounds, and initialization strategies.

\paragraph{Neural network model configurations}

In this work, we use two neural networks: one for the material surrogate and the other for the continuous representation of the material property.

For the data-driven material surrogate in Section~\ref{sec:result}, we use a PICNN architecture.
The network consists of two hidden layers with $50$ units in each of the $\boldsymbol{z}$- and $\boldsymbol{i}$-paths, followed by a scalar linear output layer.
Softplus activations,
\begin{equation}
    \mathrm{softplus}(x)
    =
    \frac{1}{\beta}
    \ln{(1+e^{\beta x})},
\end{equation}
with $\beta=10$, are used in the hidden layers.
Convexity with respect to $\lambda$ is enforced by restricting the weights and gate operations associated with the $\boldsymbol{i}$-path recurrence to be nonnegative, implemented through a softplus reparameterization with the same $\beta$.
The surrogate is trained using the objective in \cref{eq:surrogate_loss} with regularization parameter $s=1$ for $10{,}000$ epochs and learning rate $1\times10^{-4}$.
The latent variables $z$ are initialized at zero and updated with learning rate $1\times10^{-3}$.
The non-smooth counterpart shown in \cref{fig:pr1_bar_fixed_latent_b} is trained under the same setting until it reaches the same data-misfit level as constitutive prior A in \cref{fig:pr1_deepSDF_d}.

For the continuous material-property representation in Section~\ref{sec:result:circle}, we use sinusoidal representation networks (SIRENs) \cite{sitzmann2020implicit} with two hidden layers of width 50 and a scalar output. We choose $w_0=20$ for the shown cases.
The network takes the two-dimensional spatial coordinate as input and is trained with learning rate $5\times10^{-3}$.

\paragraph{Inverse identifications} The detailed hyperparameter settings for each problem in Section~\ref{sec:result} are summarized in Table~\ref{tab:optimization}.
In the table, each case is identified by the representative boldface label used in the main text; for example, the example in Section~\ref{sec:result:bar} is denoted by \textbf{1D Bar}.
The tolerance values for the Chamfer distance are chosen empirically, as they work well for the problem settings considered in this study.

\begin{table*}[htbp]
\centering
\caption{Optimization settings for each problem.}
\label{tab:optimization}
\begin{tabular}{l c c c c c}
\hline
\textbf{Setting} & \textbf{1D Bar} & \textbf{Airfoil} & \textbf{Fracture} & \textbf{Circle} & \textbf{Heart} \\
\hline
Learning rate & $1\times10^{-4}$ & $1\times10^{-2}$ & $1\times10^{-1}$ & $1\times10^{-1}$ & $1\times10^{-2}$ \\
Tolerance     & $1.0\times10^{-6}$ & $3.0\times10^{-6}$ & $1.0\times10^{-3}$ & $2.8\times10^{-3}$ & $2.1\times10^{-3}$ \\
Load steps    & $20$ & $20$ & $10$ & $20$ & $20$ \\
\hline
\end{tabular}
\end{table*}

\paragraph{Affine registration}
The initial target shape in Section~\ref{sec:result:crack} is obtained by optimizing the entries of a linear transformation $\boldsymbol{A}\in\mathbb{R}^{2\times2}$ and a translation vector $\boldsymbol{b}\in\mathbb{R}^2$.
The transformation is initialized with $\boldsymbol{A}$ equal to the identity matrix and $\boldsymbol{b}$ equal to the zero vector.
The optimization is performed using a learning rate of $1\times10^{-3}$ for $3{,}000$ epochs.

\paragraph{Parameter bounds}
Suitable bounding strategies are additionally adopted for each problem.
For the problems in Sections~\ref{sec:result:bar}, \ref{sec:result:airfoil}, and \ref{sec:result:heart}, the design variable $z$ is clamped after each iteration so that it satisfies the bound $[z_{\mathrm{min}},z_{\mathrm{max}}]$, since the truss system is designed with respect to the material templates shown in \cref{fig:pr1_deepSDF_a}.
For the fracture problem in Section~\ref{sec:result:crack}, the tangent modulus of each truss member is updated indirectly by introducing an unconstrained variable $\tilde{E}_e\in\mathbb{R}$ such that
\begin{equation}\label{eq:tangent_modulus_param}
    E_e = E_0\exp(\tilde{E}_e),
\end{equation}
which enforces the constraint $E_e>0$, where $E_0=1\,\mathrm{MPa}$ is chosen as the reference initial tangent modulus.
In Section~\ref{sec:result:circle}, for both the SIREN-based and member-wise $\alpha_e$ demonstrations, the bound $(\alpha_{\mathrm{min}},\alpha_{\mathrm{max}})$ is enforced through the parametrization
\begin{equation}
    \alpha_e
    =
    \mathrm{sigmoid}(\tilde{\alpha}_e)
    \bigl(\alpha_{\mathrm{max}}-\alpha_{\mathrm{min}}\bigr)
    + \alpha_{\mathrm{min}},
\end{equation}
where $\tilde{\alpha}_e\in\mathbb{R}$ denotes either the SIREN output or the directly optimized unconstrained variable, respectively.

\paragraph{Initial values}
The initial values of the design variables not specified in Section~\ref{sec:result} are set as follows.
The latent variables $z$ in Section~\ref{sec:result:airfoil} are initialized to zero.
The indirect variables $\tilde{E}$ associated with the tangent modulus $E$ in Section~\ref{sec:result:crack} are initialized to zero.
The indirect variables $\tilde{\alpha}$ associated with the thermal expansion coefficient $\alpha$ in Section~\ref{sec:result:circle} are initialized from $\mathcal{N}(0,1)$.
The latent variables $z$ in Section~\ref{sec:result:heart} are initialized to zero.

\section{Additional Trials on Non-Uniqueness}
\label{apdx:nonunique}

This supplementary example illustrates the non-uniqueness inherent in the inverse design problem considered in this work: given a target shape, one seeks a material configuration that reproduces it.
As a representative case, we consider the fracture example in Section~\ref{sec:result:crack}.
To demonstrate this non-uniqueness, we run ten trials with different initial guesses, where the latent variables $\tilde{E}_e$ are initially sampled from $\mathcal{N}(0,1)$ (see \cref{eq:tangent_modulus_param}).

\begin{figure}[htbp]
    \centering
    \begin{subfigure}{0.4\textwidth}
        \centering
        \includegraphics[width=\textwidth]{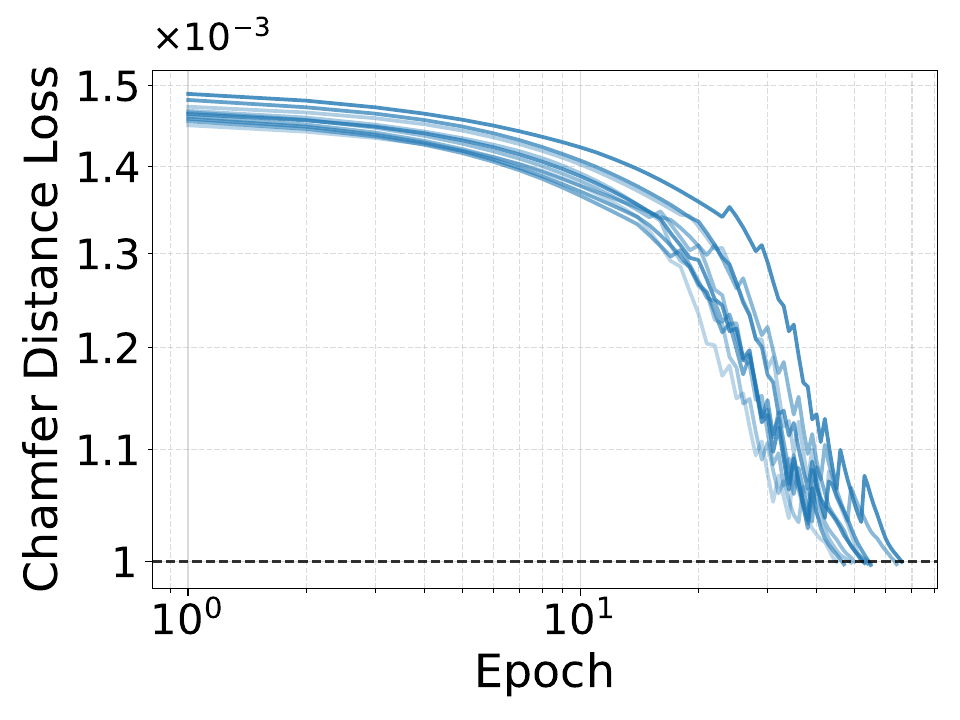}
        \label{fig:apdx_nonuniq_d}
    \end{subfigure}
    \caption{
    Loss histories over multiple trials of the problem in Section~\ref{sec:result:crack}.
    The tolerance in Chamfer distance for the last step is shown as a black dashed line.
    }
    \label{fig:apdx_nonuniq_crack_loss}
\end{figure}

\begin{figure}[htbp]
    \centering
    \begin{subfigure}{0.27\textwidth}
        \centering
        \includegraphics[width=\textwidth]{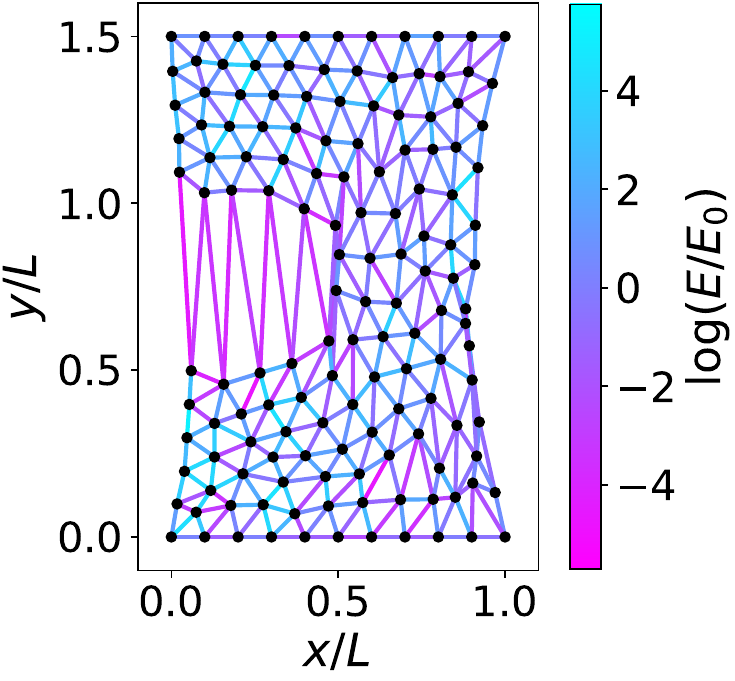}
    \end{subfigure}
    \begin{subfigure}{0.27\textwidth}
        \centering
        \includegraphics[width=\textwidth]{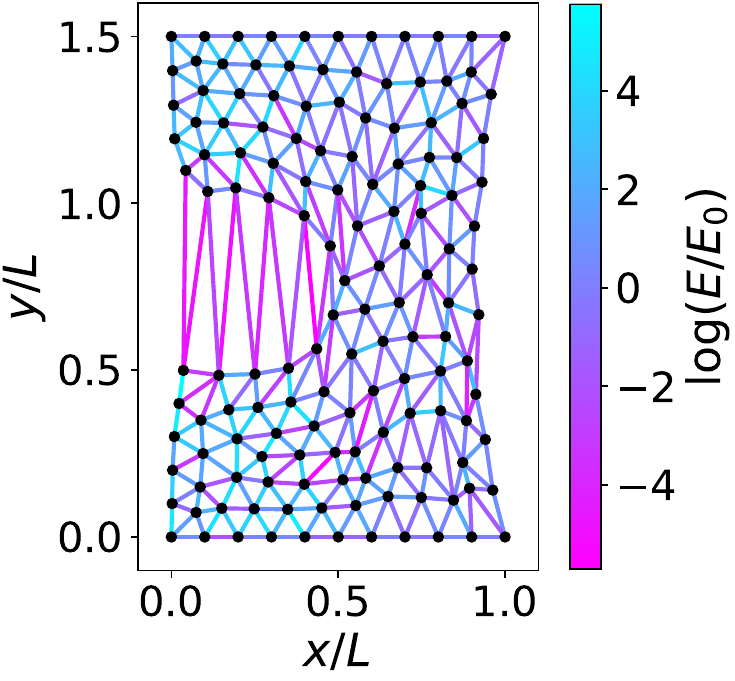}
    \end{subfigure}
    \begin{subfigure}{0.27\textwidth}
        \centering
        \includegraphics[width=\textwidth]{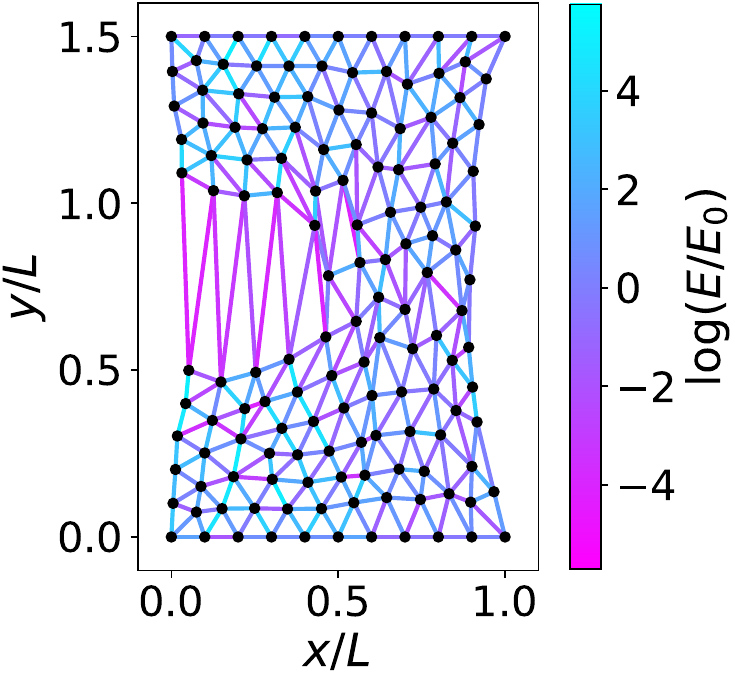}
    \end{subfigure}
    \caption{
    Current configurations at the final load step, obtained from three different initializations of the truss-member material properties.
    }
    \label{fig:apdx_nonuniq}
\end{figure}

The loss histories are shown in \cref{fig:apdx_nonuniq_crack_loss}.
As shown in the figure, all trials reach the prescribed threshold in Table~\ref{tab:optimization}.
This convergence, however, does not imply uniqueness of the inferred material distribution, as different initial guesses may still lead to distinct material fields with comparable final objective values.

\cref{fig:apdx_nonuniq} shows three representative configurations selected from the ten inferred solutions.
Although the trials converge to noticeably different designs, they attain comparable final Chamfer distance values.
This result highlights the intrinsic non-uniqueness of the inverse design problem and suggests that such non-uniqueness may provide additional flexibility for incorporating extra design requirements, such as manufacturability constraints associated with prescribed material templates.

Next, the differences among the inferred designs are quantified. Each design is interpreted as a distribution of design variables over the reference configuration and represented by a point set of triplets
\begin{equation}
\left(
x/L,\;
y/L,\;
\log(E/E_0)
\right),
\end{equation}
where $(x/L,y/L)$ denotes the midpoint coordinates of each element in the reference configuration, and $\log(E/E_0)$ denotes the corresponding normalized material property.

\begin{figure}[htbp]
    \centering
    \begin{subfigure}{0.4\textwidth}
        \centering
        \includegraphics[width=\textwidth]{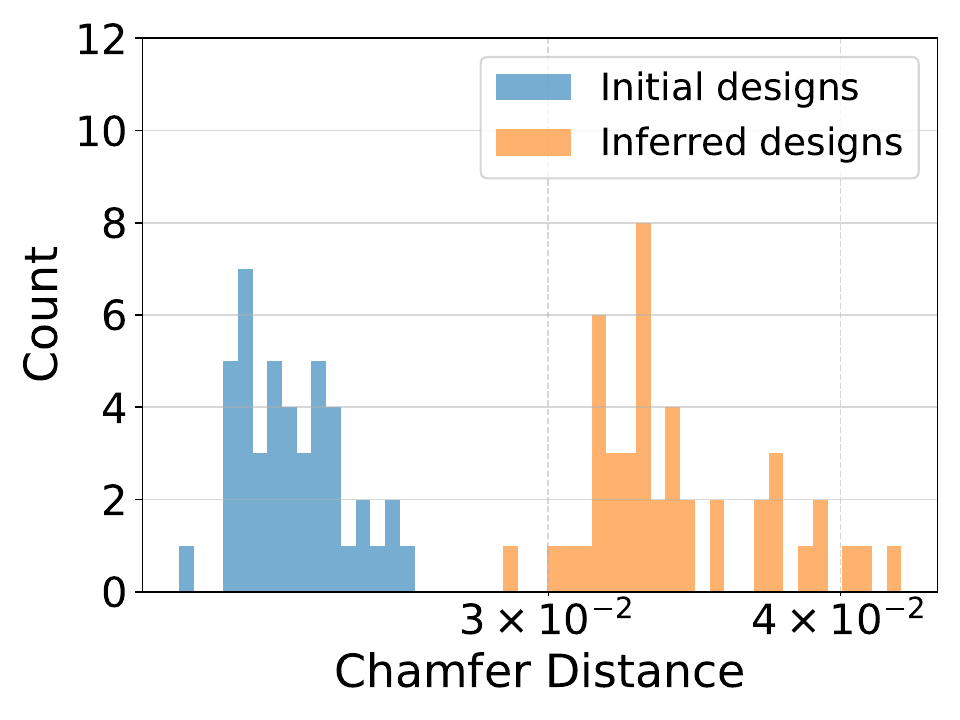}
    \end{subfigure}
    \caption{Histograms of the pairwise Chamfer distances between designs. The ten independent trials yield 45 design pairs for both the initial and inferred designs.}
    \label{fig:apdx_nonunique_design_hist}
\end{figure}

The differences between designs are quantified using the Chamfer distance in \cref{eq:chamfer_distance} between their corresponding point sets. Pairwise Chamfer distances are evaluated separately among the initial designs sampled from the isotropic Gaussian distribution and among the inferred designs. This choice is not unique, as alternative discrepancy measures or component scalings could also be adopted. It is used here because it is consistent with interpreting each design as a spatial distribution of material response, and because the isotropic Gaussian initialization provides an intuitive reference for assessing changes in design diversity.

The resulting histograms are shown in \cref{fig:apdx_nonunique_design_hist}. The inferred designs exhibit a shift toward larger pairwise Chamfer distances relative to the initial designs, indicating increased diversity in the spatial material distributions. Together with the results in \cref{fig:apdx_nonuniq_crack_loss}, this observation suggests that distinct inferred designs can produce closely matching final configurations that satisfy the prescribed objective threshold.

\section{Additional Trials on Design-Variable Representation}\label{apdx:continuous}

In this section, we report supplementary results comparing the SIREN-based continuous representation approach with the member-wise optimization approach.
Our objective is to characterize their training behavior and to examine how structured prior knowledge influences the optimization outcome.
For both approaches, each trial is run until the Chamfer-distance-based stopping criterion is satisfied.
We repeat the optimization until three successful cases are obtained for each method.

\begin{figure}[htbp]
    \centering
    \begin{subfigure}{0.4\textwidth}
        \centering
        \includegraphics[width=\textwidth]{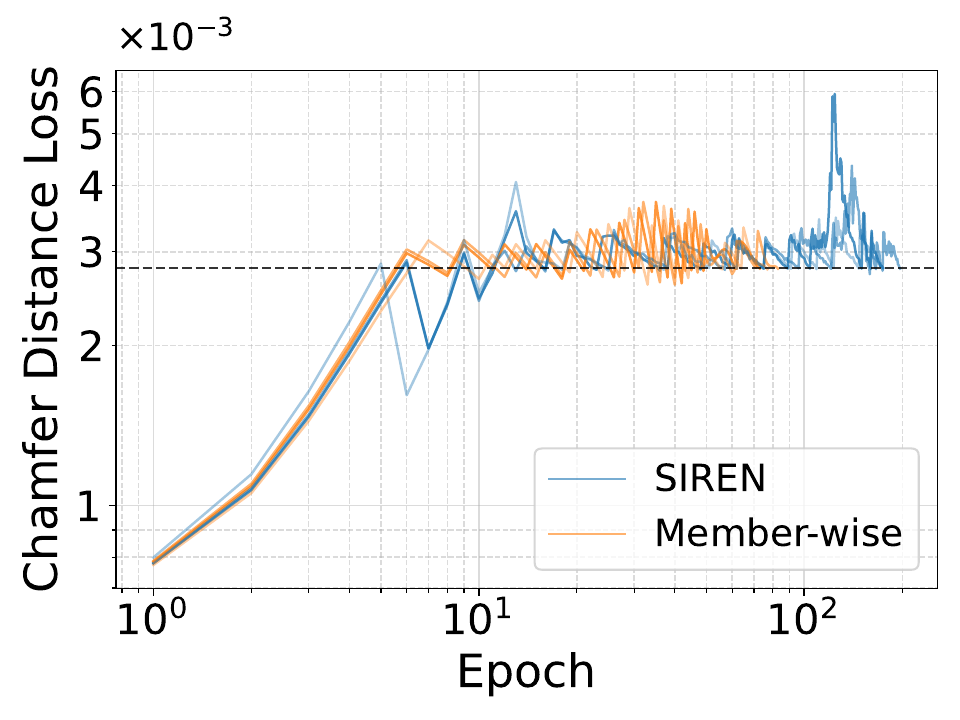}
    \end{subfigure}
    \caption{
    Loss histories over multiple trials of the problem in Section~\ref{sec:result:circle}.
    The tolerance in Chamfer distance for each step is shown as a black dashed line.
    }
    \label{fig:apdx_circle_loss}
\end{figure}

The loss histories for the two approaches are shown in \cref{fig:apdx_circle_loss}.
As shown in the figure, the SIREN-based continuous representation exhibits substantially noisier optimization behavior than the approach that directly updates the indirect variables $\alpha_e$ on a member-wise basis.
In particular, the SIREN-based losses often display pronounced spikes, with the loss not decreasing monotonically and in some cases increasing sharply before decreasing again.
This behavior suggests that the hard-coded spatial correlation and the more complex dependence of the design on the SIREN weights make the optimization problem more difficult than the member-wise parameterization.

\begin{figure}[htbp]
    \centering
    \begin{subfigure}{0.32\textwidth}
        \centering
        \includegraphics[width=\textwidth]{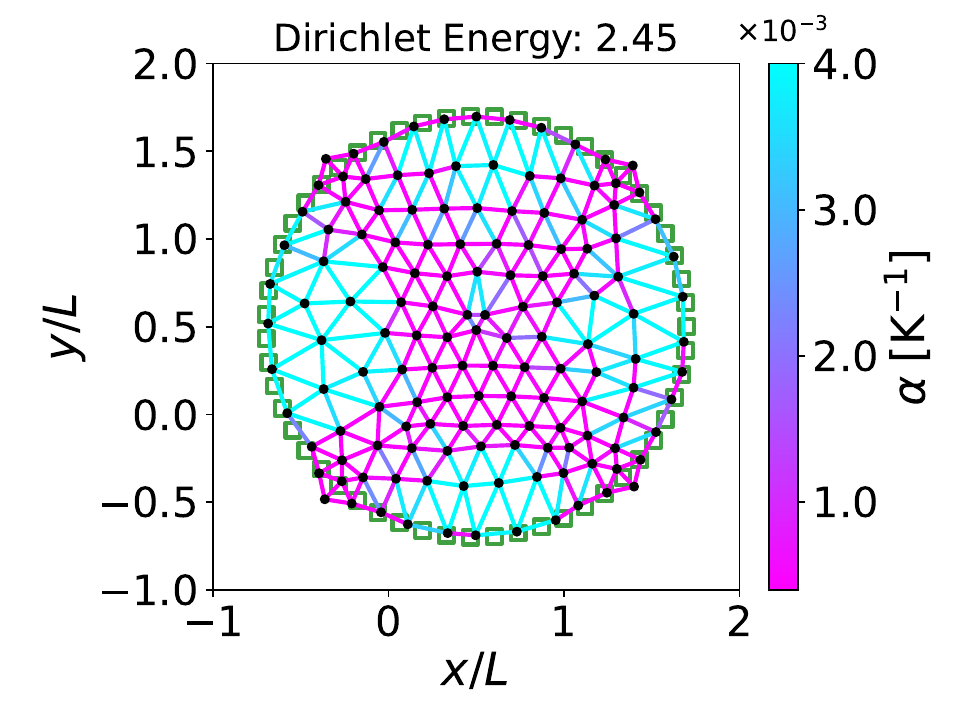}
        \label{fig:apdx_circle_a}
    \end{subfigure}
    \begin{subfigure}{0.32\textwidth}
        \centering
        \includegraphics[width=\textwidth]{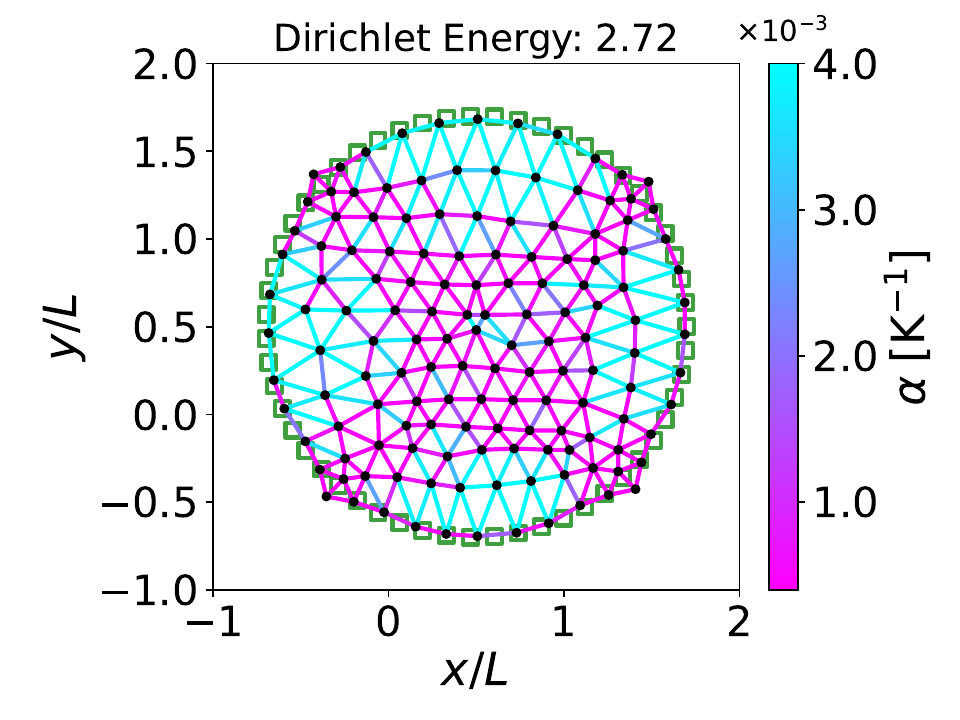}
        \label{fig:apdx_circle_b}
    \end{subfigure}
    \begin{subfigure}{0.32\textwidth}
        \centering
        \includegraphics[width=\textwidth]{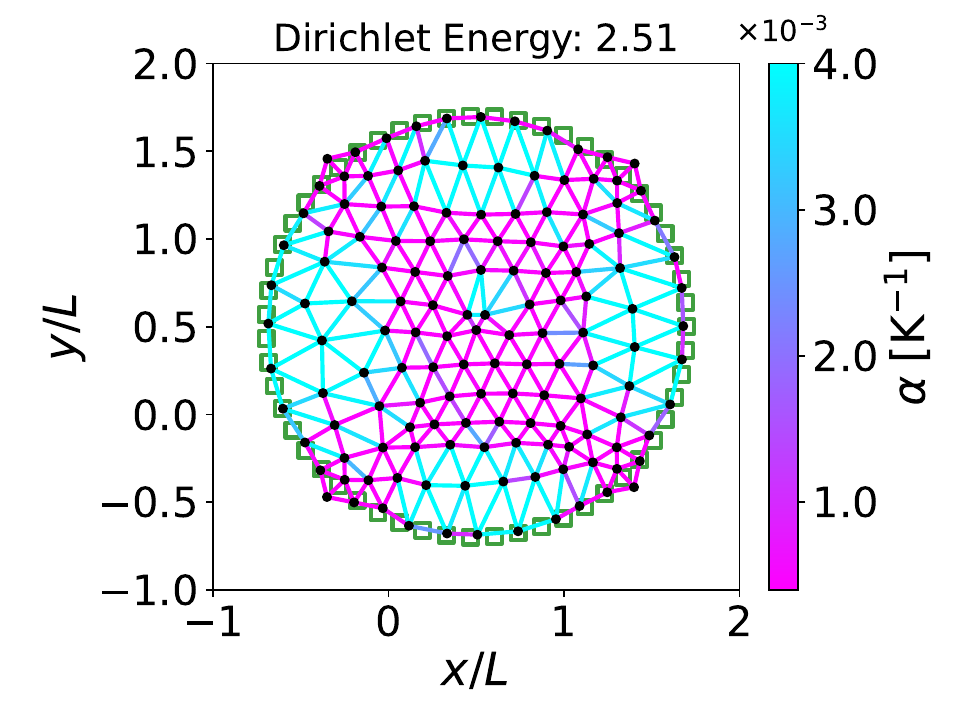}
        \label{fig:apdx_circle_c}
    \end{subfigure}
    \begin{subfigure}{0.32\textwidth}
        \centering
        \includegraphics[width=\textwidth]{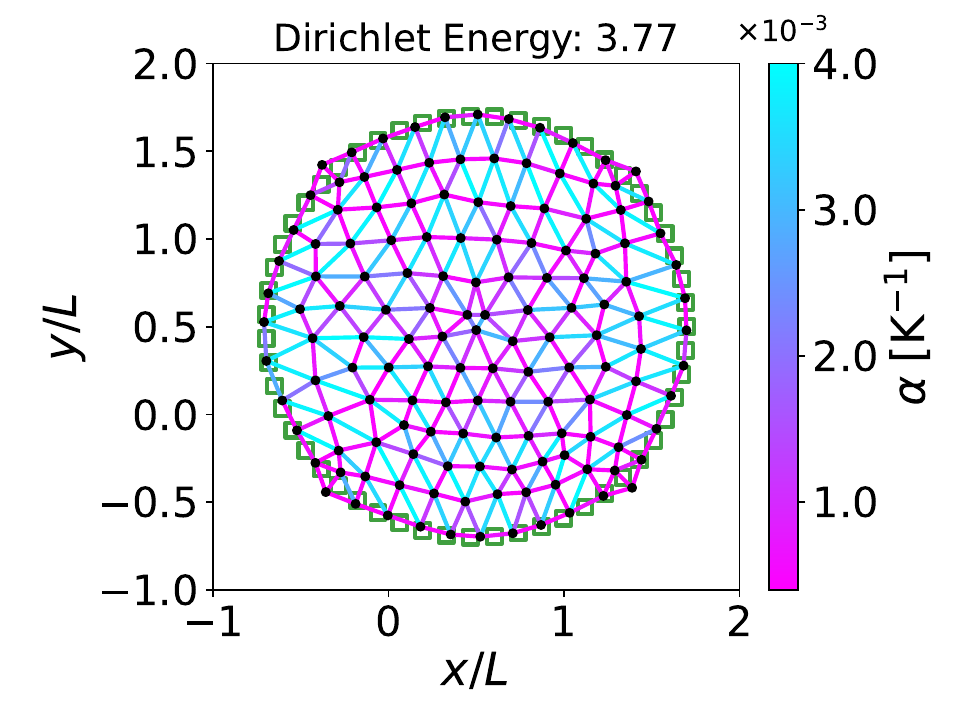}
        \label{fig:apdx_circle_d}
    \end{subfigure}
    \begin{subfigure}{0.32\textwidth}
        \centering
        \includegraphics[width=\textwidth]{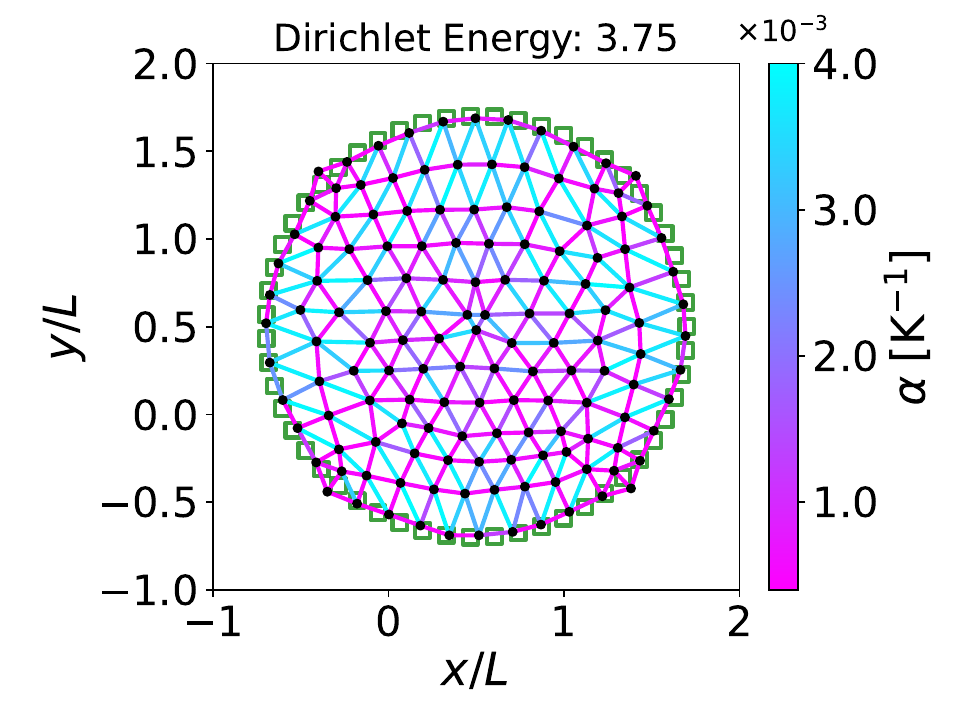}
        \label{fig:apdx_circle_e}
    \end{subfigure}
    \begin{subfigure}{0.32\textwidth}
        \centering
        \includegraphics[width=\textwidth]{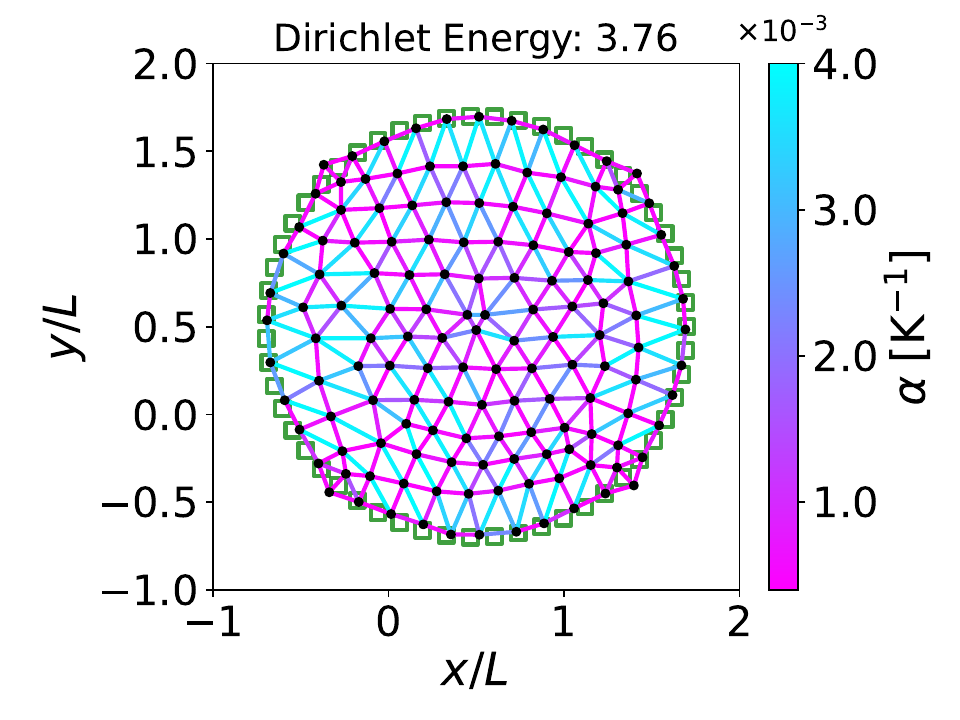}
        \label{fig:apdx_circle_f}
    \end{subfigure}
    \caption{
    Additional multiple trials of the problem in Section~\ref{sec:result:circle}.
    All results are obtained under the same setting as in that section.
    Results of (top) SIREN-based approach obtained from different random initializations of the network weights, and
    (bottom) member-wise approach obtained from random initialization of the indirect variables, $\tilde{\alpha}_e\sim\mathcal{N}(0,1)$. 
    The graph Dirichlet energy values are shown in the title of each panel.
    }
    \label{fig:apdx_circle}
\end{figure}

\cref{fig:apdx_circle} presents the corresponding optimized designs for those multiple trials.
The results show that the SIREN-based representation yields substantially more spatially correlated configurations than the member-wise approach; in particular, even the largest graph Dirichlet energy among the SIREN-based results is smaller than the smallest graph Dirichlet energy among the member-wise results.

\section{Method of Moving Asymptotes}\label{apdx:mma}

In this appendix, we describe implementation details of the method of moving asymptotes (MMA) \cite{svanberg1987method} used in the comparative study of Section~\ref{sec:result:circle}.
Let $\boldsymbol{\psi}\in\mathbb{R}^{N_\psi}$ denote the flattened vector of trainable parameters of the SIREN that continuously represents the thermal expansion coefficient $\alpha$ over the reference configuration.
The unconstrained design problem associated with this comparative study reads
\begin{equation}\label{eq:mma_design_problem}
    \underset{\boldsymbol{\psi}\in[\boldsymbol{\psi}_{\mathrm{min}},\,\boldsymbol{\psi}_{\mathrm{max}}]}{\mathrm{argmin}}
    \;\mathcal{J}(\boldsymbol{u}(\boldsymbol{\psi})),
    \qquad
    \text{subject to}\qquad
    \boldsymbol{R}(\boldsymbol{u};\boldsymbol{\psi})=\boldsymbol{0},
\end{equation}
where the objective function $\mathcal{J}$ is chosen as the Chamfer distance defined in \cref{eq:chamfer_distance}, and $[\boldsymbol{\psi}_{\mathrm{min}},\boldsymbol{\psi}_{\mathrm{max}}]$ denotes component-wise box bounds on the SIREN parameters.
In the present study these bounds are set to $\boldsymbol{\psi}_{\mathrm{min}}=-\mathbf{1}$ and $\boldsymbol{\psi}_{\mathrm{max}}=+\mathbf{1}$, where the obtained weights of the SIREN model in \cref{fig:pr4_circle_a} are found to be bounded in $[-1,1]$.
This selection is an empirical choice motivated by the observation that the parameters of the SIREN in the figure resides within this interval, and we regard it as a reasonable working domain under the premise that the MMA, if suitable for training SIREN, should yield a competitive result within a domain consistent with the baseline solution.

The design gradient $\nabla_{\boldsymbol{\psi}} \mathscr{L}$ is obtained from the same adjoint sensitivity analysis scheme as Section~\ref{sec:formulations:adjoint} combined with the chain rule through the SIREN.
MMA iteratively replaces \cref{eq:mma_design_problem} with a sequence of separable, strictly convex subproblems constructed from the current iterate and its gradient.
It is worth noting that MMA has the same optimizer-level input and output structure as conventional machine learning optimizers such as Adam, in the sense that it updates the trainable parameters using objective and gradient information; yet MMA is specifically designed for nonlinear optimization problems with multiple inequality constraints.
The remainder of this section presents the convex subproblem and its analytic solution, followed by the asymptote update, the move limit assignment, and the summary of hyperparameter settings.
The MMA-based parameter update is illustrated in Algorithm~\ref{alg:mma_step}.

\paragraph{Separable Convex Subproblem}
At iteration $\kappa$, let
$\boldsymbol{\psi}^{(\kappa)}\in\mathbb{R}^{N_\psi}$ denote the current iterate,
$\boldsymbol{\Gamma}^{(\kappa)}:=\nabla_{\boldsymbol{\psi}} \mathscr{L}|_{\boldsymbol{\psi}^{(\kappa)}}$
the gradient of the objective, and
$\boldsymbol{\psi}^{-,(\kappa)},\boldsymbol{\psi}^{+,(\kappa)}\in\mathbb{R}^{N_\psi}$ the lower and
upper moving asymptotes satisfying
$\psi_j^{-,(\kappa)}<\psi_j^{(\kappa)}<\psi_j^{+,(\kappa)}$ component-wise.
Define, for each coordinate $j=1,\dots,N_\psi$,
\begin{align}
    \mathcal{P}_j^{(\kappa)}
    &=
    \bigl(\psi_j^{+,(\kappa)}-\psi_j^{(\kappa)}\bigr)^2
    \mathrm{max}(\Gamma_j^{(\kappa)},\,0),
    \label{eq:mma_pj}\\
    \mathcal{Q}_j^{(\kappa)}
    &=
    (\psi_j^{(\kappa)}-\psi_j^{-,(\kappa)})^2
    \mathrm{max}(-\Gamma_j^{(\kappa)},0).
    \label{eq:mma_qj}
\end{align}
The solution $\psi_j^\star$ of the subproblem at iteration $\kappa$ is then
\begin{equation}\label{eq:mma_subproblem}
    \underset{\boldsymbol{\psi}}{\mathrm{argmin}}
    \sum_{j=1}^{N_\psi}
    \left[
        \frac{\mathcal{P}_j^{(\kappa)}}{\psi_j^{+,(\kappa)}-\psi_j}
        +
        \frac{\mathcal{Q}_j^{(\kappa)}}{\psi_j-\psi_j^{-,(\kappa)}}
    \right]
    \quad
    \text{subject to}
    \quad
    \psi_j\in[\xi_j^{-,(\kappa)},\,\xi_j^{+,(\kappa)}],
\end{equation}
where the per-coordinate move limits $[\xi_j^{-,(\kappa)},\xi_j^{+,(\kappa)}]$ are defined
in \cref{eq:mma_move_limits} below.
Since no inequality constraints are imposed beyond the box bounds, the
subproblem in \cref{eq:mma_subproblem} is separable across the
coordinates of $\boldsymbol{\psi}$.
Its unique minimizer admits the closed form
\begin{equation}\label{eq:mma_primal_closed_form}
    \psi_j^{\star}
    =
    \frac{
        \sqrt{\mathcal{P}_j^{(\kappa)}+\epsilon}\,\psi_j^{-,(\kappa)}
        +
        \sqrt{\mathcal{Q}_j^{(\kappa)}+\epsilon}\,\psi_j^{+,(\kappa)}
    }{
        \sqrt{\mathcal{P}_j^{(\kappa)}+\epsilon}
        +
        \sqrt{\mathcal{Q}_j^{(\kappa)}+\epsilon}
    },
\end{equation}
where $\epsilon=10^{-12}$ is a numerical regularization that prevents
division by zero when both $\mathcal{P}_j^{(\kappa)}$ and $\mathcal{Q}_j^{(\kappa)}$ vanish.
The next iterate $\boldsymbol{\psi}^{(\kappa+1)}$ is obtained by clipping
$\boldsymbol{\psi}^{\star}$ to the move-limit box
$[\boldsymbol{\xi}^{-,(\kappa)},\boldsymbol{\xi}^{+,(\kappa)}]$ and subsequently to the
global box $[\boldsymbol{\psi}_{\mathrm{min}},\boldsymbol{\psi}_{\mathrm{max}}]$.

\paragraph{Asymptote Update}
The asymptotes $\boldsymbol{\psi}^{-,(\kappa)}$ and $\boldsymbol{\psi}^{+,(\kappa)}$ are updated
using the standard oscillation-aware heuristic proposed
in~\cite{svanberg1987method}.
Let $\Delta\boldsymbol{\psi}:=\boldsymbol{\psi}_{\mathrm{max}}-\boldsymbol{\psi}_{\mathrm{min}}$
denote the component-wise design span.
For the first two iterations ($\kappa=0,1$), at which iteration history is not
yet available, the asymptotes are initialized as
\begin{equation}\label{eq:mma_asymptote_init}
    \psi_j^{-,(\kappa)}=\psi_j^{(\kappa)}-\zeta_0\,\Delta\psi_j,
    \qquad
    \psi_j^{+,(\kappa)}=\psi_j^{(\kappa)}+\zeta_0\,\Delta\psi_j,
    \qquad
    \zeta_0=0.5.
\end{equation}
For $\kappa\ge 2$, the asymptote distances are contracted when the sign of the
coordinate increment reverses between consecutive iterations, and expanded
otherwise, according to
\begin{equation}\label{eq:mma_sj}
    \chi_j^{(\kappa)}
    =
    \begin{cases}
        \chi^{+}=1.2, & \bigl(\psi_j^{(\kappa)}-\psi_j^{(\kappa-1)}\bigr)\bigl(\psi_j^{(\kappa-1)}-\psi_j^{(\kappa-2)}\bigr)>0,\\[2pt]
        \chi^{-}=0.7, & \bigl(\psi_j^{(\kappa)}-\psi_j^{(\kappa-1)}\bigr)\bigl(\psi_j^{(\kappa-1)}-\psi_j^{(\kappa-2)}\bigr)<0,\\[2pt]
        1, & \text{otherwise},
    \end{cases}
\end{equation}
and the asymptotes are updated by
\begin{equation}\label{eq:mma_asymptote_update}
    \psi_j^{-,(\kappa)}=\psi_j^{(\kappa)}-\chi_j^{(\kappa)}\bigl(\psi_j^{(\kappa-1)}-\psi_j^{-,(\kappa-1)}\bigr),
    \qquad
    \psi_j^{+,(\kappa)}=\psi_j^{(\kappa)}+\chi_j^{(\kappa)}\bigl(\psi_j^{+,(\kappa-1)}-\psi_j^{(\kappa-1)}\bigr).
\end{equation}
To prevent the asymptotes from collapsing onto the current iterate under
repeated contractions, which would reduce the subproblem in
\cref{eq:mma_subproblem} to a vanishing move-limit step, a minimum asymptote gap of $\zeta_{\mathrm{gap}}\,\Delta\psi_j$ is enforced by the post-processing
\begin{equation}\label{eq:mma_asymptote_clamp}
    \psi_j^{-,(\kappa)}\gets\min\!\bigl(\psi_j^{-,(\kappa)},\,\psi_j^{(\kappa)}-\zeta_{\mathrm{gap}}\Delta\psi_j\bigr),
    \qquad
    \psi_j^{+,(\kappa)}\gets\max\!\bigl(\psi_j^{+,(\kappa)},\,\psi_j^{(\kappa)}+\zeta_{\mathrm{gap}}\Delta\psi_j\bigr).
\end{equation}

\paragraph{Move Limits}
Move limits restrict each component of the subproblem solution to lie strictly inside the current asymptote interval, thereby preventing the denominators in \cref{eq:mma_subproblem} from vanishing.
The per-coordinate move-limit bounds are defined as
\begin{equation}\label{eq:mma_move_limits}
    \xi_j^{-,(\kappa)}
    =
    0.9\psi_j^{-,(\kappa)}+0.1\psi_j^{(\kappa)},
    \qquad
    \xi_j^{+,(\kappa)}
    =
    0.9\psi_j^{+,(\kappa)}+0.1\psi_j^{(\kappa)},
\end{equation}
subsequently intersected with the global box $[\boldsymbol{\psi}_{\mathrm{min}},\boldsymbol{\psi}_{\mathrm{max}}]$.

\paragraph{Hyperparameter Summary}
The hyperparameters used in the MMA implementation adopted for the comparative study in Section~\ref{sec:result:circle} are set to $\zeta_0=0.5$, $\chi^{+}=1.2$, $\chi^{-}=0.7$, and $\zeta_{\mathrm{gap}}=10^{-2}$.
These values are applied identically over the $300$ iterations reported in \cref{fig:pr4_circle_mma_b}.

\begin{algorithm}[htbp]
\caption{One update step of MMA.}
\label{alg:mma_step}
\begin{algorithmic}[1]
\Require Current iterate $\boldsymbol{\psi}^{(\kappa)}$, past iterates $\boldsymbol{\psi}^{(\kappa-1)},\boldsymbol{\psi}^{(\kappa-2)}$, gradient $\boldsymbol{\Gamma}^{(\kappa)}=\nabla_{\boldsymbol{\psi}} \mathscr{L}|_{\boldsymbol{\psi}^{(\kappa)}}$, past asymptotes $\boldsymbol{\psi}^{-,(\kappa-1)},\boldsymbol{\psi}^{+,(\kappa-1)}$, box bounds $\boldsymbol{\psi}_{\mathrm{min}},\boldsymbol{\psi}_{\mathrm{max}}$, hyperparameters $\zeta_0,\chi^{+},\chi^{-},\zeta_{\mathrm{gap}}$
\Ensure Updated iterate $\boldsymbol{\psi}^{(\kappa+1)}$, updated asymptotes $\boldsymbol{\psi}^{-,(\kappa)},\boldsymbol{\psi}^{+,(\kappa)}$
\If{$\kappa<2$}
    \State Initialize $\boldsymbol{\psi}^{-,(\kappa)},\,\boldsymbol{\psi}^{+,(\kappa)}$ \Comment{\cref{eq:mma_asymptote_init}}
\Else
    \State Update $\boldsymbol{\chi}^{(\kappa)}$ \Comment{\cref{eq:mma_sj}}
    \State Update $\boldsymbol{\psi}^{-,(\kappa)},\,\boldsymbol{\psi}^{+,(\kappa)}$ \Comment{\cref{eq:mma_asymptote_update}}
\EndIf
\State Enforce the minimum asymptote gap \Comment{\cref{eq:mma_asymptote_clamp}}
\State Compute $\boldsymbol{\mathcal{P}}^{(\kappa)}$ \Comment{\cref{eq:mma_pj}}
\State Compute $\boldsymbol{\mathcal{Q}}^{(\kappa)}$ \Comment{\cref{eq:mma_qj}}
\State Set the move limit bounds $\boldsymbol{\xi}^{-,(\kappa)},\,\boldsymbol{\xi}^{+,(\kappa)}$ \Comment{\cref{eq:mma_move_limits}}
\State Compute the subproblem minimizer $\boldsymbol{\psi}^{\star}$ \Comment{\cref{eq:mma_primal_closed_form}}
\State Project $\boldsymbol{\psi}^{\star}$ onto $[\boldsymbol{\xi}^{-,(\kappa)},\boldsymbol{\xi}^{+,(\kappa)}]\cap[\boldsymbol{\psi}_{\mathrm{min}},\boldsymbol{\psi}_{\mathrm{max}}]$ to obtain $\boldsymbol{\psi}^{(\kappa+1)}$
\State \Return $\boldsymbol{\psi}^{(\kappa+1)},\,\boldsymbol{\psi}^{-,(\kappa)},\,\boldsymbol{\psi}^{+,(\kappa)}$
\end{algorithmic}
\end{algorithm}

\section{Effect of Noise on Data}\label{apdx:noise}

In this section, we investigate the effect of noise in the training data on the learned constitutive prior and the subsequent inverse-design results. Controlled noise is introduced into the training data, and the errors of the trained models and the corresponding inverse-design results are evaluated over multiple trials. Five independent end-to-end trials are conducted at each noise level, where each trial comprises data perturbation, where applicable, model training, and inverse design. Thus, a total of 15 trials are performed. Unless otherwise stated, the problem setting follows Section~\ref{sec:result:bar}.

The effect of noise is investigated at three distinct levels by adding sampled perturbations to the clean data. The perturbations are realized as pointwise independent and identically distributed Gaussian random variables sampled from $\mathcal{N}(0,\sigma^2)$, with $\sigma=0$, $\sigma=0.05$, and $\sigma=0.10$. these sampled noises are then multiplied with the reference stress value $P_0$.
The case $\sigma=0$ corresponds to the clean data. The underlying clean data are generated using the neo-Hookean model in \cref{eq:neohookean} with $E=1\;\mathrm{MPa}$, $E=2\;\mathrm{MPa}$, and $E=3\;\mathrm{MPa}$. For each material response, 20 equispaced data points are generated over the stretch interval $\lambda\in[0.5,3]$.

\begin{figure}[htbp]
    \centering
    \begin{subfigure}{0.33\textwidth}
        \centering
        \includegraphics[width=\textwidth]{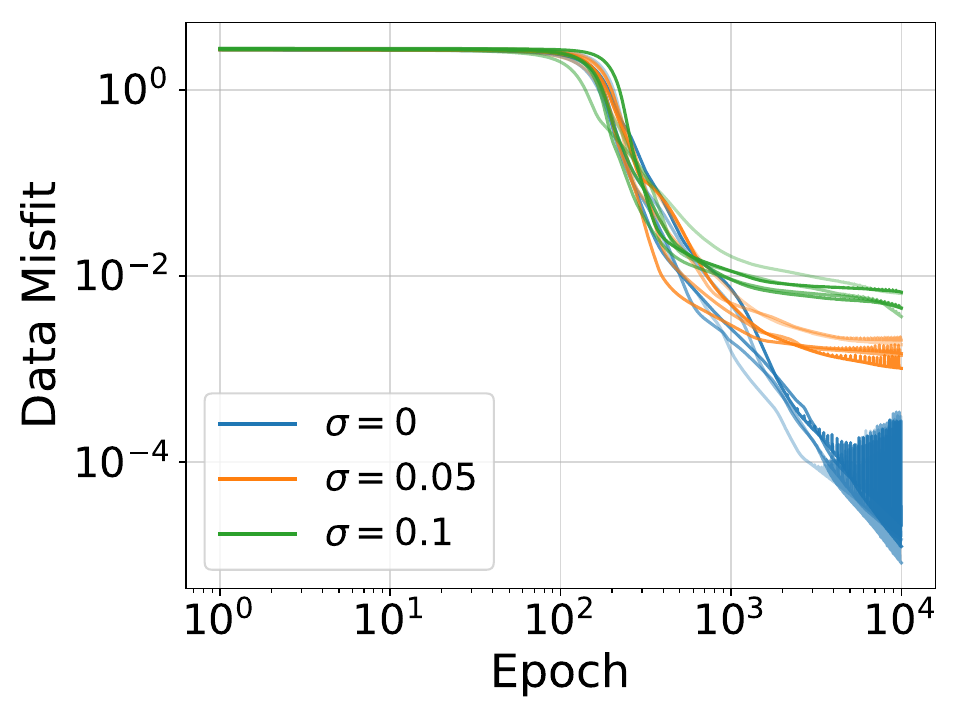}
    \end{subfigure}
    \hspace{0.05\textwidth}
    \begin{subfigure}{0.33\textwidth}
        \centering
        \includegraphics[width=\textwidth]{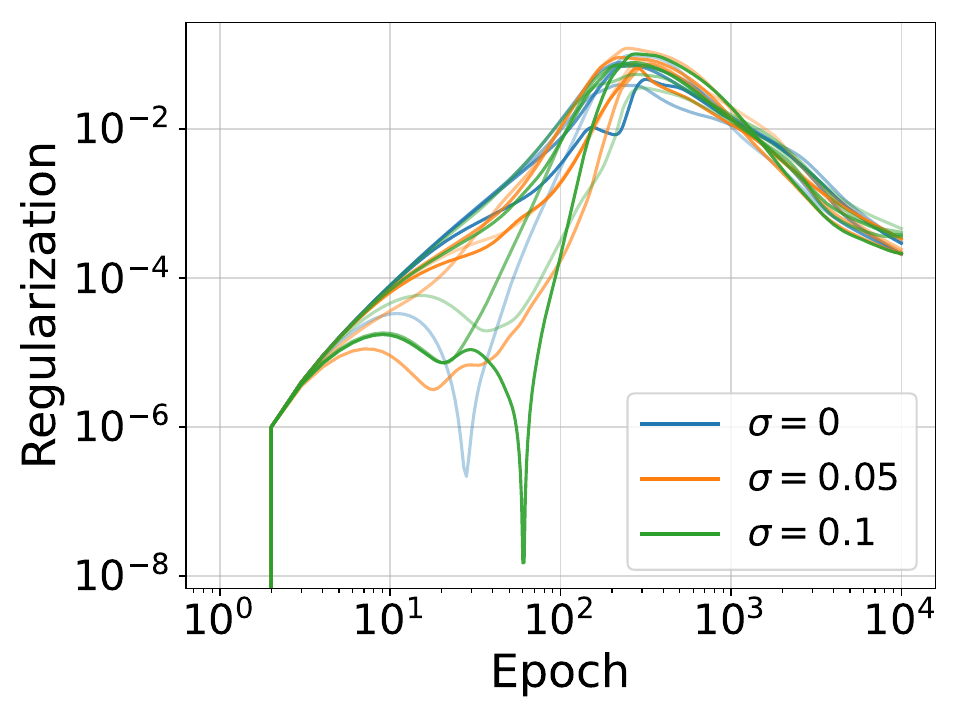}
    \end{subfigure}
    \caption{
    Loss histories over multiple trials of model training on the perturbed datasets.
    A total of 15 histories are shown, with five trials for each noise level. 
    (left) Data-misfit term and (right) regularization term are shown.
    }
    \label{fig:apdx_noise_loss}
\end{figure}

For each realized dataset, a model is independently trained from a different random initialization. Each model is trained for a fixed number of epochs using the objective function in \cref{eq:surrogate_loss}. The resulting loss histories are shown in \cref{fig:apdx_noise_loss}. The left panel shows a clear separation between the models trained on the clean and perturbed data. In particular, the final data-misfit term generally increases with $\sigma$, indicating that the models achieve a poorer fit to the perturbed observations than to the clean data.

This behavior is further illustrated in \cref{fig:apdx_noise_pred}. Because the pointwise perturbations may introduce locally non-monotone responses, the perturbed data need not be consistent with an energy function that is convex in $\lambda$, stress-free, and energy-free. In contrast, the constitutive model adopted in the present work enforces these considerations by construction. The trained model therefore provides a relatively smooth and mechanically admissible approximation rather than reproducing all perturbations in the data, resulting in a larger data-misfit term as the noise level increases.

\begin{figure}[htbp]
    \centering
    \begin{subfigure}{0.3\textwidth}
        \centering
        \includegraphics[width=\textwidth]{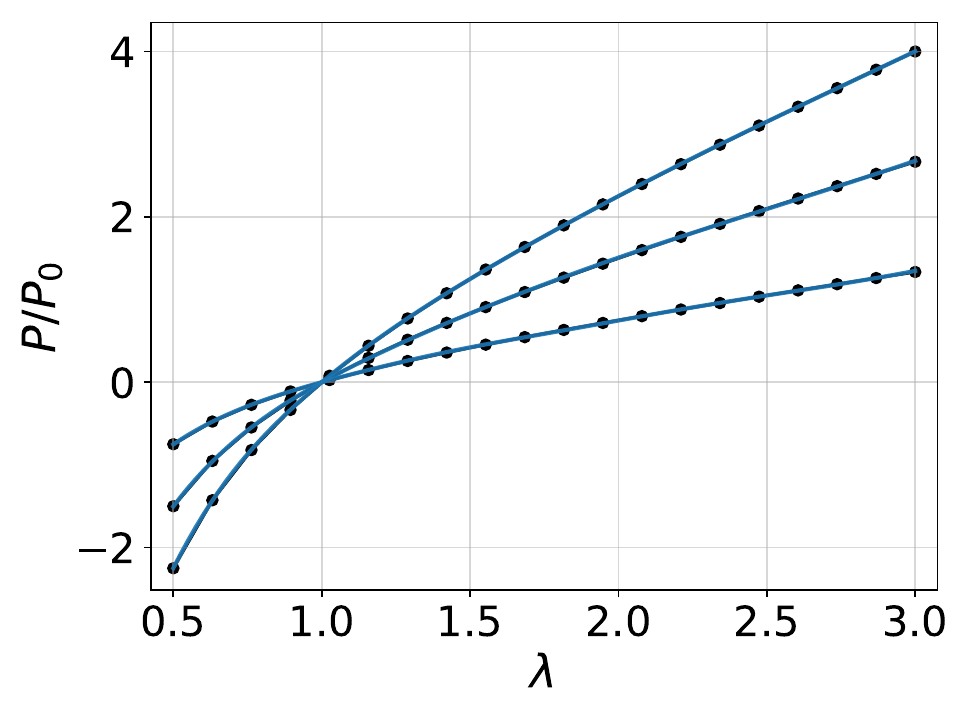}
        \caption{}
        \label{fig:apdx_noise_pred_a}
    \end{subfigure}
    \begin{subfigure}{0.3\textwidth}
        \centering
        \includegraphics[width=\textwidth]{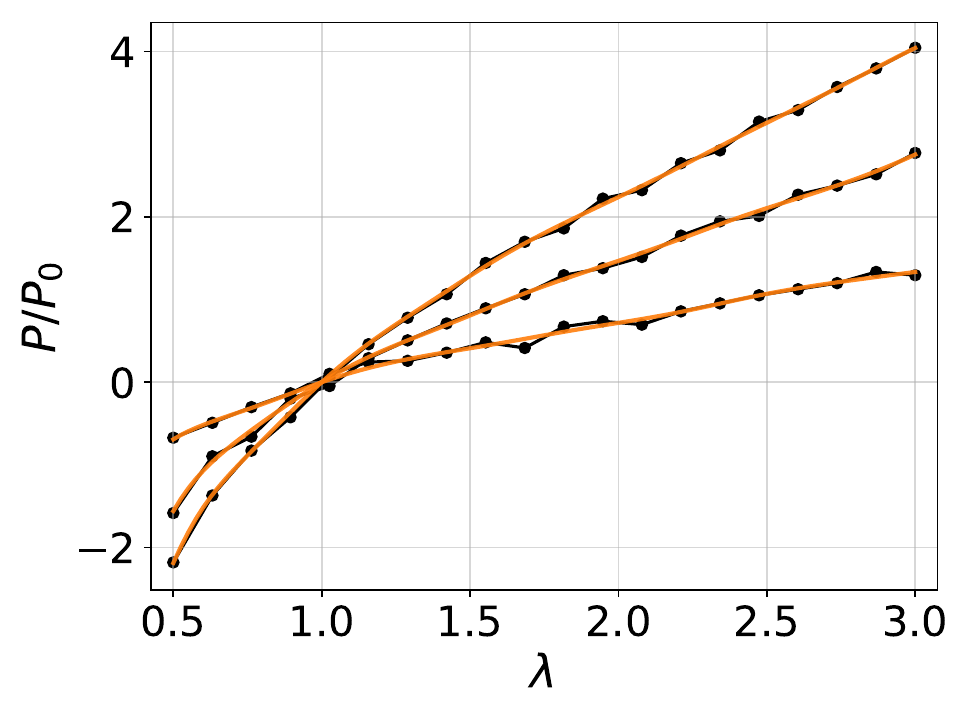}
        \caption{}
        \label{fig:apdx_noise_pred_b}
    \end{subfigure}
    \begin{subfigure}{0.3\textwidth}
        \centering
        \includegraphics[width=\textwidth]{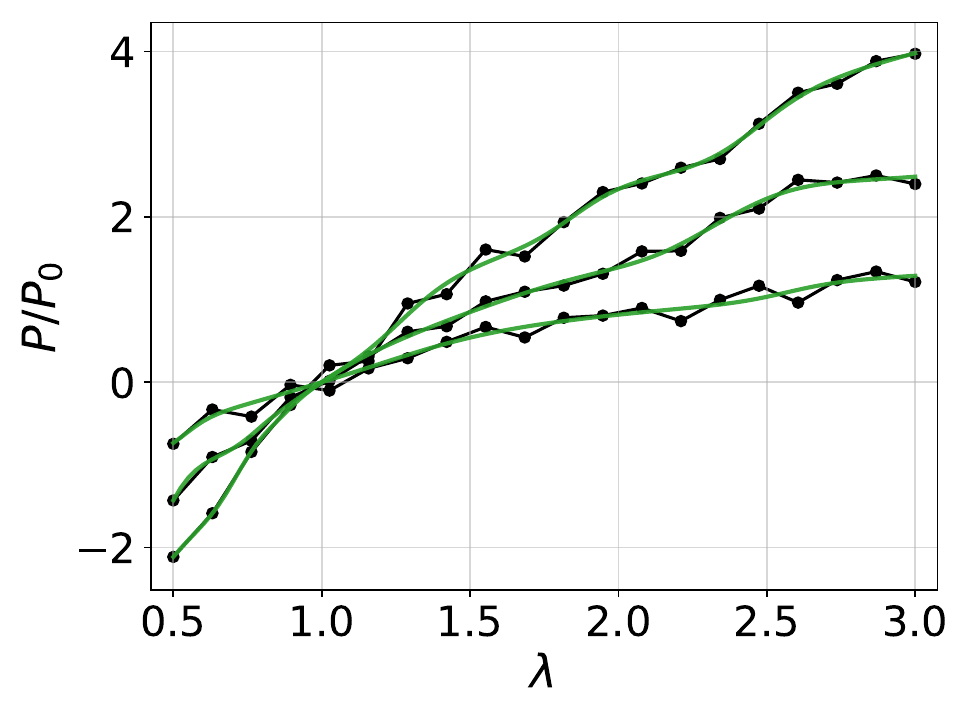}
        \caption{}
        \label{fig:apdx_noise_pred_c}
    \end{subfigure}
    \caption{
    Representative trained models obtained at different noise levels: (\subref{fig:apdx_noise_pred_a}) clean data, (\subref{fig:apdx_noise_pred_b}) $\sigma=0.05$, and (\subref{fig:apdx_noise_pred_c}) $\sigma=0.10$.
    The training data used for each model are depicted as black markers connected by solid lines.
    }
    \label{fig:apdx_noise_pred}
\end{figure}

The right panel of \cref{fig:apdx_noise_loss} shows that the regularization terms remain comparable across the three noise levels. This indicates that the magnitudes of the learned latent variables may remain similar under the considered perturbations. The results therefore suggest that the imposed mechanical constraints primarily limit the model's ability to reproduce noise-induced irregularities in the observed responses.

\begin{figure}[htbp]
    \centering
    \begin{subfigure}{0.4\textwidth}
        \centering
        \includegraphics[width=\textwidth]{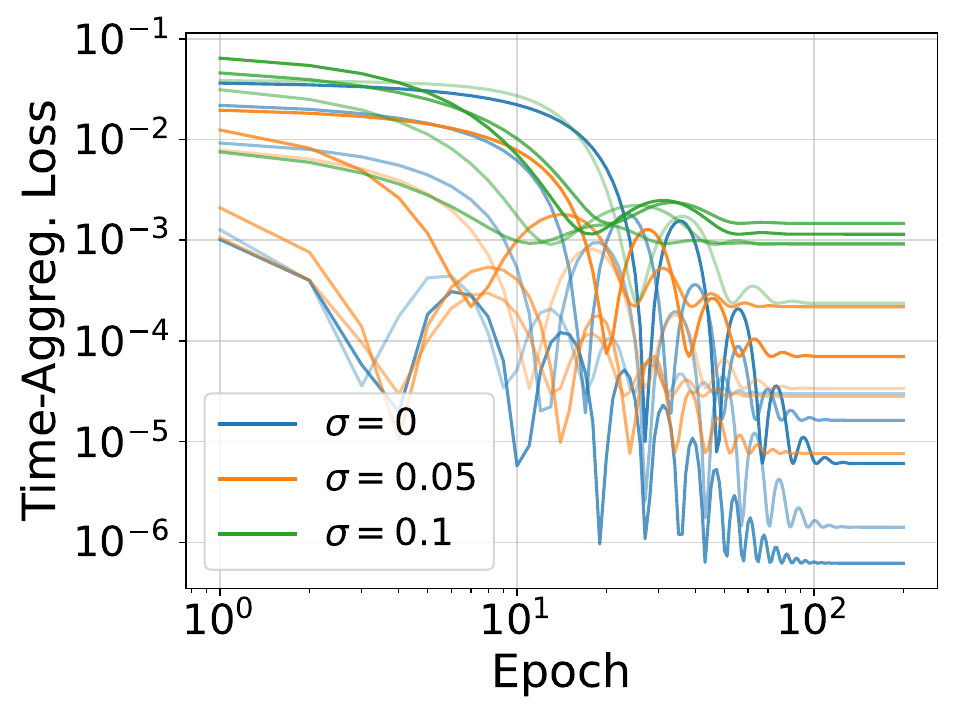}
    \end{subfigure}
    \caption{
    Loss histories of the multiple trials of the considered inverse problem.
    }
    \label{fig:apdx_noise_inverse}
\end{figure}

For each trained model, inverse-design trials are conducted using different initial guesses. The initial guesses are sampled uniformly from the learned admissible latent range $\mathcal{Z}$ from training of each model. 
Here, we set ground truth to be the trajectory of material response with $E=1.5\;\mathrm{MPa}$, with external force $f=1\;\mathrm{N}$ and number of load step $20$.
The maximum number of inverse-design iterations is fixed at 200 for all trials.
The learning rate is set to $1\times10^{-3}$, which performs well in this setting.

\begin{figure}[htbp]
    \centering
    \begin{subfigure}{0.3\textwidth}
        \centering
        \includegraphics[width=\textwidth]{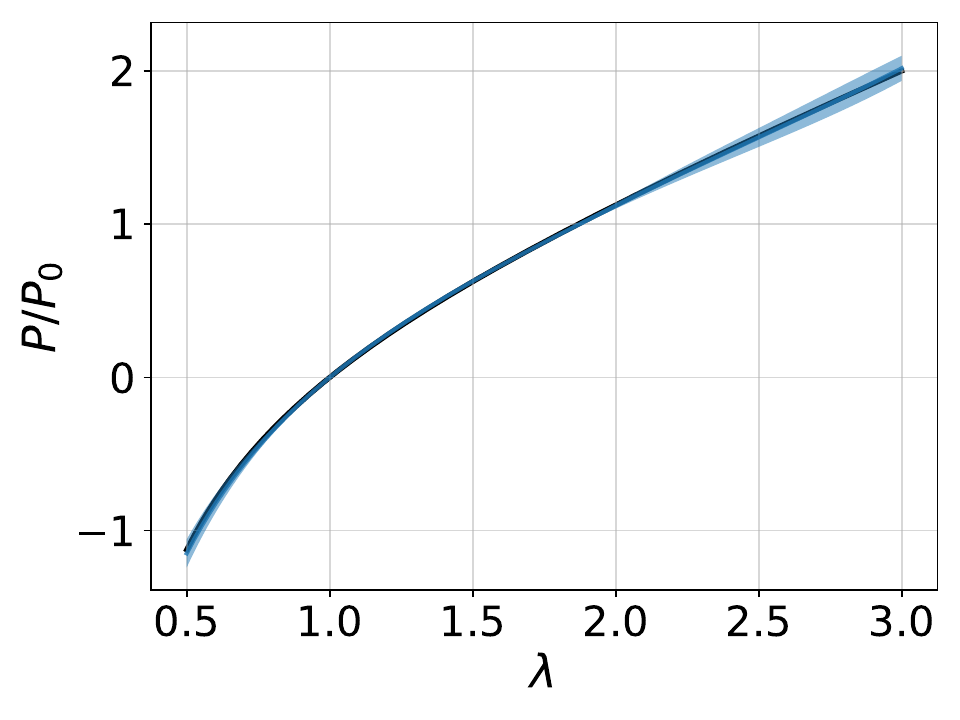}
        \caption{}
        \label{fig:apdx_noise_recover_a}
    \end{subfigure}
    \begin{subfigure}{0.3\textwidth}
        \centering
        \includegraphics[width=\textwidth]{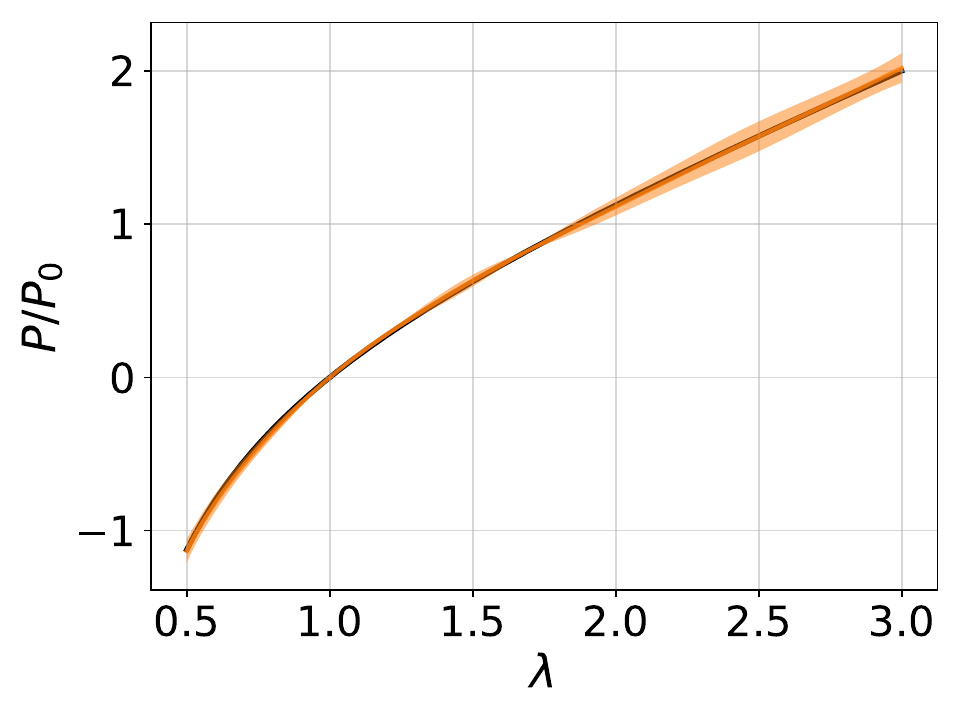}
        \caption{}
        \label{fig:apdx_noise_recover_b}
    \end{subfigure}
    \begin{subfigure}{0.3\textwidth}
        \centering
        \includegraphics[width=\textwidth]{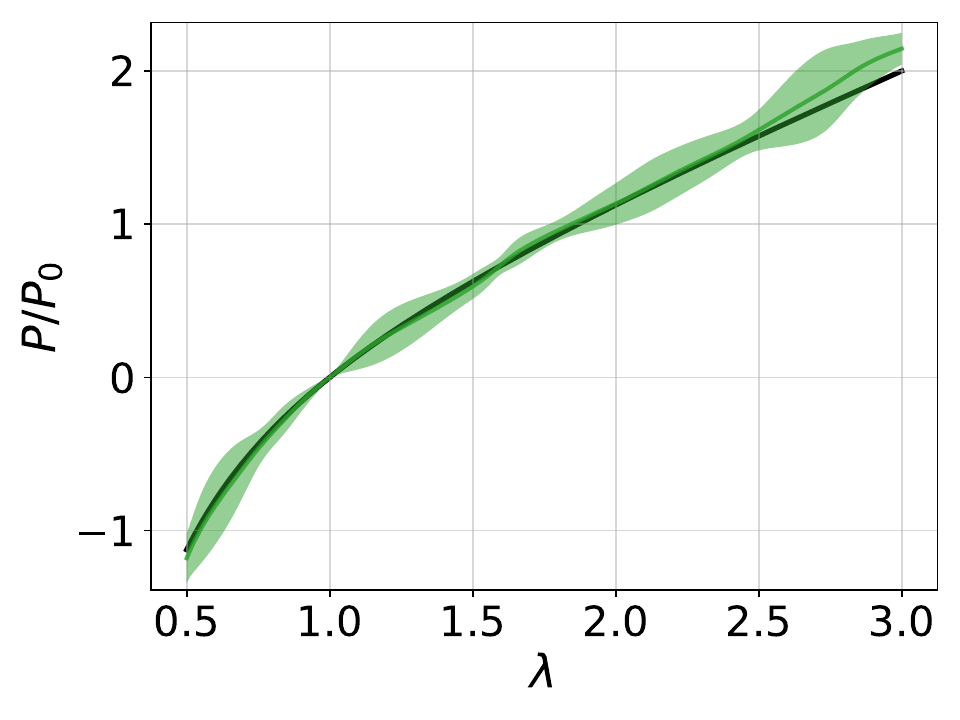}
        \caption{}
        \label{fig:apdx_noise_recover_c}
    \end{subfigure}
    \caption{
    Pointwise empirical mean and variability of the constitutive responses recovered from the inverse-identification trials at three noise levels: (\subref{fig:apdx_noise_recover_a}) clean data, (\subref{fig:apdx_noise_recover_b}) $\sigma=0.05$, and (\subref{fig:apdx_noise_recover_c}) $\sigma=0.10$. The ground-truth response is depicted by a solid black line. The sample mean of the recovered responses is shown together with overlapping translucent bands representing $\pm2$ sample standard deviations.
    }
    \label{fig:apdx_noise_recover}
\end{figure}

The corresponding inverse-identification results are presented in \cref{fig:apdx_noise_inverse}. The objective values decrease during the initial iterations and approach a plateau after approximately 100 iterations, with only minor subsequent changes. The results also show a general tendency for the models trained on more strongly perturbed data to attain higher final objective values.

The constitutive responses recovered from the inverse-identification trials are shown in \cref{fig:apdx_noise_recover}. Consistent with the objective values in \cref{fig:apdx_noise_inverse}, the models trained on the clean data recover responses close to the ground truth with relatively small trial-to-trial variability. As the noise level increases, the recovered responses exhibit larger deviations and greater variability.
Nevertheless, \cref{fig:apdx_noise_recover_c} shows that the responses recovered under the highest noise level remain mechanically admissible and that their average response does not deviate substantially from the ground truth, despite the comparatively large final objective values reported in \cref{fig:apdx_noise_inverse}.

Collectively, these results demonstrate that the proposed approach remains applicable when the constitutive prior is learned from perturbed training data. Although the accuracy and consistency of the recovered responses deteriorate as $\sigma$ increases, the inferred responses remain within a mechanically admissible range and retain the overall characteristics of the ground-truth response under the noise levels considered.

\section{Computational Setup and Cost}\label{apdx:compute_cost}

In this section, we summarize the computational setup and cost for the numerical examples presented in Section~\ref{sec:result}.
All computations are performed on an Apple M5 Pro chip with 15 CPU cores (5 Super cores and 10 Performance cores) and 48~GB of unified memory.
The implementations are executed using \texttt{Python} and \texttt{PyTorch}~\cite{paszke2019pytorch}, with double-precision floating-point arithmetic (\texttt{float64}) used throughout. 

The computational cost associated with each example is summarized in Table~\ref{tab:compute_cost}.
In the table, each case is identified using the corresponding boldface label adopted in the main text; for example, the case in Section~\ref{sec:result:bar} is labeled \textbf{1D Bar}. For each numerical example, the reported values correspond to the representative results shown in \cref{fig:pr1_trainsient_bar,fig:pr2_airfoil_result,fig:pr3_crack_result,fig:pr4_circle_b,fig:pr3_heart}, respectively. The wall-clock time is reported as the mean and sample standard deviation over five trials. The system degrees of freedom used in each forward solve and the number of design variables, corresponding to the number of elements in each system, are also reported.

\begin{table*}[htbp]
\centering
\caption{Summary of the computational cost of each problem in Section~\ref{sec:result}.}
\label{tab:compute_cost}
\begin{tabular}{l c c c c c}
\hline
\textbf{Setting} & \textbf{1D Bar} & \textbf{Airfoil} & \textbf{Fracture} & \textbf{Circle} & \textbf{Heart} \\
\hline
Wall-clock time [s] & $2.227\pm0.015$ & $1.813\pm0.019$ & $5.323\pm0.030$ & $8.316\pm0.025$ & $69.912\pm0.709$ \\ 
Degrees of freedom & 1 & 56 & 278 & 278 & 278 \\ 
Number of design variables & 1 & 65 & 383 & 383 & 383 \\
Number of epochs & 27 & 29 & 48 & 72 & 76 \\
\hline
\end{tabular}
\end{table*}

\section{Energy- and Stress-Free Normalization and Convexity Preservation}
\label{apdx:stress_free}

In this section, we show that the construction introduced in
\cref{eq:stress_free} yields an energy-free and stress-free reference configuration. We further show that the normalization preserves convexity in $\boldsymbol{i}$.

In the present setting, it suffices to show that the Biot stress $\boldsymbol{T}_{\boldsymbol{\phi}}=\partial\mathcal{M}_{\boldsymbol{\phi}}/\partial\boldsymbol{U}$ vanishes at the reference configuration, i.e., $\boldsymbol{T}_{\boldsymbol{\phi}}(\boldsymbol{I};\boldsymbol{z})=\boldsymbol{0}$.

The derivatives of the principal invariants of $\boldsymbol{U}$ are~\cite{truesdell2004non}
\begin{align}
\frac{\partial i_1}{\partial\boldsymbol{U}}
=
\boldsymbol{I},
\qquad
\frac{\partial i_2}{\partial\boldsymbol{U}}
=
i_1\boldsymbol{I}-\boldsymbol{U},
\qquad
\frac{\partial i_3}{\partial\boldsymbol{U}}
=
i_3\boldsymbol{U}^{-1}.
\end{align}
Taking the derivative of \cref{eq:stress_free} with respect to
$\boldsymbol{U}$ gives
\begin{equation}
\boldsymbol{T}_{\boldsymbol{\phi}}
=
\frac{
\partial g_{\boldsymbol{\phi}}
}{
\partial i_k
}
\frac{
\partial i_k
}{
\partial\boldsymbol{U}
}
-
c(\boldsymbol{z})
\frac{
\partial i_3
}{
\partial\boldsymbol{U}
}, \quad \text{summation over $k=1,2,3$}.
\end{equation}

At the reference configuration, $\boldsymbol{U}=\boldsymbol{I}$, and therefore
\begin{align}
\boldsymbol{T}_{\boldsymbol{\phi}}
(\boldsymbol{I};\boldsymbol{z})
=
\left.
\left(
\frac{\partial g_{\boldsymbol{\phi}}}{\partial i_1}
+
2
\frac{\partial g_{\boldsymbol{\phi}}}{\partial i_2}
+
\frac{\partial g_{\boldsymbol{\phi}}}{\partial i_3}
-
c(\boldsymbol{z})
\right)
\right|_{\boldsymbol{i}
=
\boldsymbol{i}_0}
\boldsymbol{I}
\
=
\boldsymbol{0},
\end{align}
where the last equality follows directly from the definition of
$c(\boldsymbol{z})$ in \cref{eq:stress_free_fact}. 

The energy-free condition follows directly from
\begin{equation}
\mathcal{M}_{\boldsymbol{\phi}}
(\boldsymbol{i}_0;\boldsymbol{z})
=
g_{\boldsymbol{\phi}}
(\boldsymbol{i}_0;\boldsymbol{z})
-
g_{\boldsymbol{\phi}}
(\boldsymbol{i}_0;\boldsymbol{z})
-
c(\boldsymbol{z})(1-1)
=
0.
\end{equation}

\begin{lemma}[Affine perturbations preserve convexity]
\label{lemma:affine_convexity}
Let $g:\mathbb{R}^m\to\mathbb{R}$ be convex, and let
$a(\boldsymbol{i})
=\boldsymbol{b}^{\mathsf T}\boldsymbol{i}+d$
be affine. Then the function
\[
\widetilde{g}(\boldsymbol{i})
=
g(\boldsymbol{i})
+
a(\boldsymbol{i})
\]
is convex in $\boldsymbol{i}$.
\end{lemma}

\begin{proof}
For any $\boldsymbol{i}_1,\boldsymbol{i}_2\in\mathbb{R}^m$
and $\alpha\in[0,1]$, convexity of $g$ and affinity of $a$ give
\begin{align}
\widetilde{g}
\bigl(
\alpha\boldsymbol{i}_1
+
(1-\alpha)\boldsymbol{i}_2
\bigr)
&=
g\bigl(
\alpha\boldsymbol{i}_1
+
(1-\alpha)\boldsymbol{i}_2
\bigr)
+
a\bigl(
\alpha\boldsymbol{i}_1
+
(1-\alpha)\boldsymbol{i}_2
\bigr)
\\
&\le
\alpha g(\boldsymbol{i}_1)
+
(1-\alpha)g(\boldsymbol{i}_2)
+
\alpha a(\boldsymbol{i}_1)
+
(1-\alpha)a(\boldsymbol{i}_2)
\\
&=
\alpha\widetilde{g}(\boldsymbol{i}_1)
+
(1-\alpha)\widetilde{g}(\boldsymbol{i}_2).
\end{align}
Hence, $\widetilde{g}$ is convex.
\end{proof}

For each fixed $\boldsymbol{z}$, the normalization in
\cref{eq:stress_free} modifies
$g_{\boldsymbol{\phi}}(\boldsymbol{i};\boldsymbol{z})$
only by an affine function of $\boldsymbol{i}$. Therefore,
by Lemma~\ref{lemma:affine_convexity}, 
$\mathcal{M}_{\boldsymbol{\phi}}(\boldsymbol{i};\boldsymbol{z})$
is convex in $\boldsymbol{i}$ whenever
$g_{\boldsymbol{\phi}}(\boldsymbol{i};\boldsymbol{z})$
is convex in $\boldsymbol{i}$.
\begin{remark}
    Related affine modifications have also been discussed in the literature. In particular, \citet{linden2023neural} consider stress normalization and polyconvexity for hyperelastic models formulated in terms of invariants of $\mathbf{C}$. The discussion here instead concerns the stretch-based formulation in terms of invariants of $\mathbf{U}$, following \citet{steigmann2003isotropic}.
\end{remark}

\end{document}